\documentclass[11pt,onecolumn]{IEEEtran}
\usepackage{amsthm}
\usepackage{times,amssymb,amsmath,amsfonts,float,nicefrac,color,bbm,mathrsfs,float}
\usepackage{enumerate,multirow,caption,tikz,graphicx}
\usepackage{algpseudocode}
\usepackage{stackengine}
\usepackage[noadjust]{cite}
\usepackage{subcaption}
\usepackage{diagbox}
\usepackage[ruled,vlined,linesnumbered]{algorithm2e}

\usepackage[top=1in,bottom=1in,left=1in,right=1in]{geometry}
\allowdisplaybreaks

\newcommand{\RN}[1]{%
	\textup{\expandafter{\romannumeral#1}}%
}

\usetikzlibrary{shapes.misc, positioning}
\usetikzlibrary{decorations.pathreplacing}

\tikzset{
	  block/.style = {draw, thick, rectangle, minimum width = 1em},
	 sblock/.style = {draw, thick, rectangle, minimum height = 6em,   minimum width = 6em},
	lgblock/.style = {draw, thick, rectangle, minimum height = 9em,   minimum width = 6em},
	 vblock/.style = {draw, thick, rectangle, minimum height = 3.7em, minimum width = 1.8em},
}

\tikzset{XOR/.style={draw,circle,append after command={
			[shorten >=\pgflinewidth, shorten <=\pgflinewidth,]
			(\tikzlastnode.north) edge (\tikzlastnode.south)
			(\tikzlastnode.east) edge (\tikzlastnode.west)
		}
	}
}

\newcommand\remove[1]{}

\allowdisplaybreaks

\newtheorem{proposition}{Proposition}

\newtheorem{lemma}{Lemma}

\newtheorem{remark}{Remark}

\newtheorem{cnstr}{Construction}

\def\mathbi#1{{\textbf{\textit #1}}}

\newcommand{\bP}{\mathbb{P}}

\newcommand{\cA}{\mathcal{A}}

\newcommand{\cI}{\mathcal{I}}

\newcommand{\cS}{\mathcal{S}}

\newcommand{\cY}{\mathcal{Y}}

\DeclareMathOperator{\polar}{polar}
\DeclareMathOperator{\RM}{RM}
\DeclareMathOperator{\ABS}{ABS}

\DeclareMathOperator{\dB}{dB}

\DeclareMathOperator{\argmax}{argmax}

\DeclareMathOperator{\first}{first}
\DeclareMathOperator{\second}{second}

\newcommand{\mG}{\mathbf{G}}
\newcommand{\mP}{\mathbf{P}}

\newcommand{\mI}{\mathbf{I}}

\newcommand{\tD}{\mathtt{D}}
\newcommand{\tB}{\mathtt{B}}
\newcommand{\numD}{\mathtt{N_D}}
\newcommand{\numB}{\mathtt{N_B}}
\newcommand{\tP}{\mathtt{P}}
\newcommand{\tR}{\mathtt{R}}
\newcommand{\tH}{\mathtt{H}}
\newcommand{\barP}{\mathtt{\bar{P}}}
\newcommand{\barR}{\mathtt{\bar{R}}}
\newcommand{\barH}{\mathtt{\bar{H}}}
\newcommand{\score}{\mathtt{score}}

\DeclareMathOperator{\prob}{prob}
\DeclareMathOperator{\temppointer}{temp}

\makeatletter
\DeclareRobustCommand*\triangledowndot{%
	\ifmmode
	{\mathpalette\triangleda@t\triangleda@@t}%
	\else
	\triangleda@t\@empty\triangleda@@t
	\fi
}
\def\triangleda@t#1#2{#2{#1}}
\def\triangleda@@t#1{%
	\setbox0=\hbox{$\m@th#1\triangledown$}
	\ooalign{\hbox to \wd0{\hss$\m@th#1 {\cdot}$\hss}\crcr\box0\crcr}
}
\makeatother

\makeatletter
\DeclareRobustCommand*\lozengedot{%
	\ifmmode
	{\mathpalette\triangledb@t\triangledb@@t}%
	\else
	\triangledb@t\@empty\triangledb@@t
	\fi
}
\def\triangledb@t#1#2{#2{#1}}
\def\triangledb@@t#1{%
	\setbox0=\hbox{$\m@th#1\lozenge$}
	\ooalign{\hbox to \wd0{\hss$\m@th#1 \cdot$\hss}\crcr\box0\crcr}
}
\makeatother

\makeatletter
\DeclareRobustCommand*\vartriangledot{%
	\ifmmode
	{\mathpalette\triangledc@t\triangledc@@t}%
	\else
	\triangledc@t\@empty\triangledc@@t
	\fi
}
\def\triangledc@t#1#2{#2{#1}}
\def\triangledc@@t#1{%
	\setbox0=\hbox{$\m@th#1\vartriangle$}
	\ooalign{\hbox to \wd0{\hss$\m@th#1 \cdot$\hss}\crcr\box0\crcr}
}
\makeatother

\newcommand{\oria}{\triangledown}
\newcommand{\orib}{\lozenge}
\newcommand{\oric}{\vartriangle}
\newcommand{\swpa}{\blacktriangledown}
\newcommand{\swpb}{\blacklozenge}
\newcommand{\swpc}{\blacktriangle}

\begin{document}
\title{Adjacent-Bits-Swapped Polar codes: A new code construction to speed up polarization}

\author{Guodong Li \and  \hspace*{.5in} Min Ye  \and  \hspace*{.5in}  Sihuang Hu}

\maketitle
{\renewcommand{\thefootnote}{}
	\footnotetext{
		\vspace{-.2in}

		\noindent\rule{1.5in}{.4pt}

		Research partially funded by National Key R\&D Program of China under Grant No. 2021YFA1001000, National Natural Science Foundation of China under Grant No. 12001322, Shandong Provincial Natural Science Foundation under Grant No. ZR202010220025, and a Taishan scholar program of Shandong Province.
		This paper was presented in part at the 2022 IEEE International Symposium on Information Theory~\cite{Li2022ISIT}.

		Guodong Li is with School of Cyber Science and Technology, Shandong University, Qingdao, Shandong, 266237, China.
		Email: guodongli@mail.sdu.edu.cn

		Min Ye is with Tsinghua-Berkeley Shenzhen Institute, Tsinghua Shenzhen International Graduate School, Shenzhen 518055, China.
		Email: yeemmi@gmail.com

		Sihuang Hu is with  Key Laboratory of Cryptologic Technology and Information Security, Ministry of Education, Shandong University, Qingdao, Shandong, 266237, China and School of Cyber Science and Technology, Shandong University, Qingdao, Shandong, 266237, China.
		Email: husihuang@sdu.edu.cn
	}
}

\renewcommand{\thefootnote}{\arabic{footnote}}
\setcounter{footnote}{0}

\begin{abstract}
	The construction of polar codes with code length $n=2^m$ involves $m$ layers of polar transforms.
	In this paper, we observe that after each layer of polar transforms, one can swap certain pairs of adjacent bits to accelerate the polarization process.
	More precisely, if the previous bit is more reliable than its next bit under the successive decoder, then switching the decoding order of these two adjacent bits will make the reliable bit even more reliable and the noisy bit even noisier.

	Based on this observation, we propose a new family of codes called the Adjacent-Bits-Swapped (ABS) polar codes.
	We add a permutation layer after each polar transform layer in the construction of the ABS polar codes.
	In order to choose which pairs of adjacent bits to swap in the permutation layers, we rely on a new polar transform that combines two independent channels with $4$-ary inputs.
	This new polar transform allows us to track the evolution of every pair of adjacent bits through different layers of polar transforms, and it also plays an essential role in the successive cancellation list (SCL) decoder for the ABS polar codes.
	Extensive simulation results show that ABS polar codes consistently outperform standard polar codes by $0.15\dB$---$0.3\dB$ when we use CRC-aided SCL decoder with list size $32$ for both codes.
	The implementations of all the algorithms in this paper are available at  \texttt{https://github.com/PlumJelly/ABS-Polar}
\end{abstract}

\section{Introduction} \label{sect:introduction}

Polar codes and Reed-Muller (RM) codes are two closely related code families in the sense that their generator matrices are formed of rows from the same square matrix.
Although RM codes were discovered several decades earlier than polar codes, the capacity-achieving property of RM codes was established very recently.
Specifically, polar codes were proposed by Ar{\i}kan in \cite{Arikan09} and were shown to achieve capacity on all binary memoryless symmetric (BMS) channels in the same paper.
In contrast, RM codes were proposed back in the 1950s \cite{Reed54,Muller54}, but the question of whether RM codes achieve capacity remained open for more than 60 years until the recent breakthroughs.
It was shown in \cite{Kudekar17} that RM codes achieve capacity on binary erasure channels (BEC) under the block-MAP decoder.
More recently, Reeves and Pfister proved that RM codes achieve capacity on all BMS channels under the bit-MAP decoder \cite{Reeves21}.

While both code families achieve capacity of BMS channels, simulation results \cite{Mondelli14,Ye20} and theoretical analysis \cite{Hassani14,Hassani18} suggest that RM codes have better finite-length performance than polar codes.
It was conjectured in \cite{AY20} that this is because RM codes polarize even faster than polar codes.
More precisely, in polar coding framework, we multiply a message vector consisting of $n=2^m$ message bits with the matrix $\mathbf{G}_n^{\polar}=(\mathbf{G}_2^{\polar})^{\otimes m}$ and transmit the resulting codeword vector through a BMS channel.
Here $\mathbf{G}_2^{\polar}=\begin{bmatrix}
		1 & 0 \\
		1 & 1
	\end{bmatrix}$, and $\otimes$ is the Kronecker product.
The message bits are divided into information bits and frozen bits according to their reliability under the successive decoder.
This polar coding framework can also be used to analyze RM codes.
To that end, we replace the recursive relation $\mathbf{G}_{n}^{\polar}=\mathbf{G}_{n/2}^{\polar} \otimes \mathbf{G}_2^{\polar}$ in the standard polar code construction with $\mathbf{G}_{n}^{\RM}= \mathbf{P}_n^{\RM} (\mathbf{G}_{n/2}^{\RM} \otimes \mathbf{G}_2^{\polar}).$
Here $\mathbf{P}_n^{\RM}$ is an $n\times n$ permutation matrix which reorders the rows of $\mathbf{G}_{n/2}^{\RM} \otimes \mathbf{G}_2^{\polar}$ according to their Hamming weights.
It was conjectured in \cite{AY20} that for the matrix $\mathbf{G}_{n}^{\RM}$, the reliability of each message bit under the successive decoder becomes completely ordered, i.e., each message bit is always more reliable than its previous bit.
If this conjecture were true, then one could show that RM codes polarize faster than polar codes, which leads to a better finite-length performance.

Inspired by the recursive relation $\mathbf{G}_{n}^{\RM}= \mathbf{P}_n^{\RM} (\mathbf{G}_{n/2}^{\RM} \otimes \mathbf{G}_2^{\polar})$ of RM codes, we propose a new family of codes called the Adjacent-Bits-Swapped (ABS) polar codes.
In the construction of ABS polar codes, we use a similar recursive relation
$
	\mathbf{G}_{n}^{\ABS}= \mathbf{P}_n^{\ABS} (\mathbf{G}_{n/2}^{\ABS} \otimes \mathbf{G}_2^{\polar}).
$
The matrix $\mathbf{P}_n^{\ABS}$ is an $n\times n$ permutation matrix which swaps two adjacent rows if the two corresponding message bits are ``unordered", i.e., if the previous bit is more reliable than its next bit under the successive decoder.
Swapping such two adjacent rows always accelerates polarization because it makes the reliable bit even more reliable and the noisy bit even noisier.
While the permutation matrix $\mathbf{P}_n^{\RM}$ for RM codes involves a large number of swaps of adjacent rows, we limit the number of swaps in $\mathbf{P}_n^{\ABS}$ so that the overall structure of ABS polar codes is still close to standard polar codes.
In this way, we are able to devise a modified successive cancellation list (SCL) decoder to efficiently decode ABS polar codes.

Recall that both the code construction and the decoding algorithm of standard polar codes rely on a recursive relation between the bit-channels, which are the channels mapping from a message bit to the previous message bits and all the channel outputs.
Since we swap certain pairs of adjacent bits in the ABS polar code construction, there is no explicit recursive relation between bit-channels.
Instead, we introduce the notion of adjacent-bits-channels, which are $4$-ary-input channels mapping from two adjacent message bits to the previous message bits and all the channel outputs.
As the main technical contribution of this paper, we derive a recursive relation between the adjacent-bits-channels.
This recursive relation serves as the foundation of efficient code construction and decoding algorithms for ABS polar codes.

	{
		We provide two sets of simulation results to compare the performance of ABS polar codes and standard polar codes. First, we empirically calculate the scaling exponents of ABS polar codes and standard polar codes over a binary erasure channel with erasure probability $0.5$. Our calculations show that the scaling exponent of ABS polar codes is $3.37$ while the scaling exponent of standard polar codes is $3.65$, confirming that the polarization of ABS polar codes is indeed faster than standard polar codes. Second, we conduct extensive simulations over the binary-input AWGN channels for various choices of parameters.}
In particular, we have tested the performance for code length $256,512,1024,2048$.
For each choice of code length, we test $3$ code rates $0.3,0.5$ and $0.7$.
When we set the list size to be $32$ for the CRC-aided SCL decoders of both code families, ABS polar codes consistently outperform standard polar codes by $0.15\dB$---$0.3\dB$, but the decoding time of ABS polar decoder is longer than that of standard polar codes by roughly $60\%$.
If we use list size $20$ for ABS polar codes and keep the list size to be $32$ for standard polar codes, then the decoding time is more or less the same for these two codes, and ABS polar codes still outperform standard polar codes for most choices of parameters.
In this case, the improvement over standard polar codes is up to $0.15\dB$.

The organization of this paper is as follows: In Section~\ref{sect:main_idea}, we describe the main idea behind the ABS polar code construction and explain why ABS polar codes polarize faster than standard polar codes.
In Section~\ref{sect:construction}, we derive the new recursive relation between the adjacent-bits-channels and use this recursive relation to construct ABS polar codes.
In Section~\ref{sect:encoding}, we present an efficient encoding algorithm for ABS polar codes.
In Section~\ref{sect:SCL}, we present the new SCL decoder for ABS polar codes.
Finally, in Section~\ref{sect:simu}, we provide the simulation results.

\section{Main idea of the new code construction} \label{sect:main_idea}

\subsection{The polarization framework} \label{sect:polarization_framework}
Let $U_1,U_2,\dots,U_n$ be $n$ i.i.d.
Bernoulli-$1/2$ random variables.
Let $\mathbf{G}_n$ be an $n\times n$ invertible matrix over the binary field.
Define $(X_1,X_2,\dots,X_n)=(U_1,U_2,\dots,U_n) \mathbf{G}_n$.
We transmit each $X_i$ through a BMS channel $W$ and denote the channel output vector as $(Y_1,Y_2,\dots,Y_n)$.
In this framework, $(U_1,U_2,\dots,U_n)$ is the message vector, $\mathbf{G}_n$ is the encoding matrix, and $(X_1,X_2,\dots,X_n)$ is the codeword vector.
We use a successive decoder to recover the message vector from the channel output vector.
More precisely, we decode the coordinates in the message vector one by one from $U_1$ to $U_n$.
When decoding $U_i$, the successive decoder knows the values of all the previous message bits $U_1,\dots,U_{i-1}$ and all the channel outputs $Y_1,\dots,Y_n$.
Note that the codeword vector $(X_1,X_2,\dots,X_n)$ depends on the matrix $\mathbf{G}_n$, and the channel output vector $(Y_1,Y_2,\dots,Y_n)$ depends on both the matrix $\mathbf{G}_n$ and the BMS channel $W$, although we omit the dependence from their notation.
Next we define
\begin{equation} \label{eq:conditional_entropy}
	H_i(\mathbf{G}_n,W) := H(U_i | U_1,\dots,U_{i-1},Y_1,\dots,Y_n)
	\text{~for~} 1\le i\le n,
\end{equation}
where $H(\cdot | \cdot)$ is the conditional entropy.
$H_i(\mathbf{G}_n,W)$ measures the reliability of the $i$th message bit under the successive decoder when we use the encoding matrix $\mathbf{G}_n$ and transmit the codeword vector through the BMS channel $W$.
Since $\mathbf{G}_n$ is an invertible matrix, we have
\begin{equation} \label{eq:chain_rule}
	H_1(\mathbf{G}_n,W)+H_2(\mathbf{G}_n,W)+\dots+H_n(\mathbf{G}_n,W)=n(1-I(W)),
\end{equation}
where $I(W)$ is the channel capacity of $W$.
We say that a family of matrices $\{\mathbf{G}_n\}$ is polarizing over a BMS channel $W$ if $H_i(\mathbf{G}_n,W)$ is close to either $0$ or $1$ for almost all $i\in\{1,2,\dots,n\}$ as $n\to\infty$.
In order to quantify the polarization level of a given encoding matrix $\mathbf{G}_n$ over a BMS channel $W$, we define
$$
	\Gamma(\mathbf{G}_n, W)= \frac{1}{n}\sum_{i=1}^n H_i(\mathbf{G}_n,W) (1-H_i(\mathbf{G}_n,W) ).
$$
According to the definition above, a family of matrices $\{\mathbf{G}_n\}$ is polarizing over $W$ if and only if $\Gamma(\mathbf{G}_n, W)\to 0$ as $n\to\infty$.
A family of polarizing matrix $\{\mathbf{G}_n\}$ over a BMS channel $W$ allows us to construct capacity-achieving codes as follows: We include the $i$th row of $\mathbf{G}_{n}$ in the generator matrix if and only if $H_i(\mathbf{G}_n,W)$ is very close to $0$.
The condition $H_i(\mathbf{G}_n,W)\approx 0$ guarantees that the decoding error of the constructed codes approaches $0$ under the successive decoder.
We can further use \eqref{eq:chain_rule} to show that the code rate of the constructed codes approaches $I(W)$.
To see this, we first assume the extreme case where $\Gamma(\mathbf{G}_n, W) = 0$, i.e., $H_i(\mathbf{G}_n,W)$ is either $0$ or $1$ for all $1\le i\le n$.
Then by \eqref{eq:chain_rule} we know that the dimension of the constructed polar code is precisely $nI(W)$, i.e., the code rate is $R=I(W)$.
For the realistic case of $\Gamma(\mathbf{G}_n, W) \to 0$ as $n \to \infty$, one can show that the gap to capacity $I(W)-R$ also decreases to $0$ as $n\to\infty$.
Moreover, the smaller $\Gamma(\mathbf{G}_n, W)$ is, the smaller gap to capacity we have.

In the standard polar code construction \cite{Arikan09}, we construct the family of matrices $\{\mathbf{G}_{2^m}^{\polar}\}_{m=1}^\infty$ recursively using the following relation:
$$
	\mathbf{G}_2^{\polar} := \begin{bmatrix}
		1 & 0 \\
		1 & 1
	\end{bmatrix}
	\text{~and~~}
	\mathbf{G}_n^{\polar}= \mathbf{G}_{n/2}^{\polar} \otimes \mathbf{G}_2^{\polar}
	\text{~for~} n=2^m \ge 4,
$$
where $\otimes$ is the Kronecker product and $m>1$ is a positive integer.
It was shown in \cite{Arikan09} that $\{\mathbf{G}_{2^m}^{\polar}\}_{m=1}^\infty$ is polarizing over every BMS channel $W$, and the codes constructed from these matrices can be efficiently decoded.
In this paper, our objective is to construct another family of polarizing matrices $\{\mathbf{G}_{2^m}^{\ABS}\}_{m=1}^\infty$ satisfying the following two conditions: (1) $\Gamma(\mathbf{G}_{2^m}^{\ABS}, W) < \Gamma(\mathbf{G}_{2^m}^{\polar}, W)$, i.e., the matrices $\mathbf{G}_{2^m}^{\ABS}$ polarize even faster than $\mathbf{G}_{2^m}^{\polar}$; (2) the codes constructed from $\{\mathbf{G}_{2^m}^{\ABS}\}_{m=1}^\infty$ can also be efficiently decoded.
The first condition allows us to construct a new family of codes with smaller gap to capacity and better finite-length performance than standard polar codes.

\subsection{Swapping unordered adjacent bits accelerates polarization}\label{sect:reason}
The key observation in the standard polar code construction is that $\Gamma(\mathbf{G}_n, W)$ decreases as we perform the Kronecker product $\mathbf{G}_{2n} = \mathbf{G}_{n} \otimes \mathbf{G}_2^{\polar}$.
More precisely, we always have
$$
	\Gamma(\mathbf{G}_{n} \otimes \mathbf{G}_2^{\polar}, W) < \Gamma(\mathbf{G}_n, W)
$$
for every invertible matrix $\mathbf{G}_n$
as long as $I(W)$ is not equal to $0$ or $1$.
Therefore, the Kronecker product $\mathbf{G}_{2n}=\mathbf{G}_n \otimes \mathbf{G}_2^{\polar}$ deepens the polarization at the cost of increasing the code length by a factor of $2$.

In this paper, we observe that there is another method to deepen the polarization without increasing the code length, and this simple observation forms the foundation of our new code construction.
Given a matrix $\mathbf{G}_{n}$ and a BMS channel $W$, we say that two adjacent message bits $U_i$ and $U_{i+1}$ are unordered if $H_i(\mathbf{G}_n,W)\le H_{i+1}(\mathbf{G}_n,W)$.
This inequality means that $U_i$ is more reliable than $U_{i+1}$ under the successive decoder although $U_i$ is decoded before $U_{i+1}$.
Our key observation is that in this case, switching the decoding order of $U_i$ and $U_{i+1}$ deepens the polarization.
Intuitively, this is because switching the decoding order of these two bits makes the reliable bit even more reliable and the noisy bit even noisier.

Note that switching the decoding order of $U_i$ and $U_{i+1}$ is equivalent to swapping the $i$th row and the $(i+1)$th row of $\mathbf{G}_{n}$.
More specifically, let us define a new matrix $\overline{\mathbf{G}}_{n}$ as the matrix obtained from swapping the $i$th row and the $(i+1)$th row of $\mathbf{G}_{n}$ and keeping all the other rows the same as $\mathbf{G}_{n}$.
Following the framework in Section~\ref{sect:polarization_framework}, let $(\overline{U}_1,\overline{U}_2,\dots,\overline{U}_n)$ be the message vector associated with the new matrix $\overline{\mathbf{G}}_{n}$, where $\overline{U}_1,\dots,\overline{U}_n$ are $n$ i.i.d.
Bernoulli-$1/2$ random variables.
Let $(\overline{X}_1,\dots,\overline{X}_n)=(\overline{U}_1,\dots,\overline{U}_n) \overline{\mathbf{G}}_{n}$ be the codeword vector transmitted through the BMS channel $W$ and let $(\overline{Y}_1,\dots,\overline{Y}_n)$ be the corresponding channel output vector.
By definition \eqref{eq:conditional_entropy}, we have
$$
	H_j(\overline{\mathbf{G}}_n,W)= H(\overline{U}_j | \overline{U}_1,\dots,\overline{U}_{j-1},\overline{Y}_1,\dots,\overline{Y}_n)
	\text{~for~} 1\le j\le n.
$$
By the relation between the matrices $\overline{\mathbf{G}}_{n}$ and $\mathbf{G}_{n}$, we have
\begin{align}
	    & H_j(\mathbf{G}_n,W)=H_j(\overline{\mathbf{G}}_n,W) \text{~for all~} j\in\{1,2,\dots,n\}\setminus\{i,i+1\},
	\label{eq:d1}                                                                                                                                                  \\
	    & H_i(\mathbf{G}_n,W) = H(\overline{U}_{i+1} | \overline{U}_1,\dots,\overline{U}_{i-1},\overline{Y}_1,\dots,\overline{Y}_n)  \nonumber                     \\
	\ge & H(\overline{U}_{i+1} | \overline{U}_1,\dots,\overline{U}_{i-1},\overline{U}_i,\overline{Y}_1,\dots,\overline{Y}_n) = H_{i+1}(\overline{\mathbf{G}}_n,W),
	\label{eq:d2}                                                                                                                                                  \\
	    & H_{i+1}(\mathbf{G}_n,W) = H(\overline{U}_i | \overline{U}_1,\dots,\overline{U}_{i-1},\overline{U}_{i+1},\overline{Y}_1,\dots,\overline{Y}_n)  \nonumber  \\
	\le & H(\overline{U}_i | \overline{U}_1,\dots,\overline{U}_{i-1},\overline{Y}_1,\dots,\overline{Y}_n) = H_i(\overline{\mathbf{G}}_n,W). \label{eq:d3}
\end{align}
Now suppose that $U_i$ and $U_{i+1}$ are unordered, i.e., $H_i(\mathbf{G}_n,W)\le H_{i+1}(\mathbf{G}_n,W)$.
Combining this inequality with \eqref{eq:d2}--\eqref{eq:d3}, we obtain
\begin{equation} \label{eq:3inq}
	H_{i+1}(\overline{\mathbf{G}}_n,W) \le H_i(\mathbf{G}_n,W)\le H_{i+1}(\mathbf{G}_n,W) \le H_i(\overline{\mathbf{G}}_n,W).
\end{equation}
Moreover,
$$
	H_i(\mathbf{G}_n,W) + H_{i+1}(\mathbf{G}_n,W)
	= H_i(\overline{\mathbf{G}}_n,W) + H_{i+1}(\overline{\mathbf{G}}_n,W) = H(\overline{U}_i, \overline{U}_{i+1} | \overline{U}_1,\dots,\overline{U}_{i-1},\overline{Y}_1,\dots,\overline{Y}_n).
$$
This equality together with \eqref{eq:3inq} implies that
\begin{align*}
	    & (H_i(\mathbf{G}_n,W))^2 + (H_{i+1}(\mathbf{G}_n,W))^2                                                                                                                                       \\
	=   & \frac{1}{2} \Big( \big(H_i(\mathbf{G}_n,W) + H_{i+1}(\mathbf{G}_n,W) \big)^2 + \big(H_i(\mathbf{G}_n,W) - H_{i+1}(\mathbf{G}_n,W) \big)^2 \Big)                                             \\
	\le & \frac{1}{2} \Big( \big(H_i(\overline{\mathbf{G}}_n,W) + H_{i+1}(\overline{\mathbf{G}}_n,W) \big)^2 + \big(H_i(\overline{\mathbf{G}}_n,W) - H_{i+1}(\overline{\mathbf{G}}_n,W) \big)^2 \Big) \\
	=   & (H_i(\overline{\mathbf{G}}_n,W))^2 + (H_{i+1}(\overline{\mathbf{G}}_n,W))^2.
\end{align*}
Therefore,
\begin{align*}
	    & H_i(\overline{\mathbf{G}}_n,W) \big(1-H_i(\overline{\mathbf{G}}_n,W) \big) + H_{i+1}(\overline{\mathbf{G}}_n,W) \big(1- H_{i+1}(\overline{\mathbf{G}}_n,W) \big) \\
	\le & H_i(\mathbf{G}_n,W) \big(1- H_i(\mathbf{G}_n,W) \big) + H_{i+1}(\mathbf{G}_n,W) \big(1- H_{i+1}(\mathbf{G}_n,W) \big).
\end{align*}
Combining this with \eqref{eq:d1}, we conclude that $\Gamma(\overline{\mathbf{G}}_n, W) \le
	\Gamma(\mathbf{G}_n, W)$.
This formally justifies that switching the decoding order of two unordered adjacent bits deepens polarization.

\subsection{Our new code construction and its connection to RM codes}  \label{sect:connection_to_RM}

We view the operation of taking the Kronecker product $\mathbf{G}_{n}^{\polar}=\mathbf{G}_{n/2}^{\polar} \otimes \mathbf{G}_2^{\polar}$ in the standard polar code construction as one layer of polar transform.
Then the construction of a standard polar code with code length $n=2^m$ consists of $m$ consecutive layers of polar transforms.
In light of the discussion in Section~\ref{sect:reason}, we add a permutation layer after each polar transform layer in our ABS polar code construction.
More precisely, we replace the recursive relation $\mathbf{G}_{n}^{\polar}=\mathbf{G}_{n/2}^{\polar} \otimes \mathbf{G}_2^{\polar}$ in the standard polar code construction with
\begin{equation} \label{eq:GnABS}
	\mathbf{G}_{n}^{\ABS}= \mathbf{P}_n^{\ABS} (\mathbf{G}_{n/2}^{\ABS} \otimes \mathbf{G}_2^{\polar}),
\end{equation}
where the matrix $\mathbf{P}_n^{\ABS}$ is an $n\times n$ permutation matrix.
In this case, $\mathbf{G}_{n}^{\ABS}$ is a row permutation of the Kronecker product $\mathbf{G}_{n/2}^{\ABS} \otimes \mathbf{G}_2^{\polar}$.
The permutation associated with $\mathbf{P}_n^{\ABS}$ is a composition of multiple swaps of unordered adjacent bits.
The starting point of the recursive relation \eqref{eq:GnABS} is $\mathbf{G}_1^{\ABS}=[1]$, the identity matrix of size $1\times 1$.

Before we present how to choose $\mathbf{P}_n^{\ABS}$ in \eqref{eq:GnABS}, let us point out an interesting connection between our new code and RM codes.
In fact, RM codes can also be constructed using a similar recursive relation:
\begin{equation}  \label{eq:GnRM}
	\mathbf{G}_{n}^{\RM}= \mathbf{P}_n^{\RM} (\mathbf{G}_{n/2}^{\RM} \otimes \mathbf{G}_2^{\polar}).
\end{equation}
Here $\mathbf{P}_n^{\RM}$ is an $n\times n$ permutation matrix which reorders the rows of $\mathbf{G}_{n/2}^{\RM} \otimes \mathbf{G}_2^{\polar}$ according to their Hamming weights.
In other words, $\mathbf{G}_{n}^{\RM}$ is a row permutation of $\mathbf{G}_{n/2}^{\RM} \otimes \mathbf{G}_2^{\polar}$, and the Hamming weights of the rows of $\mathbf{G}_{n}^{\RM}$ are monotonically increasing from the first row to the last row.
It was shown in \cite{AY20} that the family of matrices $\{\mathbf{G}_{n}^{\RM}\}$ is polarizing over every BMS channel $W$, i.e., $H_i(\mathbf{G}_n^{\RM},W)$ is close to either $0$ or $1$ for almost all $i\in\{1,2,\dots,n\}$ as $n\to\infty$.
It was further conjectured\footnote{The authors of \cite{AY20} provided some theoretical analysis and simulation results to support this conjecture.} in \cite{AY20} that $\{H_i(\mathbf{G}_n^{\RM},W)\}_{i=1}^n$ is decreasing for every BMS channel $W$, i.e.,
\begin{equation} \label{eq:conjecture}
	H_1(\mathbf{G}_n^{\RM},W) \ge H_2(\mathbf{G}_n^{\RM},W) \ge \dots \ge H_n(\mathbf{G}_n^{\RM},W).
\end{equation}
If this conjecture were true, then we can immediately conclude that RM codes achieve capacity of BMS channels.
Indeed, RM codes choose rows with heaviest Hamming weight in $\mathbf{G}_{n}^{\RM}$ to form the generator matrices.
Since the rows of $\mathbf{G}_{n}^{\RM}$ are sorted according to their Hamming weights, RM codes simply pick the rows with large row indices.
By \eqref{eq:conjecture}, these rows correspond to the most reliable bits under the successive decoder.
Moreover, since almost all the conditional entropy in \eqref{eq:conjecture} are close to either $0$ or $1$, the conditional entropy of the most reliable bits must be close to $0$, and the number of such bits is close to $nI(W)$ as $n\to\infty$.

Moreover, the conjecture \eqref{eq:conjecture} indicates that RM codes do not have any unordered adjacent bits. According to the discussion in Section~\ref{sect:reason}, this  suggests that RM codes have fast polarization. In fact, it is widely believed that RM codes have a smaller gap to capacity than polar codes with the same parameters, which was suggested to be the case by both theoretical analysis \cite{Hassani14,Hassani18} and simulation results \cite{Mondelli14,Ye20}.

Although RM codes are believed to have better performance than polar codes under the Maximum Likelihood (ML) decoder, the problem of designing an efficient decoder whose performance is almost the same as the ML decoder still remains open for RM codes, except for a certain range of parameters.
In particular, the performance of currently known decoding algorithms for RM codes \cite{Dumer06,Ye20,Lian20,Geiselhart21} is close to the ML decoder only in the short code length or the low code rate regimes.
In contrast, the performance of the successive cancellation list (SCL) decoder with list size $32$ is almost the same as the ML decoder for polar codes.

Our new code construction is an intermediate point between RM codes and polar codes.
On the one hand, the recursive relation \eqref{eq:GnABS} of our new code is similar to the recursion \eqref{eq:GnRM} of RM codes in the sense that both codes add a permutation layer after each polar transform layer to accelerate polarization.
On the other hand, we only use a relatively small number of swaps in the permutation matrix $\mathbf{P}_n^{\ABS}$ while the permutation matrix $\mathbf{P}_n^{\RM}$ for RM codes involves a large number of swaps.
As a consequence, the overall structure of our new code is still close to the standard polar codes, and it allows a modified SCL decoder to efficiently decode.

In order to explain how to choose $\mathbf{P}_n^{\ABS}$ in \eqref{eq:GnABS}, we introduce a sequence of permutation matrices. For $1\le i\le n-1$, we use $\mathbf{S}_n^{(i)}$ to denote the $n\times n$ permutation matrix that swaps $i$ and $i+1$ while mapping all the other elements to themselves. More precisely, only $4$ entries of $\mathbf{S}_n^{(i)}$ are different from the identity matrix. These $4$ entries are $\mathbf{S}_n^{(i)}(i,i)=\mathbf{S}_n^{(i)}(i+1,i+1)=0$ and $\mathbf{S}_n^{(i)}(i,i+1)=\mathbf{S}_n^{(i)}(i+1,i)=1$, where $\mathbf{S}_n^{(i)}(a,b)$ is the entry of $\mathbf{S}_n^{(i)}$ located at the cross of the $a$th row and the $b$th column. The permutation matrix $\mathbf{P}_n^{\ABS}$ can be written as
\begin{equation} \label{eq:pi}
	\mathbf{P}_n^{\ABS} = \prod_{i\in\cI^{(n)}} \mathbf{S}_n^{(i)} ,
\end{equation}
where $\cI^{(n)}$ is a subset of $\{1,2,\dots,n-1\}$. Let us write $\cI^{(n)}=\{i_1,i_2,\dots,i_s\}$, where $s$ is the size of $\cI^{(n)}$. In the ABS polar code construction, we require that
\begin{equation} \label{eq:separated}
	i_2\ge i_1+4, \quad
	i_3\ge i_2+4, \quad
	i_4\ge i_3+4, \quad
	\dots, \quad
	i_s\ge i_{s-1}+4.
\end{equation}
This condition guarantees that the swapped elements are fully separated, and it is the foundation of efficient code construction and efficient decoding for ABS polar codes.
More specifically, the condition \eqref{eq:separated} allows us to efficiently track the evolution of every pair of adjacent bits through different layers of polar transforms in a recursive way.
We will explain the details about this in Section~\ref{sect:construction}.
As a final remark, we note that one needs to choose $m$ permutation matrices $\mathbf{P}_2^{\ABS},\mathbf{P}_4^{\ABS},\mathbf{P}_8^{\ABS},\dots,\mathbf{P}_n^{\ABS}$ in the construction of an ABS polar code with code length $n=2^m$.

\subsection{Comparison with the large kernel method} \label{sect:cmp_lkm}

The finite-length scaling of polar codes is an important research topic in the polar coding literature \cite{Hassani14, Guruswami15, Mondelli15, Mondelli16}. The ABS polar code construction proposed in this paper is one way to improve the scaling exponent of polar codes. Another extensively-studied method is to use large kernels instead of the Ar{\i}kan kernel $\mathbf{G}_2^{\polar}$ in the polar code construction \cite{Buz2017ISIT, Buz2017, Ye2015, Fazeli21, Wang21, Guruswami22}. In particular, it was shown in \cite{Fazeli21, Wang21, Guruswami22} that when the kernel size goes to infinity, the scaling exponent of polar codes approaches the optimal value $2$.

Compared to the ABS polar code construction, the large kernel method has the following three disadvantages: (i) The choice of code length is more restrictive. The code length of ABS polar codes can be any power of $2$, but the code length of polar codes with large kernels must be a power of the kernel size $\ell$, where $\ell$ is larger than $2$. Some typical choices of $\ell$ are $4,8,16$. (ii) The code construction is also more restrictive. In the original large kernel method, the same kernel is used repetitively throughout the whole code construction. In contrast, we use different permutation matrices $\mathbf{P}_2^{\ABS},\mathbf{P}_4^{\ABS},\mathbf{P}_8^{\ABS},\dots,\mathbf{P}_n^{\ABS}$ in different layers. (iii) The decoding complexity is much larger. For $\ell\times\ell$ kernels, the decoding time increases by a factor of $2^\ell$ compared to standard polar codes. In contrast, the decoding time of ABS polar codes only increases by $60\%$ compared to standard polar codes, as indicated by the simulation results in Section~\ref{sect:simu}.

Among the research on polar codes with large kernels, the permuted kernels and the permuted successive cancellation (PSC) decoder proposed in \cite{Buz2017ISIT, Buz2017} are particularly relevant to our paper. More specifically, \cite{Buz2017ISIT, Buz2017} proposed to use permuted kernels, whose size $\ell$ is a power of $2$. As suggested by its name, the permuted kernel is a row permutation of $\mathbf{G}_{\ell}^{\polar}$. This is similar in nature to the ABS polar code construction because the encoding matrix $\mathbf{G}_n^{\ABS}$ of ABS polar codes is also a row permutation of $\mathbf{G}_n^{\polar}$. Moreover, \cite{Buz2017ISIT, Buz2017} further proposed the PSC decoder to efficiently decode polar codes with permuted kernels. The PSC decoder together with the permuted kernels significantly reduces the decoding time compared to the standard SC decoder for polar codes with large kernels. In other words, the PSC decoder and permuted kernels mitigate the third disadvantage above. However, the first two disadvantages still remain, i.e., the choice of code length and the code construction are still more restrictive than ABS polar codes.

\section{Code construction of ABS polar codes} \label{sect:construction}

The construction of ABS polar codes with code length $n=2^m$ consists of two main steps.
The first step is to pick the permutation matrices $\mathbf{P}_2^{\ABS},\mathbf{P}_4^{\ABS},\mathbf{P}_8^{\ABS},\dots,\mathbf{P}_n^{\ABS}$ in the recursive relation \eqref{eq:GnABS}, as mentioned at the end of the previous section.
After picking these permutation matrices, the second step is to find which bits are information bits and which bits are frozen bits.
Although the second step is also needed in the construction of standard polar codes \cite{Arikan09,Tal13}, the techniques used in this paper are quite different.
In the standard polar code construction, we can directly track the evolution of bit-channels in a recursive way.
However, in the ABS polar code construction, it is not possible to identify a recursive relation between bit-channels directly because we swap certain pairs of adjacent bits in the code construction.
Instead, we find a recursive relation between pairs of adjacent bits from different layers of polar transforms.
After obtaining the joint distribution of every pair of adjacent bits, we are able to calculate the transition probability of the bit-channels and locate the information bits and the frozen bits.

The organization of this section is as follows: In Section~\ref{sect:2x2transform}, we first recall how to track the evolution of bit-channels in standard polar codes using the basic $2\times 2$ transform.
In Section~\ref{sect:DB}, we introduce a new transform and use it to establish a recursive relation between pairs of adjacent bits for standard polar codes.
The purpose of Section~\ref{sect:DB} is to illustrate the application of the new transform in a familiar setting.
In Section~\ref{sect:SDB}, we use the new transform to track the evolution of adjacent bits in the ABS polar codes.
The result in Section~\ref{sect:SDB} accomplishes the second step  of the ABS polar code construction, i.e., it allows us to locate the information bits and the frozen bits when the permutation matrices $\mathbf{P}_2^{\ABS},\mathbf{P}_4^{\ABS},\mathbf{P}_8^{\ABS},\dots,\mathbf{P}_n^{\ABS}$ in the recursive relation \eqref{eq:GnABS} are known.
Next, in Section~\ref{sect:cons_PnABS}, we explain how to pick these permutation matrices in the ABS polar code construction.
Recall that the quantization operation is needed in the standard polar code construction \cite{Tal13} because the output alphabet size of the bit-channels grows exponentially with $n$.
The same issue also arises in the ABS polar code construction, and we will discuss this in Section~\ref{sect:quantization}.
Finally, we put everything together and summarize the code construction algorithm for ABS polar codes in Section~\ref{sect:summary_cons}.

\subsection{Tracking the evolution of bit-channels in standard polar codes using the $2\times 2$ transform} \label{sect:2x2transform}

Let us first recall the $2\times 2$ transform in the standard polar code construction.

\begin{figure}[ht]
	\centering
	\begin{subfigure}{.45\linewidth}
		\centering
		\begin{tikzpicture}
			\draw
			node at (0,10.5) [] (u1)  {$U_1$}
			node at (0,9) [] (u2)  {$U_2$}
			node at (1.5,10.5) [XOR,scale=1.2] (x1) {}
			node at (2.5,10.5) [] (xx1)  {$X_1$}
			node at (2.5,9) [] (xx2)  {$X_2$}
			node at (3.8,10.5) [block] (v1)  {$W$}
			node at (3.8,9) [block] (v2)  {$W$}
			node at (5.5,10.5) [] (y1)  {$Y_1$}
			node at (5.5,9) [] (y2)  {$Y_2$};
			\draw[fill] (1.5, 9) circle (.6ex);
			\draw[very thick,->](u1) -- node {}(x1);
			\draw[very thick,->](u2) -| node {}(x1);
			\draw[very thick,->](x1) -- (xx1);
			\draw[very thick,->](u2) -- (xx2);
			\draw[very thick,->](xx1) -- (v1);
			\draw[very thick,->](xx2) -- (v2);
			\draw[very thick,->](v1) -- node {}(y1);
			\draw[very thick,->](v2) -- node {}(y2);
		\end{tikzpicture}
		\caption{Multiply i.i.d.
		Bernoulli-$1/2$ random variables $(U_1,U_2)$ with the matrix $\mathbf{G}_2^{\polar}$, and then transmit the results through two copies of $W$.
		Under the successive decoder, this transforms two copies of $W$ into $W^-:U_1\to Y_1,Y_2$ and $W^+:U_2\to U_1,Y_1,Y_2$.}
		\label{fig:polar_transform_a}
	\end{subfigure}
	\hspace*{0.2in}
	\begin{subfigure}{.45\linewidth}
		\centering
		\begin{tikzpicture}
			\draw
			node at (0,1.5) [] (u1)  {$W$}
			node at (0,0) [] (u2)  {$W$}
			node at (1.2,1.5) [] (v1)  {}
			node at (1.2,0) [] (v2)  {}
			node at (3.7,1.5) [] (x1)  {}
			node at (3.7,0) [] (x2)  {}
			node at (3.7,1.5) [] (x1)  {}
			node at (3.7,0) [] (x2)  {}
			node at (5.2,1.5) [] (y1)  {$W^-$}
			node at (5.2,0) [] (y2)  {$W^+$}
			node at (2.45,0.75) [text width=2cm,align=center] {$2\times 2$\\Transform};
			\draw[thick] (1.2,-0.5) rectangle (3.7,2);
			\draw[very thick,->](u1) -- node {}(v1);
			\draw[very thick,->](u2) -- node {}(v2);
			\draw[very thick,->](x1) -- node[above] {``$-$"}(y1);
			\draw[very thick,->](x2) -- node [above] {``$+$"} (y2);
		\end{tikzpicture}
		\caption{We take two independent copies of $W$ as inputs.
			After the transform, we obtain a ``worse" channel $W^-:U_1\to Y_1,Y_2$ and a ``better" channel $W^+:U_2\to U_1,Y_1,Y_2$.}
	\end{subfigure}
	\caption{The $2\times 2$ basic polar transform}
	\label{fig:polar_transform}
\end{figure}
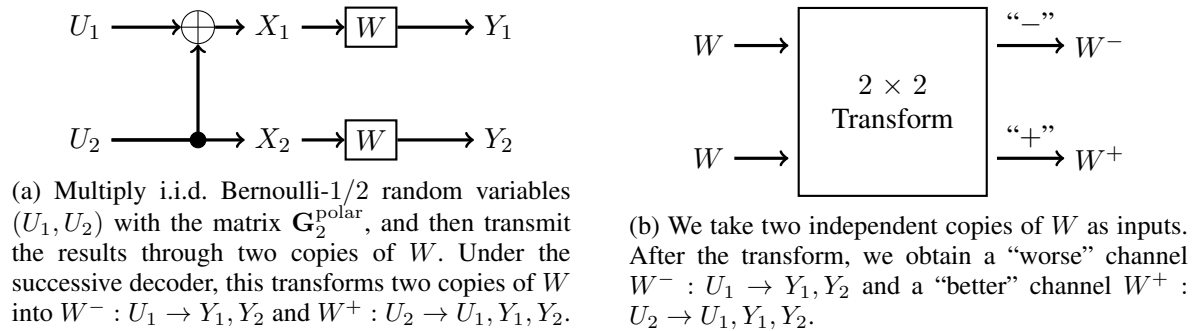
Given a BMS channel $W:\{0,1\}\to\cY$, the transition probabilities of $W^-:\{0,1\}\to\cY^2$ and $W^+:\{0,1\}\to\{0,1\}\times\cY^2$ in Fig.~\ref{fig:polar_transform} are given by
\begin{equation} \label{eq:polar_transform}
	\begin{aligned}
		 & W^-(y_1,y_2|u_1)  = \frac{1}{2} \sum_{u_2\in\{0,1\}} W(y_1|u_1+u_2) W(y_2|u_2) \quad \text{for~} u_1\in\{0,1\} \text{~and~} y_1,y_2\in\cY, \\
		 & W^+(u_1,y_1,y_2|u_2)  = \frac{1}{2} W(y_1|u_1+u_2) W(y_2|u_2) \quad \text{for~} u_1,u_2\in\{0,1\} \text{~and~} y_1,y_2\in\cY.
	\end{aligned}
\end{equation}
The basic $2\times 2$ transform plays a fundamental role in the standard polar code construction because it allows us to efficiently track the evolution of bit-channels in a recursive way.
More specifically, the bit-channels induced by the matrix $\mathbf{G}_n^{\polar}$ are defined in Fig.~\ref{fig:bit_channels_polar} below.
It is well known that the bit-channels associated with $\mathbf{G}_n^{\polar}$ and the bit-channels associated with $\mathbf{G}_{n/2}^{\polar}$ satisfy the following recursive relation:
\begin{equation} \label{eq:recur_bit_channels}
	W_{2i-1}^{(n)}=(W_i^{(n/2)})^-  \text{~~and~~}
	W_{2i}^{(n)}=(W_i^{(n/2)})^+   \text{~~for~} 1\le i\le n/2.
\end{equation}
Both the code construction and the decoding algorithm of standard polar codes rely on this recursive relation.

\begin{figure}[ht]
	\centering
	\begin{tikzpicture}
		\node [block, align=center] at (3,1.6)  (y1) { $U_1$ \\[0.5em]  $U_2$  \\[0.5em]  \vdots \\[0.5em]  $U_n$ };
		\node [sblock, align=center] at (5,1.6)  (y2) {$\mathbf{G}_n^{\polar}$};
		\node [block, align=center] at (7,1.6)  (y3) { $X_1$ \\[0.5em] $X_2$  \\[0.5em] \vdots  \\[0.5em]  $X_n$ };
		\node [block] at (8.5, 2.7) (w1) {$W$};
		\node [block] at (8.5, 2) (w2) {$W$};
		\node  at (8.5, 1.3)  {\vdots};
		\node [block] at (8.5, 0.6) (w3) {$W$};

		\node at (10, 2.7) (z1) {$Y_1$};
		\node at (10, 2) (z2) {$Y_2$};
		\node  at (10, 1.3) {\vdots};
		\node at (10, 0.6) (z3) {$Y_n$};

		\draw[->,thick] (y1)--(y2);
		\draw[->,thick] (y2)--(y3);

		\draw[->,thick] (w1)--(z1);
		\draw[->,thick] (w2)--(z2);
		\draw[->,thick] (w3)--(z3);

		\draw[->,thick] (7.4, 2.7)--(w1);
		\draw[->,thick] (7.4, 2)--(w2);
		\draw[->,thick] (7.4, 0.6)--(w3);

		\node at (5, -0.4) (p1) {$(X_1,\dots,X_n)=(U_1,\dots,U_n) \mathbf{G}_n^{\polar}$};

		\node at (15,1.6) [align=left] {$W_1^{(n)}:U_1\to Y_1,\dots,Y_n$ \\[3pt]
		$W_2^{(n)}:U_2\to U_1,Y_1,\dots,Y_n$ \\[3pt]
		$W_3^{(n)}:U_3\to U_1,U_2,Y_1,\dots,Y_n$\\
		\hspace*{1in} \vdots\\
		$W_n^{(n)}:U_n\to U_1,\dots,U_{n-1},Y_1,\dots,Y_n$};

		\node at (15, -0.4) {$n$ bit-channels induced by $\mathbf{G}_n^{\polar}$};
	\end{tikzpicture}
	\caption{$U_1,\dots,U_n$ are $n=2^m$ i.i.d.
	Bernoulli-$1/2$ random variables.
	$(X_1,\dots,X_n)=(U_1,\dots,U_n) \mathbf{G}_n^{\polar}$ is the codeword vector, and $(Y_1,\dots,Y_n)$ is the channel output vector.
	The $n$ bit-channels induced by $\mathbf{G}_n^{\polar}$ are listed on the right side of the figure.
	$W_i^{(n)}$ is the bit-channel mapping from $U_i$ to $U_1,\dots,U_{i-1},Y_1,\dots,Y_n$.}
	\label{fig:bit_channels_polar}
\end{figure}
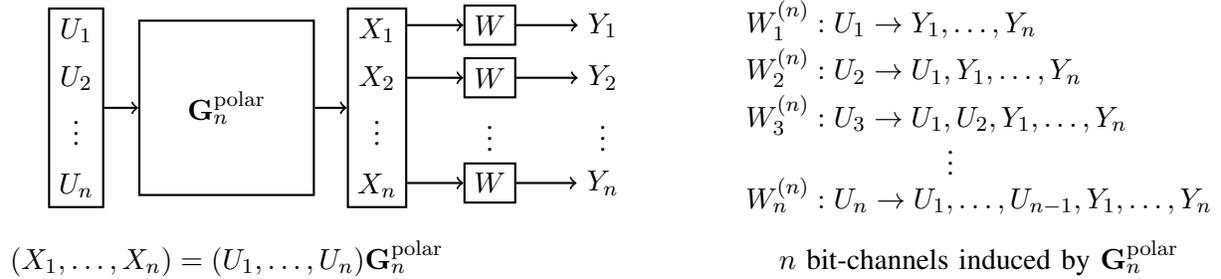

\subsection{Tracking the evolution of adjacent bits in standard polar codes using a new transform} \label{sect:DB}

In the construction of ABS polar codes, we need to track the joint distribution of every pair of adjacent bits, not just the distribution of every single bit given the previous bits and channel outputs.
To that end, we introduce a new transform, named as the Double-Bits (DB) polar transform.
All the channels involved in the DB polar transform have $4$-ary inputs.
To distinguish between binary-input channels and $4$-ary-input channels, we use $W$ to denote the former channels and use $V$ to denote latter channels\footnote{More precisely, $W$ and its variations such as $W^+,W^-,W_i^{(n)}$ are used for binary-input channels; $V$ and its variations such as $V^{\triangledown},V^{\lozenge},V^{\vartriangle},V_i^{(n)}$ are used for channels with $4$-ary inputs.}.
The details of the DB polar transform are illustrated in Fig.~\ref{fig:DBpt}.
Given a $4$-ary-input channel $V:\{0,1\}^2\to\cY$, the transition probabilities of $V^{\triangledown}:\{0,1\}^2\to\cY^2,V^{\lozenge}:\{0,1\}^2\to\{0,1\}\times\cY^2$, and $V^{\vartriangle}:\{0,1\}^2\to\{0,1\}^2\times\cY^2$ in Fig.~\ref{fig:DBpt} are given by
\begin{equation} \label{eq:DBpt}
	\begin{aligned}
		 & V^{\triangledown}(y_1,y_2|u_1,u_2)=\frac{1}{4}\sum_{u_3,u_4\in\{0,1\}} V(y_1|u_1+u_2,u_3+u_4) V(y_2|u_2,u_4) \\
		 & \hspace*{2.7in} \text{~for~} u_1,u_2\in\{0,1\} \text{~and~} y_1,y_2\in\cY,                                   \\
		 & V^{\lozenge}(u_1,y_1,y_2|u_2,u_3)=\frac{1}{4}\sum_{u_4\in\{0,1\}} V(y_1|u_1+u_2,u_3+u_4) V(y_2|u_2,u_4)      \\
		 & \hspace*{2.7in} \text{~for~} u_1,u_2,u_3\in\{0,1\} \text{~and~} y_1,y_2\in\cY,                               \\
		 & V^{\vartriangle}(u_1,u_2,y_1,y_2|u_3,u_4)=\frac{1}{4} V(y_1|u_1+u_2,u_3+u_4) V(y_2|u_2,u_4)                  \\
		 & \hspace*{2.7in} \text{~for~} u_1,u_2,u_3,u_4\in\{0,1\} \text{~and~} y_1,y_2\in\cY.
	\end{aligned}
\end{equation}

\begin{figure}[ht]
	\centering
	\begin{subfigure}{.53\linewidth}
		\centering
		\begin{tikzpicture}
			\draw
			node at (0,10.5) [] (u1)  {$U_1$}
			node at (0,9.5) [] (u2)  {$U_2$}
			node at (1.5,10.5) [XOR,scale=1.2] (x1) {}
			node at (4.1,10.5) (v1) {}
			node at (4.1,9.5) (v2) {}
			node at (4.5,10) (t1) {}
			node at (4.3,10) [vblock] (76)  {$V$}
			node at (6,10) [] (y1)  {$Y_1$};
			\draw[fill] (1.5, 9.5) circle (.6ex);

			\draw
			node at (0,8.5) [] (u3)  {$U_3$}
			node at (0,7.5) [] (u4)  {$U_4$}
			node at (1.5,8.5) [XOR,scale=1.2] (x3) {}
			node at (4.1,8.5) (v3) {}
			node at (4.1,7.5) (v4) {}
			node at (4.5,8) (t3) {}
			node at (4.3,8) [vblock] (83)  {$V$}
			node at (6,8) [] (y3)  {$Y_2$};
			\draw[fill] (1.5, 7.5) circle (.6ex);

			\draw[very thick,->](u1) -- node {}(x1);
			\draw[very thick,->](u2) -| node {}(x1);
			\draw[very thick,->](x1) -- (v1);
			\draw[very thick,->](u2) -- (2.2,9.5) -- (3.3,8.5) -- (v3);
			\draw[very thick,->](t1) -- node {}(y1);

			\draw[very thick,->](x3) -- (2.2,8.5) -- (3.3,9.5) -- (v2);
			\draw[very thick,->](u4) -| node {}(x3);
			\draw[very thick,->](u3) -- node {}(x3);
			\draw[very thick,->](u4) -- node {}(v4);
			\draw[very thick,->](t3) -- node {}(y3);

		\end{tikzpicture}
		\caption{$U_1,U_2,U_3,U_4$ are i.i.d.
			Bernoulli-$1/2$ random variables.
			The channel $V:\{0,1\}^2\to \cY$ takes two bits as its inputs, i.e., $V$ has $4$-ary inputs.
			Under the successive decoder, we have the following three channels: (1) $V^{\triangledown}:U_1,U_2\to Y_1,Y_2$; (2) $V^{\lozenge}:U_2,U_3\to U_1,Y_1,Y_2$; (3) $V^{\vartriangle}:U_3,U_4\to U_1,U_2,Y_1,Y_2$.}
		\label{fig:DBpt_a}
	\end{subfigure}
	\hfill
	\begin{subfigure}{.43\linewidth}
		\centering
		\begin{tikzpicture}
			\draw
			node at (0,1.5) [] (u1)  {$V$}
			node at (0,0) [] (u2)  {$V$}
			node at (1.4,1.5) [] (v1)  {}
			node at (1.4,0) [] (v2)  {}
			node at (3.5,2.2) [] (x1)  {}
			node at (3.5,0.75) [] (x2)  {}
			node at (3.5,-0.7) [] (x3)  {}
			node at (5.2,2.2) [] (y1)  {$V^{\triangledown}$}
			node at (5.2,0.75) [] (y2)  {$V^{\lozenge}$}
			node at (5.2,-0.7) [] (y3)  {$V^{\vartriangle}$}
			node at (2.45,0.75) [lgblock, align=center] {DB polar\\Transform};
			\draw[very thick,->](u1) -- node {}(v1);
			\draw[very thick,->](u2) -- node {}(v2);
			\draw[very thick,->](x1) -- node[above] {``$\triangledown$"}(y1);
			\draw[very thick,->](x2) -- node [above] {``$\lozenge$"} (y2);
			\draw[very thick,->](x3) -- node [above] {``$\vartriangle$"} (y3);
		\end{tikzpicture}
		\caption{Two independent copies of $V$ are transformed into three channels $V^{\triangledown},V^{\lozenge},V^{\vartriangle}$.
		These three channels also have $4$-ary inputs.
		Note that the inputs of $V^{\triangledown}$ and $V^{\lozenge}$ have one-bit overlap, and the inputs of $V^{\lozenge}$ and $V^{\vartriangle}$ also have one-bit overlap.}
	\end{subfigure}
	\caption{The Double-Bits (DB) polar transform}
	\label{fig:DBpt}
\end{figure}
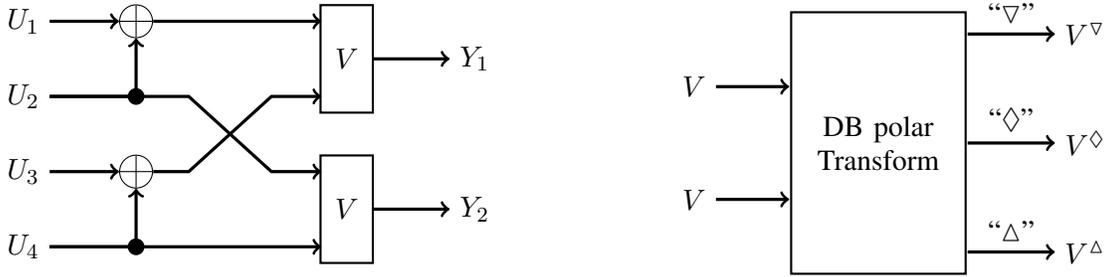

The role of the DB polar transform in the construction of ABS polar codes is the same as the role of the $2\times 2$ basic polar transform in the standard polar code construction.
Instead of jumping directly into the ABS polar code construction, let us first use standard polar codes to illustrate how to track the evolution of adjacent bits recursively using the DB polar transform.
In order to calculate the joint distribution of adjacent bits, we introduce the notion of adjacent-bits-channels, which is the counterpart of the bit-channels used for tracking the distribution of every single bit.
We still use the setting in Fig.~\ref{fig:bit_channels_polar}, where we defined the bit-channels.
For the matrix $\mathbf{G}_n^{\polar}$ and a BMS channel $W$, we define $n-1$ adjacent-bits-channels $V_1^{(n)},V_2^{(n)},\dots,V_{n-1}^{(n)}$ as follows:
\begin{equation}  \label{eq:st_adc}
	V_i^{(n)}: U_i,U_{i+1} \to U_1,\dots,U_{i-1},Y_1,\dots,Y_n
	\text{~~for~} 1\le i\le n-1,
\end{equation}
where $U_1,\dots,U_n,Y_1,\dots,Y_n$ are defined in Fig.~\ref{fig:bit_channels_polar}.
By definition, $V_1^{(n)},V_2^{(n)},\dots,V_{n-1}^{(n)}$ take two bits as their inputs, i.e., all of them have $4$-ary inputs.
Moreover, these adjacent-bits-channels depend on the BMS channel $W$, although we omit this dependence in the notation.

The following lemma allows us to calculate $V_1^{(n)},V_2^{(n)},\dots,V_{n-1}^{(n)}$ recursively from $V_1^{(n/2)},V_2^{(n/2)},\dots,
	\linebreak[4]
	V_{n/2-1}^{(n/2)}$.
\begin{lemma} \label{lemma:recur_ST_DB}
	Let $n\ge 4$.
	We have
	\begin{equation} \label{eq:abc}
		V_{2i-1}^{(n)} = (V_i^{(n/2)})^\triangledown,
		\quad
		V_{2i}^{(n)} = (V_i^{(n/2)})^\lozenge, \quad
		V_{2i+1}^{(n)} = (V_i^{(n/2)})^\vartriangle  \quad
		\text{for~} 1\le i\le n/2-1.
	\end{equation}
\end{lemma}

The proof of Lemma~\ref{lemma:recur_ST_DB} is given in Appendix~\ref{ap:lm1}.
The relation \eqref{eq:abc} is similar in nature to the relation \eqref{eq:recur_bit_channels}, and the proof of \eqref{eq:abc} also uses the same method as the proof of \eqref{eq:recur_bit_channels}.
There is, however, one difference between these two recursive relations: The ``$+$" and ``$-$" transforms of different bit-channels are distinct while the ``$\triangledown$", ``$\lozenge$" and ``$\vartriangle$" transforms of different adjacent-bits-channels may overlap.
More precisely, the $n/2$ sets $\{(W_i^{(n/2)})^-,(W_i^{(n/2)})^+\}_{i=1}^{n/2}$ are disjoint while the two sets $\{(V_i^{(n/2)})^\triangledown,(V_i^{(n/2)})^\lozenge,(V_i^{(n/2)})^\vartriangle\}$ and $\{(V_{i+1}^{(n/2)})^\triangledown,(V_{i+1}^{(n/2)})^\lozenge,(V_{i+1}^{(n/2)})^\vartriangle\}$ have the following element in their intersection for every $1\le i\le n/2-2$:
\begin{equation} \label{eq:2w}
	V_{2i+1}^{(n)} = (V_i^{(n/2)})^\vartriangle
	= (V_{i+1}^{(n/2)})^\triangledown.
\end{equation}
This gives us two methods of calculating $V_{2i+1}^{(n)}$ recursively for $1\le i\le n/2-2$.

Lemma~\ref{lemma:recur_ST_DB} tells us how to calculate $\{V_i^{(n)}\}_{i=1}^{n-1}$ from $\{V_i^{(n/2)}\}_{i=1}^{n/2-1}$ recursively for $n\ge 4$.
The last question we need to answer is how to calculate the adjacent-bits-channel $V_1^{(2)}$ from the BMS channel $W$, because $V_1^{(2)}$ is the starting point of the recursive relation in Lemma~\ref{lemma:recur_ST_DB}.
Fortunately, this is an easy task.
Let us go back to the setting in Fig.~\ref{fig:polar_transform}.
Given a BMS channel $W$, the adjacent-bits-channel $V_1^{(2)}$ is simply the channel mapping from $U_1,U_2$ to $Y_1,Y_2$.
More precisely, we have
\begin{equation} \label{eq:v_init}
	V_1^{(2)}(y_1,y_2|u_1,u_2)=W(y_1|u_1+u_2) W(y_2|u_2).
\end{equation}
After obtaining the transition probabilities of the adjacent-bits-channels $\{V_i^{(n)}\}_{i=1}^{n-1}$, it is straightforward to calculate the transition probabilities of the bit-channels $\{W_i^{(n)}\}_{i=1}^n$.
More precisely, we have
\begin{equation}  \label{eq:v_to_w}
	\begin{aligned}
		 & W_i^{(n)}(y_1,y_2,\dots,y_n,u_1,u_2,\dots,u_{i-1}|u_i) = \frac{1}{2} \sum_{u_{i+1}\in\{0,1\}} V_i^{(n)}(y_1,y_2,\dots,y_n,u_1,u_2,\dots,u_{i-1}|u_i,u_{i+1}), \\
		 & W_{i+1}^{(n)}(y_1,y_2,\dots,y_n,u_1,u_2,\dots,u_i|u_{i+1}) = \frac{1}{2}  V_i^{(n)}(y_1,y_2,\dots,y_n,u_1,u_2,\dots,u_{i-1}|u_i,u_{i+1})
	\end{aligned}
\end{equation}
for $1\le i\le n-1$.

As a final remark, we note that the output alphabet size of the adjacent-bits-channels $\{V_i^{(n)}\}_{i=1}^{n-1}$ grows exponentially with $n$.
Therefore, accurate calculations of $\{V_i^{(n)}\}_{i=1}^{n-1}$ are intractable.
We need to quantize the output alphabets by merging output symbols with similar posterior distributions.
Recall that in the standard polar code construction \cite{Tal13}, we also need the quantization operation to calculate an approximation of the bit-channels $\{W_i^{(n)}\}_{i=1}^n$.
Our quantization method is different from the one used in \cite{Tal13} because the adjacent-bits-channels have $4$-ary inputs while the bit-channels have binary inputs.
We will present our quantization method later in Section~\ref{sect:quantization}.

\subsection{Tracking the evolution of adjacent bits in ABS polar codes} \label{sect:SDB}

As discussed at the beginning of this section, the construction of ABS polar codes consists of two main steps.
The first step is to pick the permutation matrices $\mathbf{P}_2^{\ABS},\mathbf{P}_4^{\ABS},\mathbf{P}_8^{\ABS},\dots,\mathbf{P}_n^{\ABS}$ in the recursive relation \eqref{eq:GnABS}, and the second step is to find which bits are information bits and which bits are frozen bits after picking these permutation matrices.
In this subsection, we explain how to accomplish the second step.
More precisely, we define the bit-channels and the adjacent-bits-channels for ABS polar codes in Fig.~\ref{fig:abc_in_ABS}.
The task of this subsection is to show how to calculate the capacity of the bit-channels $\{W_i^{(n),\ABS}\}_{i=1}^n$ when the permutation matrices $\mathbf{P}_2^{\ABS},\mathbf{P}_4^{\ABS},\mathbf{P}_8^{\ABS},\dots,\mathbf{P}_n^{\ABS}$ in \eqref{eq:GnABS} are known.
Then the information bits are simply the $U_i$'s satisfying that $I(W_i^{(n),\ABS})\approx 1$, where $I(\cdot)$ is the channel capacity.
Unlike the standard polar codes, there does not exist a recursive relation between the bit-channels $\{W_i^{(n),\ABS}\}_{i=1}^n$ and $\{W_i^{(n/2),\ABS}\}_{i=1}^{n/2}$ for ABS polar codes.
Instead, we derive a recursive relation between the adjacent-bits-channels $\{V_i^{(n),\ABS}\}_{i=1}^{n-1}$ and $\{V_i^{(n/2),\ABS}\}_{i=1}^{n/2-1}$.
After that, the transition probabilities of $\{W_i^{(n),\ABS}\}_{i=1}^n$ can be calculated from the transition probabilities of $\{V_i^{(n),\ABS}\}_{i=1}^{n-1}$.

\begin{figure}[ht]
	\centering
	\begin{tikzpicture}
		\node [block, align=center] at (-1,1.6)  (y-1) { $U_1$ \\[0.5em]  $U_2$  \\[0.5em]  \vdots \\[0.5em]  $U_n$ };
		\node [sblock, align=center] at (1,1.6)  (y0) {$\mathbf{P}_n^{\ABS}$};
		\node [block, align=center] at (3,1.6)  (y1) { $\widehat{U}_1$ \\[0.5em]  $\widehat{U}_2$  \\[0.5em]  \vdots \\[0.5em]  $\widehat{U}_n$ };
		\node [sblock, align=center] at (5,1.6)  (y2) {$\mathbf{G}_{n/2}^{\ABS} \otimes \mathbf{G}_2^{\polar}$};
		\node [block, align=center] at (7,1.6)  (y3) { $X_1$ \\[0.5em] $X_2$  \\[0.5em] \vdots  \\[0.5em]  $X_n$ };
		\node [block] at (8.5, 2.7) (w1) {$W$};
		\node [block] at (8.5, 2) (w2) {$W$};
		\node  at (8.5, 1.3)  {\vdots};
		\node [block] at (8.5, 0.6) (w3) {$W$};

		\node at (10, 2.7) (z1) {$Y_1$};
		\node at (10, 2) (z2) {$Y_2$};
		\node  at (10, 1.3) {\vdots};
		\node at (10, 0.6) (z3) {$Y_n$};

		\draw[->,thick] (y-1)--(y0);
		\draw[->,thick] (y0)--(y1);
		\draw[->,thick] (y1)--(y2);
		\draw[->,thick] (y2)--(y3);

		\draw[->,thick] (w1)--(z1);
		\draw[->,thick] (w2)--(z2);
		\draw[->,thick] (w3)--(z3);

		\draw[->,thick] (7.4, 2.7)--(w1);
		\draw[->,thick] (7.4, 2)--(w2);
		\draw[->,thick] (7.4, 0.6)--(w3);

		\node at (4.5, -0.7) [align=center] (p1) {$(\widehat{U}_1,\dots,\widehat{U}_n)=(U_1,\dots,U_n)\mathbf{P}_n^{\ABS}, \quad
			(X_1,\dots,X_n)=(\widehat{U}_1,\dots,\widehat{U}_n) (\mathbf{G}_{n/2}^{\ABS} \otimes \mathbf{G}_2^{\polar})$ \\[0.4em]
		$(X_1,\dots,X_n)=(U_1,\dots,U_n)\mathbf{P}_n^{\ABS}(\mathbf{G}_{n/2}^{\ABS} \otimes \mathbf{G}_2^{\polar})=(U_1,\dots,U_n)\mathbf{G}_n^{\ABS}$};

		\node at (4, -3.2) [align=center] (nbc1)
		{Two sets of bit-channels\\
		$\{W_i^{(n),\ABS}:U_i\to U_1,\dots,U_{i-1},Y_1,\dots,Y_n\}_{i=1}^n$\\[0.1em]
		$\{\widehat{W}_i^{(n),\ABS}:\widehat{U}_i\to \widehat{U}_1,\dots,\widehat{U}_{i-1},Y_1,\dots,Y_n\}_{i=1}^n$ \\[0.6em]
		Two sets of adjacent-bits-channels\\
		$\{V_i^{(n),\ABS}:U_i,U_{i+1}\to U_1,\dots,U_{i-1},Y_1,\dots,Y_n\}_{i=1}^{n-1}$\\[0.1em]
		$\{\widehat{V}_i^{(n),\ABS}:\widehat{U}_i,\widehat{U}_{i+1}\to \widehat{U}_1,\dots,\widehat{U}_{i-1},Y_1,\dots,Y_n\}_{i=1}^{n-1}$};
	\end{tikzpicture}
	\caption{$U_1,\dots,U_n$ are $n=2^m$ i.i.d.
	Bernoulli-$1/2$ random variables.
	$(X_1,\dots,X_n)=(U_1,\dots,U_n) \mathbf{G}_n^{\ABS}$ is the codeword vector, and $(Y_1,\dots,Y_n)$ is the channel output vector.
	We view each Kronecker product with $\mathbf{G}_2^{\polar}$ as one layer of polar transform and view each multiplication with a permutation matrix as one layer of permutation.
	Then $\mathbf{G}_n^{\ABS}$ is obtained from $m$ layers of polar transforms and $m$ layers of permutations while $\mathbf{G}_{n/2}^{\ABS} \otimes \mathbf{G}_2^{\polar}$ is obtained from $m$ layers of polar transforms and $m-1$ layers of permutations.
	Therefore, $\{W_i^{(n),\ABS}\}_{i=1}^n$ and $\{V_i^{(n),\ABS}\}_{i=1}^{n-1}$ are the bit-channels and adjacent-bits-channels seen by the successive decoder after $m$ layers of polar transforms and $m$ layers of permutations.
	Similarly, $\{\widehat{W}_i^{(n),\ABS}\}_{i=1}^n$ and $\{\widehat{V}_i^{(n),\ABS}\}_{i=1}^{n-1}$ are the bit-channels and adjacent-bits-channels seen by the successive decoder after $m$ layers of polar transforms and $m-1$ layers of permutations.}
	\label{fig:abc_in_ABS}
\end{figure}

In order to derive the recursive relation between the adjacent-bits-channels for ABS polar codes, we need another new transform named as the Swapped-Double-Bits (SDB) polar transform in addition to the DB polar transform defined in \eqref{eq:DBpt}.
The details of the SDB polar transform are illustrated in Fig.~\ref{fig:SDBpt}.
In fact, the SDB polar transform is very similar to the DB polar transform.
By comparing Fig.~\ref{fig:DBpt_a} and Fig.~\ref{fig:SDBpt_a}, we can see that the only difference between these two transforms is the order of $U_2$ and $U_3$.
Given a $4$-ary-input channel $V:\{0,1\}^2\to\cY$, the transition probabilities of $V^{\blacktriangledown}:\{0,1\}^2\to\cY^2,V^{\blacklozenge}:\{0,1\}^2\to\{0,1\}\times\cY^2$, and $V^{\blacktriangle}:\{0,1\}^2\to\{0,1\}^2\times\cY^2$ in Fig.~\ref{fig:SDBpt} are given by
\begin{equation}
	\begin{aligned}
		 & V^{\blacktriangledown}(y_1,y_2|u_1,u_2)=\frac{1}{4}\sum_{u_3,u_4\in\{0,1\}} V(y_1|u_1+u_3,u_2+u_4) V(y_2|u_3,u_4) \\
		 & \hspace*{2.7in} \text{~for~} u_1,u_2\in\{0,1\} \text{~and~} y_1,y_2\in\cY,                                        \\
		 & V^{\blacklozenge}(u_1,y_1,y_2|u_2,u_3)=\frac{1}{4}\sum_{u_4\in\{0,1\}} V(y_1|u_1+u_3,u_2+u_4) V(y_2|u_3,u_4)      \\
		 & \hspace*{2.7in} \text{~for~} u_1,u_2,u_3\in\{0,1\} \text{~and~} y_1,y_2\in\cY,                                    \\
		 & V^{\blacktriangle}(u_1,u_2,y_1,y_2|u_3,u_4)=\frac{1}{4} V(y_1|u_1+u_3,u_2+u_4) V(y_2|u_3,u_4)                     \\
		 & \hspace*{2.7in} \text{~for~} u_1,u_2,u_3,u_4\in\{0,1\} \text{~and~} y_1,y_2\in\cY.
	\end{aligned}
\end{equation}

\begin{figure}[ht]
	\centering
	\begin{subfigure}{.53\linewidth}
		\centering
		\begin{tikzpicture}
			\draw
			node at (-1.5,10.5) [] (u1)  {$U_1$}
			node at (-1.5,9.5) [] (u2)  {$U_2$}
			node at (1.5,10.5) [XOR,scale=1.2] (x1) {}
			node at (4.1,10.5) (v1) {}
			node at (4.1,9.5) (v2) {}
			node at (4.5,10) (t1) {}
			node at (4.3,10) [vblock] (76)  {$V$}
			node at (6,10) [] (y1)  {$Y_1$};
			\draw[fill] (1.5, 9.5) circle (.6ex);

			\draw
			node at (-1.5,8.5) [] (u3)  {$U_3$}
			node at (-1.5,7.5) [] (u4)  {$U_4$}
			node at (1.5,8.5) [XOR,scale=1.2] (x3) {}
			node at (4.1,8.5) (v3) {}
			node at (4.1,7.5) (v4) {}
			node at (4.5,8) (t3) {}
			node at (4.3,8) [vblock] (83)  {$V$}
			node at (6,8) [] (y3)  {$Y_2$};
			\draw[fill] (1.5, 7.5) circle (.6ex);

			\draw[very thick,->](u1) -- node {}(x1);
			\draw[very thick,->](0.6, 9.5) -| node {}(x1);
			\draw[very thick,->](x1) -- (v1);
			\draw[very thick,->](u3) -- (-0.4, 8.5) -- (0.6, 9.5) -- (2.2,9.5) -- (3.3,8.5) -- (v3);
			\draw[very thick,->](t1) -- node {}(y1);

			\draw[very thick,->](x3) -- (2.2,8.5) -- (3.3,9.5) -- (v2);
			\draw[very thick,->](u4) -| node {}(x3);
			\draw[very thick,->](u2) -- (-0.4, 9.5) -- (0.6, 8.5) -- (x3);
			\draw[very thick,->](u4) -- (v4);
			\draw[very thick,->](t3) -- (y3);

		\end{tikzpicture}
		\caption{$U_1,U_2,U_3,U_4$ are i.i.d.
			Bernoulli-$1/2$ random variables.
			The channel $V:\{0,1\}^2\to \cY$ takes two bits as its inputs, i.e., $V$ has $4$-ary inputs.
			Under the successive decoder, we have the following three channels: (1) $V^{\blacktriangledown}:U_1,U_2\to Y_1,Y_2$; (2) $V^{\blacklozenge}:U_2,U_3\to U_1,Y_1,Y_2$; (3) $V^{\blacktriangle}:U_3,U_4\to U_1,U_2,Y_1,Y_2$.}
		\label{fig:SDBpt_a}
	\end{subfigure}
	\hfill
	\begin{subfigure}{.43\linewidth}
		\centering
		\begin{tikzpicture}
			\draw
			node at (0,1.5) [] (u1)  {$V$}
			node at (0,0) [] (u2)  {$V$}
			node at (1.4,1.5) [] (v1)  {}
			node at (1.4,0) [] (v2)  {}
			node at (3.5,2.2) [] (x1)  {}
			node at (3.5,0.75) [] (x2)  {}
			node at (3.5,-0.7) [] (x3)  {}
			node at (5.2,2.2) [] (y1)  {$V^{\blacktriangledown}$}
			node at (5.2,0.75) [] (y2)  {$V^{\blacklozenge}$}
			node at (5.2,-0.7) [] (y3)  {$V^{\blacktriangle}$}
			node at (2.45,0.75) [lgblock, align=center] {SDB polar\\Transform};
			\draw[very thick,->](u1) -- node {}(v1);
			\draw[very thick,->](u2) -- node {}(v2);
			\draw[very thick,->](x1) -- node[above] {``$\blacktriangledown$"}(y1);
			\draw[very thick,->](x2) -- node [above] {``$\blacklozenge$"} (y2);
			\draw[very thick,->](x3) -- node [above] {``$\blacktriangle$"} (y3);
		\end{tikzpicture}
		\caption{Two independent copies of $V$ are transformed into three channels $V^{\blacktriangledown},V^{\blacklozenge},V^{\blacktriangle}$.
		These three channels also have $4$-ary inputs.
		Note that the inputs of $V^{\blacktriangledown}$ and $V^{\blacklozenge}$ have one-bit overlap, and the inputs of $V^{\blacklozenge}$ and $V^{\blacktriangle}$ also have one-bit overlap.}
	\end{subfigure}
	\caption{The Swapped-Double-Bits (SDB) polar transform}
	\label{fig:SDBpt}
\end{figure}
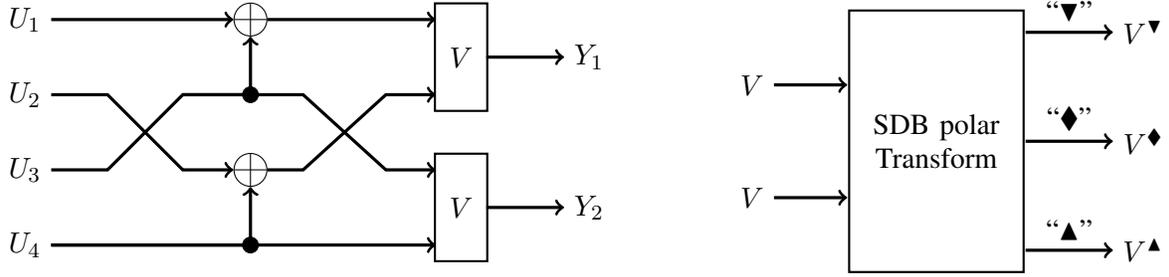

Recall that we use the set $\cI^{(n)}=\{i_1,i_2,\dots,i_s\}$ to represent the permutation matrix $\mathbf{P}_n^{\ABS}$ in \eqref{eq:pi}.
Moreover, we require that $i_1,i_2,\dots,i_s$ in the set $\cI^{(n)}$ satisfy the condition \eqref{eq:separated} because otherwise there does not exist a recursive relation between the adjacent-bits-channels $\{V_i^{(n),\ABS}\}_{i=1}^{n-1}$ and $\{V_i^{(n/2),\ABS}\}_{i=1}^{n/2-1}$.
We will give a detailed explanation about this later in Section~\ref{sect:necess}.
Here we point out another property of the elements $i_1,i_2,\dots,i_s$ in $\cI^{(n)}$: they must all be even numbers.
To see this, let us go back to the setting in Fig.~\ref{fig:abc_in_ABS}.
The role of $\cI^{(n)}$ is to decide which pairs of adjacent bits to swap in the vector $(\widehat{U}_1,\widehat{U}_2,\dots,\widehat{U}_n)$ defined in Fig.~\ref{fig:abc_in_ABS}.
According to the discussion in Section~\ref{sect:reason}, we swap the adjacent bits $\widehat{U}_i$ and $\widehat{U}_{i+1}$ only if they are unordered, i.e., if $\widehat{U}_i$ is more reliable than $\widehat{U}_{i+1}$ under the successive decoder.
In other words, we swap the adjacent bits $\widehat{U}_i$ and $\widehat{U}_{i+1}$ only if $I(\widehat{W}_i^{(n),\ABS})\ge I(\widehat{W}_{i+1}^{(n),\ABS})$, where the bit-channels $\widehat{W}_i^{(n),\ABS}$ and $\widehat{W}_{i+1}^{(n),\ABS}$ are also defined in Fig.~\ref{fig:abc_in_ABS}.
Since $\{\widehat{W}_i^{(n),\ABS}\}_{i=1}^n$ are obtained from the $2\times 2$ basic polar transform of $\{W_i^{(n/2),\ABS}\}_{i=1}^{n/2}$, they satisfy the following relation:
$$
	\widehat{W}_{2i-1}^{(n),\ABS} = (W_i^{(n/2),\ABS})^- \text{~~and~~} \widehat{W}_{2i}^{(n),\ABS} = (W_i^{(n/2),\ABS})^+  \text{~~for~} 1\le i\le n/2.
$$
Therefore,
$$
	I(\widehat{W}_{2i-1}^{(n),\ABS}) \le I(W_i^{(n/2),\ABS}) \le I(\widehat{W}_{2i}^{(n),\ABS}),
$$
so we should not swap $\widehat{U}_{2i-1}$ and $\widehat{U}_{2i}$ for any $1\le i\le n/2$.
Thus we conclude that the set $\cI^{(n)}$ in \eqref{eq:pi} only contains even numbers. Therefore, the elements of $\cI^{(n)}$ can be written as $\cI^{(n)}=\{2j_1,2j_2,\dots,2j_s\}$, and the condition \eqref{eq:separated} becomes
\begin{equation} \label{eq:separated_even}
	j_2\ge j_1+2, \quad
	j_3\ge j_2+2, \quad
	j_4\ge j_3+2, \quad
	\dots, \quad
	j_s\ge j_{s-1}+2.
\end{equation}
Now we are ready to state the recursive relation between $\{V_i^{(n),\ABS}\}_{i=1}^{n-1}$ and $\{V_i^{(n/2),\ABS}\}_{i=1}^{n/2-1}$.

\begin{lemma} \label{lemma:recur_ABS}
	Let $n\ge 4$.
	We write $\mathbf{P}_n^{\ABS}$ in the form of \eqref{eq:pi} and require that $\cI^{(n)}=\{2j_1,2j_2,\dots,2j_s\}$ satisfies \eqref{eq:separated_even}.
	For $1\le i\le n/2 -1$, we have the following results:

	\noindent
	{\em \bf Case i)} If $2i\in\cI^{(n)}$, then
	$$
		V_{2i-1}^{(n),\ABS} = (V_i^{(n/2),\ABS})^\blacktriangledown,
		\quad
		V_{2i}^{(n),\ABS} = (V_i^{(n/2),\ABS})^\blacklozenge, \quad
		V_{2i+1}^{(n),\ABS} = (V_i^{(n/2),\ABS})^\blacktriangle.
	$$

	\noindent
	{\em \bf Case ii)} If $2(i-1)\in\cI^{(n)}$ and $2(i+1)\in\cI^{(n)}$, then
	$$
		V_{2i}^{(n),\ABS} = (V_i^{(n/2),\ABS})^\lozenge.
	$$

	\noindent
	{\em \bf Case iii)} If $2(i-1)\in\cI^{(n)}$ and $2(i+1)\notin\cI^{(n)}$, then
	$$
		V_{2i}^{(n),\ABS} = (V_i^{(n/2),\ABS})^\lozenge, \quad
		V_{2i+1}^{(n),\ABS} = (V_i^{(n/2),\ABS})^\vartriangle.
	$$

	\noindent
	{\em \bf Case iv)} If $2(i-1)\notin\cI^{(n)}$ and $2(i+1)\in\cI^{(n)}$, then
	$$
		V_{2i-1}^{(n),\ABS} = (V_i^{(n/2),\ABS})^\triangledown,
		\quad
		V_{2i}^{(n),\ABS} = (V_i^{(n/2),\ABS})^\lozenge.
	$$

	\noindent
	{\em \bf Case v)} If $2(i-1)\notin\cI^{(n)}$, $2i\notin\cI^{(n)}$ and $2(i+1)\notin\cI^{(n)}$, then
	$$
		V_{2i-1}^{(n),\ABS} = (V_i^{(n/2),\ABS})^\triangledown,
		\quad
		V_{2i}^{(n),\ABS} = (V_i^{(n/2),\ABS})^\lozenge, \quad
		V_{2i+1}^{(n),\ABS} = (V_i^{(n/2),\ABS})^\vartriangle.
	$$
\end{lemma}
Note that in a previous arXiv version and the ISIT version \cite{Li2022ISIT} of this paper, the statement of this lemma was not complete. In the previous versions, {\em \bf  Case ii)} was missing, and the conditions in {\em \bf Case iii)} and {\em \bf Case iv)} were incomplete.

The proof of Lemma~\ref{lemma:recur_ABS} is omitted because it is essentially the same as the proof of Lemma~\ref{lemma:recur_ST_DB}.
Here we point out one difference between Lemma~\ref{lemma:recur_ST_DB} and Lemma~\ref{lemma:recur_ABS}.
Lemma~\ref{lemma:recur_ST_DB} tells us that $V_{2i+1}^{(n)}$ can be recursively calculated in two different ways for every $1\le i\le n/2-2$; see \eqref{eq:2w}.
However, for $i\in\{j_1-1,j_2-1,\dots,j_s-1\}\cup\{j_1,j_2,\dots,j_s\}$, there is only one way to calculate $V_{2i+1}^{(n),\ABS}$ recursively.
More precisely, if $i\in\{j_1-1,j_2-1,\dots,j_s-1\}$, then $V_{2i+1}^{(n),\ABS}$ can only be calculated from $V_{2i+1}^{(n),\ABS}=(V_{i+1}^{(n/2),\ABS})^\blacktriangledown$, and the relation $V_{2i+1}^{(n),\ABS} = (V_i^{(n/2),\ABS})^\vartriangle$ does {\em not} hold.
Similarly, if $i\in\{j_1,j_2,\dots,j_s\}$, then $V_{2i+1}^{(n),\ABS}$ can only be calculated from $V_{2i+1}^{(n),\ABS} = (V_i^{(n/2),\ABS})^\blacktriangle$, and the relation $V_{2i+1}^{(n),\ABS}=(V_{i+1}^{(n/2),\ABS})^\triangledown$ does {\em not} hold.

Since we require $n\ge 4$ in Lemma~\ref{lemma:recur_ABS}, the starting point of the recursive relation in Lemma~\ref{lemma:recur_ABS} is $V_1^{(2),\ABS}$.
It is easy to see that the permutation matrix $\mathbf{P}_2^{\ABS}$ is the identity matrix.
Therefore, given a BMS channel $W$, the transition probability of $V_1^{(2),\ABS}$ is given by
\begin{equation} \label{eq:v_abs_init}
	V_1^{(2),\ABS}(y_1,y_2|u_1,u_2)=W(y_1|u_1+u_2) W(y_2|u_2).
\end{equation}
Note that this is the same as \eqref{eq:v_init} for standard polar codes.

After obtaining the transition probabilities of the adjacent-bits-channels $\{V_i^{(n),\ABS}\}_{i=1}^{n-1}$, we can use \eqref{eq:v_to_w} to calculate the transition probabilities of the bit-channels $\{W_i^{(n),\ABS}\}_{i=1}^n$.
We only need to replace $W_i^{(n)},W_{i+1}^{(n)},V_i^{(n)}$ in \eqref{eq:v_to_w} with $W_i^{(n),\ABS},W_{i+1}^{(n),\ABS},V_i^{(n),\ABS}$.
Once the transition probabilities of $\{W_i^{(n),\ABS}\}_{i=1}^n$ are known, we are able to determine which bits are information bits and which bits are frozen bits.

\subsection{Constructing the permutation matrices $\mathbf{P}_2^{\ABS},\mathbf{P}_4^{\ABS},\mathbf{P}_8^{\ABS},\dots,\mathbf{P}_n^{\ABS}$ in \eqref{eq:GnABS}}  \label{sect:cons_PnABS}

We construct the permutation matrices in \eqref{eq:GnABS} one by one, starting from $\mathbf{P}_2^{\ABS}$.
Therefore, the matrices $\mathbf{P}_2^{\ABS},\mathbf{P}_4^{\ABS},\mathbf{P}_8^{\ABS},\dots,\mathbf{P}_{n/2}^{\ABS}$ are already known when we construct $\mathbf{P}_n^{\ABS}$.
The method described in Section~\ref{sect:SDB} allows us to calculate the transition probabilities of the adjacent-bits-channels $\{V_i^{(n/2),\ABS}\}_{i=1}^{n/2-1}$ from $\mathbf{P}_2^{\ABS},\mathbf{P}_4^{\ABS},\mathbf{P}_8^{\ABS},\dots,\mathbf{P}_{n/2}^{\ABS}$.
As a consequence, we know the transition probabilities of $\{V_i^{(n/2),\ABS}\}_{i=1}^{n/2-1}$ when constructing $\mathbf{P}_n^{\ABS}$.
Since the set $\cI^{(n)}=\{2j_1,2j_2,\dots,2j_s\}$ in \eqref{eq:pi} uniquely determines $\mathbf{P}_n^{\ABS}$, constructing $\mathbf{P}_n^{\ABS}$ is further equivalent to constructing the set $\cS^*=\{j_1,j_2,\dots,j_s\}$, where the elements $j_1,j_2,\dots,j_s$ satisfy the condition \eqref{eq:separated_even}.

Before presenting how to construct the set $\cS^*$, let us introduce some notation.
Suppose that $V:\{0,1\}^2\to \cY$ is an adjacent-bits-channel with $4$-ary inputs.
Define two bit-channels $V_{\first}:\{0,1\}\to\cY$ and $V_{\second}:\{0,1\}\to \{0,1\}\times\cY$ as
\begin{align*}
	V_{\first}(y|u_1)=\frac{1}{2}\sum_{u_2\in\{0,1\}} V(y|u_1,u_2) \text{~~and~~}
	V_{\second}(y,u_1|u_2)= \frac{1}{2} V(y|u_1,u_2).
\end{align*}
Comparing this with \eqref{eq:v_to_w}, we can see that if $V$ is $V_i^{(n)}$, then $V_{\first}$ is simply $W_i^{(n)}$, and $V_{\second}$ is $W_{i+1}^{(n)}$.
Similarly, if $V$ is $V_i^{(n),\ABS}$, then $V_{\first}$ is simply $W_i^{(n),\ABS}$, and $V_{\second}$ is $W_{i+1}^{(n),\ABS}$.
Next we define
\begin{align*}
	 & I_{\first}(V):= I(V_{\first}) \text{~~and~~} I_{\second}(V) := I(V_{\second}), \\
	 & g(V):= I_{\first}(V) (1-I_{\first}(V)) + I_{\second}(V) (1-I_{\second}(V)).
\end{align*}
The function $g(V)$ measures the polarization level of the two bit-channels induced by $V$.
In particular, $g(V)\approx 0$ means that the capacity of both bit-channels is very close to either $0$ or $1$.
Finally, for $1\le i\le n/2-1$, we define
$$
	\texttt{Score}(i) := g\big( (V_i^{(n/2),\ABS})^\lozenge \big) - g\big( (V_i^{(n/2),\ABS})^\blacklozenge \big).
$$
The interpretation of $\texttt{Score}(i)$ is as follows:
According to Lemma~\ref{lemma:recur_ABS}, if $i\in\cS^*$, then $V_{2i}^{(n),\ABS} = (V_i^{(n/2),\ABS})^\blacklozenge$; if $i\notin\cS^*$, then $V_{2i}^{(n),\ABS} = (V_i^{(n/2),\ABS})^\lozenge$.
Therefore, $g\big( (V_i^{(n/2),\ABS})^\blacklozenge \big)$ measures the polarization level of the two bit-channels $W_{2i}^{(n),\ABS}$ and $W_{2i+1}^{(n),\ABS}$ when we include $i$ in the set $\cS^*$.
Similarly, $g\big( (V_i^{(n/2),\ABS})^\lozenge \big)$ measures the polarization level of the two bit-channels $W_{2i}^{(n),\ABS}$ and $W_{2i+1}^{(n),\ABS}$ when we do not include $i$ in the set $\cS^*$.
If $\texttt{Score}(i)>0$, then including $i$ in the set $\cS^*$ accelerates polarization.
If $\texttt{Score}(i)<0$, then including $i$ in the set $\cS^*$ slows down polarization, and in this case we should not include $i$ in $\cS^*$.

If we ignore the condition \eqref{eq:separated_even}, then we can simply choose the set $\cS^*$ to be $\cS^*=\{i:\texttt{Score}(i)>0\}$.
However, as we will see in Section~\ref{sect:necess}, the condition \eqref{eq:separated_even} is crucial for us to calculate the transition probabilities of the adjacent-bits-channels, so it must be satisfied.
As a consequence, we need to find a set $\cS^*\subseteq\{1,2,\dots,n/2-1\}$ to maximize $\sum_{i\in\cS^*}\texttt{Score}(i)$ under the constraint that the distance between any two distinct elements of $\cS^*$ must be at least $2$.
In other words, we need to solve the following optimization problem:
\begin{equation} \label{eq:argmaxS}
	\begin{aligned}
		\cS^*= & \argmax_{\cS\subseteq\{1,2,\dots,n/2-1\}} \sum_{i\in\cS}\texttt{Score}(i)                         \\
		       & \text{subject to: } |i_1-i_2|\ge 2 \text{~for all~} i_1,i_2\in\cS \text{~such that~} i_1\neq i_2.
	\end{aligned}
\end{equation}
This problem can be solved using a dynamic programming method.
For $1\le j\le n/2-1$, define
\begin{align*}
	\cS_j^*= & \argmax_{\cS\subseteq\{1,2,\dots,j\}} \sum_{i\in\cS}\texttt{Score}(i)                             \\
	         & \text{subject to: } |i_1-i_2|\ge 2 \text{~for all~} i_1,i_2\in\cS \text{~such that~} i_1\neq i_2, \\
	M_j=     & \max_{\cS\subseteq\{1,2,\dots,j\}} \sum_{i\in\cS}\texttt{Score}(i)                                \\
	         & \text{subject to: } |i_1-i_2|\ge 2 \text{~for all~} i_1,i_2\in\cS \text{~such that~} i_1\neq i_2.
\end{align*}
By definition, we can see that $M_1\le M_2\le M_3\le \dots \le M_{n/2-1}$.
The sets $\cS_1^*,\cS_2^*$ and the maximum values $M_1,M_2$ can be calculated as follows: If $\texttt{Score}(1)>0$, then $\cS_1^*=\{1\}$ and $M_1=\texttt{Score}(1)$.
If $\texttt{Score}(1)\le 0$, then $\cS_1^*=\emptyset$ and $M_1=0$.
If $\texttt{Score}(2)>M_1$, then $\cS_2^*=\{2\}$ and $M_2=\texttt{Score}(2)$.
If $\texttt{Score}(2)\le M_1$, then $\cS_2^*=\cS_1^*$ and $M_2=M_1$.
For $j\ge 3$, the set $\cS_j^*$ and the maximum value $M_j$ can be calculated recursively as follows: If $\texttt{Score}(j)+M_{j-2}>M_{j-1}$, then $\cS_j^*=\cS_{j-2}^*\cup\{j\}$ and $M_j=\texttt{Score}(j)+M_{j-2}$.
If $\texttt{Score}(j)+M_{j-2}\le M_{j-1}$, then $\cS_j^*=\cS_{j-1}^*$ and $M_j=M_{j-1}$.
This dynamic programming algorithm allows us to calculate $\cS_j^*$ for every $1\le j\le n/2-1$.
In particular, we are able to calculate $\cS_{n/2-1}^*=\cS^*$, which is the set we want to construct.
Once we know the set $\cS^*=\{j_1,j_2,\dots,j_s\}$, we can immediately write out the set $\cI^{(n)}=\{2j_1,2j_2,\dots,2j_s\}$ and obtain the corresponding permutation matrix $\mathbf{P}_n^{\ABS}$ according to \eqref{eq:pi}.

As a final remark, we note that $\mathbf{P}_2^{\ABS}$ is always the identity matrix.
However, for $n\ge 4$, the permutation matrix $\mathbf{P}_n^{\ABS}$ depends on the underlying BMS channel $W$.

\subsection{Quantization of the output alphabet} \label{sect:quantization}

\begin{algorithm}[H]
	\DontPrintSemicolon
	\caption{\texttt{QuantizeChannel}$(\mu,V)$}
	\label{algo:quantization}
	\KwIn{an upper bound $\mu$ on the output alphabet size after quantization; an adjacent-bits-channel $V$ with outputs $y_1,y_2,\dots,y_M$}

	\vspace*{0.05in}
	\KwOut{quantized channel $\widetilde{V}$ with outputs $\{\tilde{y}_{i_1,i_2,i_3}: 0\le i_1,i_2,i_3\le b\}$}

	\vspace*{0.05in}
	\If{$M\le \mu$}
	{
		\vspace*{0.05in}
		Set $\widetilde{V}$ to be the same as $V$
	}
	\Else
	{
	$b\gets \lfloor\mu^{1/3}\rfloor - 1$

	\vspace*{0.05in}
	Set $\widetilde{V}(\tilde{y}_{i_1,i_2,i_3}|(u_1,u_2))=0$ for all $0\le i_1,i_2,i_3\le b$ and all $u_1,u_2\in\{0,1\}$

	\Comment{Initialize all the transition probabilities of $\widetilde{V}$ as $0$}

	\For{$j=1,2,\dots,M$}
	{
	$sum\gets V(y_j|(0,0))+V(y_j|(0,1))+V(y_j|(1,0))+V(y_j|(1,1))$

	\vspace*{0.05in}
	$p_1\gets \frac{V(y_j|(0,0))}{sum}$, \quad
	$p_2\gets \frac{V(y_j|(0,1))}{sum}$, \quad
	$p_3\gets \frac{V(y_j|(1,0))}{sum}$

	\Comment{Calculate the posterior probability of $y_j$}

	$i_1\gets \lfloor b p_1 \rfloor$, \quad
	$i_2\gets \lfloor b p_2 \rfloor$, \quad
	$i_3\gets \lfloor b p_3 \rfloor$

	\vspace*{0.05in}
	$\widetilde{V}(\tilde{y}_{i_1,i_2,i_3}|(u_1,u_2))\gets \widetilde{V}(\tilde{y}_{i_1,i_2,i_3}|(u_1,u_2)) + V(y_j|(u_1,u_2))$ for all $u_1,u_2\in\{0,1\}$

		\Comment{Merge $y_j$ into $\tilde{y}_{i_1,i_2,i_3}$}
		}
		}

		\Return $\widetilde{V}$

\end{algorithm}

An important step in the construction of standard polar codes is to quantize the output alphabets of the bit-channels $\{W_i^{(n)}\}_{i=1}^n$ because the output alphabet size grows exponentially with the code length $n$.
The most widely used quantization method for binary-input standard polar codes was given in \cite{Tal13}, where the main idea is to merge output symbols with similar posterior distributions using a greedy algorithm.
This greedy algorithm was later generalized to construct polar codes with non-binary input alphabets \cite{TSV12,Pereg17,Gulcu18}.
The time complexity of the greedy quantization algorithm is $O(\mu^2\log \mu)$, where $\mu$ is the maximum size of the output alphabet after quantization.
Since there are $2n-1$ bit-channels we need to quantize in the code construction procedure, the overall time complexity of standard polar code construction is $O(n\mu^2\log \mu)$.

In the ABS polar code construction, the output alphabet size of the adjacent-bits-channels $\{V_i^{(n)}\}_{i=1}^{n-1}$ also grows exponentially with $n$, and the quantization operations are also needed.
Since the adjacent-bits-channels have $4$-ary inputs, we can simply use the greedy quantization algorithms proposed in \cite{TSV12,Pereg17,Gulcu18} for polar codes with non-binary inputs.
However, in practical implementations, we found that these greedy algorithms for non-binary inputs usually involve implicit large constants in their time complexity.
Therefore, we propose a new quantization algorithm to merge the output symbols of the adjacent-bits-channels $\{V_i^{(n)}\}_{i=1}^{n-1}$.
	The time complexity of our new quantization algorithm is $O(\mu^2)$.
Since there are $\Theta(n)$ adjacent-bits-channels we need to quantize in the ABS polar code construction, its overall time complexity is $O(n\mu^2)$.

Our new quantization algorithm works as follows.
Given an upper bound $\mu$ on the output alphabet size after quantization, we define $b=\lfloor\mu^{1/3}\rfloor - 1$.
For an adjacent-bits-channel $V$, we write its $4$ inputs as $(0,0),(0,1),(1,0),(1,1)$, and we write its outputs as $y_1,y_2,\dots,y_M$, where $M$ is the output alphabet size of $V$.
We use $\widetilde{V}$ to denote the channel after output quantization.
The $4$ inputs of $\widetilde{V}$ are the same as the original channel $V$, and the outputs of $\widetilde{V}$ are written as $\{\tilde{y}_{i_1,i_2,i_3}: 0\le i_1,i_2,i_3\le b\}$.
Clearly, the output alphabet size of $\widetilde{V}$ is no larger than $\mu$.
With the above notation in mind, we present our quantization algorithm in Algorithm~\ref{algo:quantization}.
In our implementation, we pick $\mu=250000$.

\subsection{Summary of the ABS polar code construction} \label{sect:summary_cons}
In Section~\ref{sect:SDB}, we showed how to calculate the transition probabilities of the adjacent-bits-channels $\{V_i^{(n),\ABS}\}_{i=1}^{n-1}$ when the permutation matrices $\mathbf{P}_2^{\ABS},\mathbf{P}_4^{\ABS},\mathbf{P}_8^{\ABS},\dots,\mathbf{P}_n^{\ABS}$ in \eqref{eq:GnABS} are known.
In Section~\ref{sect:cons_PnABS}, we showed how to construct the permutation matrix $\mathbf{P}_n^{\ABS}$ when the transition probabilities of $\{V_i^{(n/2),\ABS}\}_{i=1}^{n/2-1}$ are available.
In Section~\ref{sect:quantization}, we proposed Algorithm~\ref{algo:quantization} to quantize the output alphabets of the adjacent-bits-channels.
Now we are in a position to put everything together and present the code construction algorithm for ABS polar codes in Algorithm~\ref{algo:ABS_Construction}.

\begin{algorithm}[H]
	\DontPrintSemicolon
	\caption{\texttt{ABSConstruct}$(n,k,W)$}
	\label{algo:ABS_Construction}
	\KwIn{code length $n=2^m\ge 4$, code dimension $k$, and the BMS channel $W$}
	\KwOut{the permutation matrices $\mathbf{P}_2^{\ABS},\mathbf{P}_4^{\ABS},\mathbf{P}_8^{\ABS},\dots,\mathbf{P}_n^{\ABS}$, and the index set $\cA$ of the information bits}

	Quantize the output alphabet of $W$ using the method in \cite{Tal13}
	\Comment{This step is needed when the output alphabet size of $W$ is very large, e.g., when $W$ has a continuous output alphabet.}

	Set $\mathbf{P}_2^{\ABS}$ to be the identity matrix

	Calculate the transition probability of $V_1^{(2),\ABS}$ from $W$ using \eqref{eq:v_abs_init}

	Quantize the output alphabet of $V_1^{(2),\ABS}$ using Algorithm~\ref{algo:quantization}

	\For{$n_0=4,8,16,\dots,n$}
	{
	Construct $\mathbf{P}_{n_0}^{\ABS}$ from $\{V_i^{(n_0/2),\ABS}\}_{i=1}^{n_0/2-1}$ using the method in Section~\ref{sect:cons_PnABS}

	Calculate the transition probabilities of $\{V_i^{(n_0),\ABS}\}_{i=1}^{n_0-1}$ from $\mathbf{P}_{n_0}^{\ABS}$ and $\{V_i^{(n_0/2),\ABS}\}_{i=1}^{n_0/2-1}$ using Lemma~\ref{lemma:recur_ABS}

	Quantize the output alphabets of $\{V_i^{(n_0),\ABS}\}_{i=1}^{n_0-1}$ using Algorithm~\ref{algo:quantization}
	}

	Calculate the transition probabilities of $\{W_i^{(n),\ABS}\}_{i=1}^n$ from the transition probabilities of $\{V_i^{(n),\ABS}\}_{i=1}^{n-1}$.

	Sort the capacity of the bit-channels $\{W_i^{(n),\ABS}\}_{i=1}^n$ to obtain $I(W_{i_1}^{(n),\ABS})\ge I(W_{i_2}^{(n),\ABS})\ge \dots\ge I(W_{i_n}^{(n),\ABS})$, where $\{i_1,i_2,\dots,i_n\}$ is a permutation of $\{1,2,\dots,n\}$

	$\cA\gets\{i_1,i_2,\dots,i_k\}$

	\Return $\mathbf{P}_2^{\ABS},\mathbf{P}_4^{\ABS},\mathbf{P}_8^{\ABS},\dots,\mathbf{P}_n^{\ABS},\cA$

\end{algorithm}

\subsection{Necessity of the condition \eqref{eq:separated_even}} \label{sect:necess}

The condition \eqref{eq:separated_even} is necessary for us to derive a recursive relation between $\{V_i^{(n),\ABS}\}_{i=1}^{n-1}$ and $\{V_i^{(n/2),\ABS}\}_{i=1}^{n/2-1}$. In order to prove this claim, we introduce some notation. Instead of $(U_1,U_2,\dots,U_n)$, now we use $(U_1^{(n)},U_2^{(n)},\dots,U_n^{(n)})$ to denote the message vector. We add the superscript $(n)$ in the notation to distinguish between random variables in different layers. Define
$$
	(\widehat{U}_1^{(n)},\widehat{U}_2^{(n)},\dots,\widehat{U}_n^{(n)}) = (U_1^{(n)},U_2^{(n)},\dots,U_n^{(n)}) \mathbf{P}_n^{\ABS} .
$$
We further define random vectors $(U_{1,1}^{(n/2)},U_{2,1}^{(n/2)},\dots,U_{n/2,1}^{(n/2)})$ and $(U_{1,2}^{(n/2)},U_{2,2}^{(n/2)},\dots,U_{n/2,2}^{(n/2)})$ as follows:
$$
	U_{i,1}^{(n/2)}=\widehat{U}_{2i-1}^{(n)} + \widehat{U}_{2i}^{(n)} , \quad
	U_{i,2}^{(n/2)}=\widehat{U}_{2i}^{(n)} ,
$$
i.e., the vectors $(U_{1,1}^{(n/2)},U_{2,1}^{(n/2)},\dots,U_{n/2,1}^{(n/2)})$ and $(U_{1,2}^{(n/2)},U_{2,2}^{(n/2)},\dots,U_{n/2,2}^{(n/2)})$ are obtained from applying one layer of polar transform to $(\widehat{U}_1^{(n)},\widehat{U}_2^{(n)},\dots,\widehat{U}_n^{(n)})$. By definition, $V_i^{(n),\ABS}$ gives us the conditional distribution of $(U_i^{(n)},U_{i+1}^{(n)})$ given the channel outputs and the previous message bits; $V_i^{(n/2),\ABS}$ gives us the conditional distribution of $(U_{i,1}^{(n/2)},U_{i+1,1}^{(n/2)})$ and the conditional distribution of $(U_{i,2}^{(n/2)},U_{i+1,2}^{(n/2)})$ given the channel outputs and the previous message bits. Therefore, deriving a recursive relation between $\{V_i^{(n),\ABS}\}_{i=1}^{n-1}$ and $\{V_i^{(n/2),\ABS}\}_{i=1}^{n/2-1}$ is equivalent to the following task: Suppose that we know the joint distribution\footnote{More precisely, this should be the conditional distribution of $(U_{i,j}^{(n/2)},U_{i+1,j}^{(n/2)})$ given the channel outputs and the previous message bits. Similarly, the joint distribution of $(U_i^{(n)},U_{i+1}^{(n)})$ in the next sentence also refers to the conditional distribution.} of $(U_{i,j}^{(n/2)},U_{i+1,j}^{(n/2)})$ for all $1\le i\le n/2-1$ and $j\in\{1,2\}$. The task is to calculate the joint distribution of $(U_i^{(n)},U_{i+1}^{(n)})$ for all $1\le i\le n-1$. We will show that it is not possible to accomplish this task without the condition \eqref{eq:separated_even}.

Suppose that the condition \eqref{eq:separated_even} does not hold. Then there exists an integer $i$ such that we swap the adjacent bits $\widehat{U}_{2i}^{(n)}$ and $\widehat{U}_{2i+1}^{(n)}$, and we also swap $\widehat{U}_{2i+2}^{(n)}$ and $\widehat{U}_{2i+3}^{(n)}$; see Fig.~\ref{fig:nosepar} for an illustration. According to our assumption, we know the joint distribution of $(U_{i,1}^{(n/2)},U_{i+1,1}^{(n/2)})$ and the joint distribution of $(U_{i,2}^{(n/2)},U_{i+1,2}^{(n/2)})$. Moreover, $(U_{i,1}^{(n/2)},U_{i+1,1}^{(n/2)})$ and $(U_{i,2}^{(n/2)},U_{i+1,2}^{(n/2)})$ are independent. Therefore, we know the joint distribution of $(U_{i,1}^{(n/2)},U_{i,2}^{(n/2)},U_{i+1,1}^{(n/2)},U_{i+1,2}^{(n/2)})$. Since there is a one-to-one mapping between $(\widehat{U}_{2i-1}^{(n)},\widehat{U}_{2i}^{(n)},\widehat{U}_{2i+1}^{(n)},\widehat{U}_{2i+2}^{(n)})$ and $(U_{i,1}^{(n/2)},U_{i,2}^{(n/2)},U_{i+1,1}^{(n/2)},U_{i+1,2}^{(n/2)})$, we also know the distribution of $(\widehat{U}_{2i-1}^{(n)},\widehat{U}_{2i}^{(n)},\widehat{U}_{2i+1}^{(n)},\widehat{U}_{2i+2}^{(n)})$. Since $(U_{2i-1}^{(n)},U_{2i}^{(n)},U_{2i+1}^{(n)})$ is a function of $(\widehat{U}_{2i-1}^{(n)},\widehat{U}_{2i}^{(n)},\widehat{U}_{2i+1}^{(n)})$, we are able to calculate the joint distribution of $(U_{2i-1}^{(n)},U_{2i}^{(n)})$ and the joint distribution of $(U_{2i}^{(n)},U_{2i+1}^{(n)})$. Using a similar argument, we can show that we are able to calculate the joint distribution of $(U_{2i+2}^{(n)},U_{2i+3}^{(n)})$ and the joint distribution of $(U_{2i+3}^{(n)},U_{2i+4}^{(n)})$. The only problem is that we are not able to calculate the joint distribution of $(U_{2i+1}^{(n)},U_{2i+2}^{(n)})$. By definition,
$$
	U_{2i+1}^{(n)}=\widehat{U}_{2i}^{(n)}=U_{i,2}^{(n/2)}, \qquad
	U_{2i+2}^{(n)}=\widehat{U}_{2i+3}^{(n)}=U_{i+2,1}^{(n/2)}+U_{i+2,2}^{(n/2)} .
$$
Therefore, our task is to calculate the joint distribution of $(U_{i,2}^{(n/2)}, U_{i+2,1}^{(n/2)}+U_{i+2,2}^{(n/2)})$. Since the two random vectors $(U_{1,1}^{(n/2)},U_{2,1}^{(n/2)},\dots,U_{n/2,1}^{(n/2)})$ and $(U_{1,2}^{(n/2)},U_{2,2}^{(n/2)},\dots,U_{n/2,2}^{(n/2)})$ are independent, this further requires us to know the joint distribution of $(U_{i,2}^{(n/2)}, U_{i+2,2}^{(n/2)})$, which is not available. Therefore, we are not able to calculate the joint distribution of $(U_{2i+1}^{(n)},U_{2i+2}^{(n)})$. This proves the necessity of \eqref{eq:separated_even}.

\begin{figure}[ht]
	\centering
	\includegraphics[scale = 1.0]{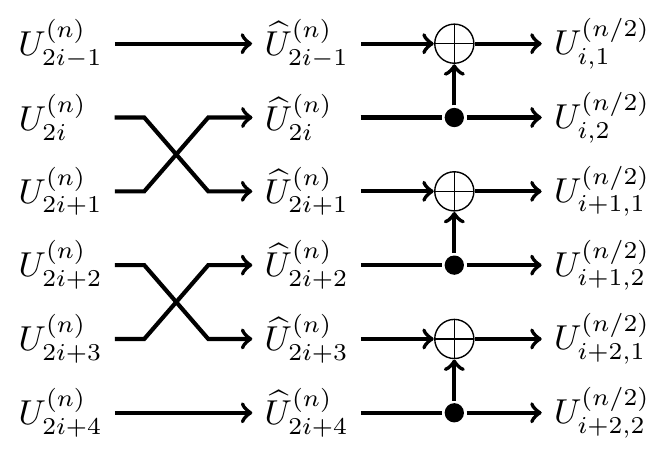}
	\caption{Swap the adjacent bits $\widehat{U}_{2i}$ and $\widehat{U}_{2i+1}$. Also swap $\widehat{U}_{2i+2}$ and $\widehat{U}_{2i+3}$.}
	\label{fig:nosepar}
\end{figure}

\section{The encoding algorithm for ABS polar codes} \label{sect:encoding}

In this section, we present the encoding algorithm of ABS polar codes.
Suppose that we have constructed an $(n,k)$ ABS polar code with permutation matrices $\mathbf{P}_2^{\ABS},\mathbf{P}_4^{\ABS},\mathbf{P}_8^{\ABS},\dots,\mathbf{P}_n^{\ABS}$ and the index set $\cA=\{i_1,i_2,\dots,i_k\}$ of the information bits.
We present the encoding algorithm of this code in Algorithm~\ref{algo:ABS_encoding} below.

\begin{algorithm}[H]
	\DontPrintSemicolon
	\caption{\texttt{Encode}$((m_1,m_2,\dots,m_k))$}
	\label{algo:ABS_encoding}
	\KwIn{the message vector $(m_1,m_2,\dots,m_k)\in\{0,1\}^k$}
	\KwOut{the codeword $(c_1,c_2,\dots,c_n)\in\{0,1\}^n$, where $n=2^m$ is the code length}

	Initialize $(c_1,c_2,\dots,c_n)$ as the all-zero vector

	$(c_{i_1},c_{i_2},\dots,c_{i_k})\gets (m_1,m_2,\dots,m_k)$

	\Comment{Recall that $i_1,i_2,\dots,i_k$ are the indices of the information bits.}

	\For{$i=0,1,2,3,\dots,m-1$}
	{
	$t\gets 2^i$

	$n_0\gets 2^{m-i}$

	\For{$h=1,2,3,\dots,t$}
	{
	$(c_h,c_{h+t},c_{h+2t},c_{h+3t},\dots,c_{h+(n_0-1)t})\gets (c_h,c_{h+t},c_{h+2t},c_{h+3t},\dots,c_{h+(n_0-1)t}) \mathbf{P}_{n_0}^{\ABS}$

	\Comment{Line 8 is the only difference between the encoding algorithms for ABS polar codes and standard polar codes}

	\For{$j=0,1,2,3,\dots,n_0/2-1$}
	{
	$c_{h+2jt} \gets c_{h+2jt}+c_{h+2jt+t}$

	\Comment{The addition between $c_{h+2jt}$ and $c_{h+2jt+t}$ is over the binary field}
	}
	}
	}

	\Return  $(c_1,c_2,\dots,c_n)$
\end{algorithm}

\begin{figure*}
	\centering
	\includegraphics[scale = 0.75]{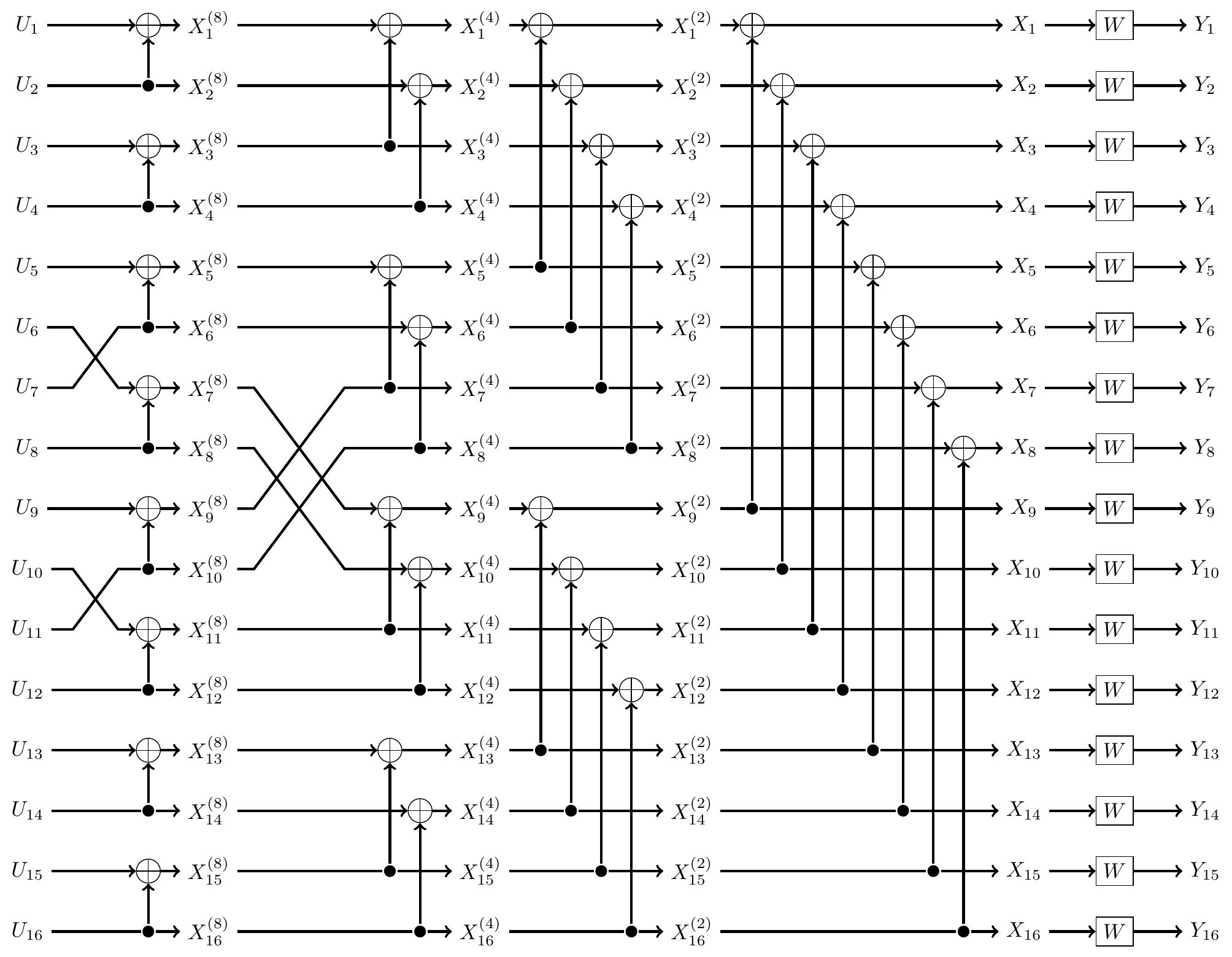}
	\caption{Encoding circuit of an $(n=16, k=8)$ ABS polar code defined by the sets in \eqref{eq:enc_example}.}
	\label{fig:example_enc}
\end{figure*}

Without Line 8, Algorithm~\ref{algo:ABS_encoding} is the same as the encoding algorithm of standard polar codes, whose time complexity is $O(n\log(n))$.
In line 8, we perform a permutation on $n_0$ elements.
According to our code construction, each of these $n_0$ elements is swapped at most once, so the number of operations involved in this permutation is no more than $n_0=2^{m-i}$.
From the for loop in Line 7, we can see that Line 8 is executed $t=2^i$ times for each $i\in\{0,1,\dots,m-1\}$.
In other words, for each fixed value of $i$, Line 8 induces at most $n_0*t=2^m=n$ operations.
Therefore, the total number of operations induced by Line 8 is upper bounded by $n*m=n\log(n)$.
Thus we conclude that the encoding complexity of ABS polar codes is still $O(n\log(n))$.

\begin{proposition}
	The encoding time complexity of ABS polar codes is $O(n\log(n))$.
\end{proposition}

Note that the set $\cI^{(n)}$ in \eqref{eq:pi} uniquely determines the permutation matrix $\mathbf{P}_n^{\ABS}$.
In Fig.~\ref{fig:example_enc}, we present the encoding circuit of an $(n=16, k=8)$ ABS polar code defined by the following sets:
\begin{equation}  \label{eq:enc_example}
	\begin{aligned}
		 & \cI^{(2)} = \emptyset,\quad \cI^{(4)} = \emptyset,\quad \cI^{(8)} = \{4\},\quad \cI^{(16)} = \{6,10\} , \\
		 & \cA=\{9,10,11,12,13,14,15,16\} .
	\end{aligned}
\end{equation}

\section{The SCL decoder for ABS polar codes}\label{sect:SCL}

In this section, we present a new SCL decoder for ABS polar codes.
The organization of this section is as follows: In Section~\ref{sect:polar_decoding}, we recap the classic SCL decoder for standard polar codes based on the $2\times 2$ polar transform.
The purpose of doing so is to get ourselves familiar with the recursive structure, which is shared by both the classic SCL decoder and our new SCL decoder.
The SCL decoder presented in Section~\ref{sect:polar_decoding} is based on the one proposed in \cite{Tal15}.
While the classic SCL decoder is based on the $2\times 2$ polar transform, our new SCL decoder is based on the DB polar transform and the SDB polar transform; see Fig.~\ref{fig:DBpt} and Fig.~\ref{fig:SDBpt} for the definitions of these two transforms.
Instead of jumping directly into the decoding of ABS polar codes, we first present a new SCL decoder for standard polar codes based on the DB polar transform in Section~\ref{sect:ST_decoder_DB}.
This new SCL decoder for standard polar codes already contains most of the new ingredients in the SCL decoder for ABS polar codes, and it helps us learn these new ingredients in a familiar setting.
Finally, in Section~\ref{sect:ABS_decoder}, we present our new SCL decoder for ABS polar codes.

\subsection{SCL decoder for standard polar codes based on the $2\times 2$ polar transform}\label{sect:polar_decoding}
In this subsection, we recap the classic SCL decoder proposed in \cite{Tal15} for standard polar codes.
Suppose that the code length is $n=2^m$, and the upper bound of the list size in the SCL decoder is $L$.
We use $L_c\in \{1, 2, \dots, L\}$ to denote the current list size. $\cA$ is the index set of the information bits.

Before describing the decoding algorithms, let us introduce some notation and intermediate variables.
Following the notation in Fig.~\ref{fig:bit_channels_polar}, $(U_1, U_2, \dots, U_n)$ is the message vector, and we use $(X_1,\dots,X_n)$ and $(Y_1,\dots,Y_n)$ to denote the random codeword vector and the random channel output vector, respectively.
We use $(y_1,\dots,y_n)$ to denote a realization of the random vector $(Y_1,\dots,Y_n)$.
For each $0\le \lambda \le m$, we introduce an intermediate vector $(X_{1}^{(2^\lambda)}, X_{2}^{(2^\lambda)},\dots, X_{n}^{(2^\lambda)})$.
For $\lambda=m$, we define the intermediate vector as
\begin{equation} \label{eq:xnn}
	(X_1^{(n)}, X_2^{(n)}, \dots, X_n^{(n)})=(U_1, U_2, \dots, U_n).
\end{equation}
For $0\le \lambda \le m-1$, the intermediate vectors are defined recursively using the following relation:
\begin{equation}\label{eq:def_intermediate}
	\begin{aligned}
		  & (X_{1}^{(2^\lambda)}, X_{2}^{(2^\lambda)},\dots, X_{n}^{(2^\lambda)})                                                                                \\
		= & (X_{1}^{(2^{\lambda+1})}, X_{2}^{(2^{\lambda+1})},\dots, X_{n}^{(2^{\lambda+1})})(\mI_{2^\lambda}\otimes\mG_2^{\polar}\otimes\mI_{2^{m-\lambda-1}}),
	\end{aligned}
\end{equation}
where $\mI_n$ is the $n\times n$ identity matrix.
By definition, $(X_1^{(1)}, X_2^{(1)}, \dots, X_n^{(1)})=(X_1, X_2, \dots, X_n)$ is the codeword vector.
Intuitively, the intermediate vector $(X_{1}^{(2^\lambda)}, X_{2}^{(2^\lambda)},\dots, X_{n}^{(2^\lambda)})$ is obtained from performing $(m-\lambda)$ layers of polar transform on the message vector $(U_1, U_2, \dots, U_n)$.
Fig.~\ref{fig:example_enc} gives a concrete example of the intermediate vectors in an ABS polar code, which are similar to the ones in standard polar codes.
For each $0\le \lambda \le m$, $1\le i\le 2^\lambda$ and $1\le \beta \le 2^{m-\lambda}$, we introduce the shorthand notation
\begin{equation}\label{eq:subscript}
	\begin{aligned}
		 & X_{i,\beta}^{(\lambda)} = X_{\beta + (i-1)2^{m-\lambda}}^{(2^\lambda)}, \quad Y_{i,\beta}^{(\lambda)} = Y_{\beta + (i-1)2^{m-\lambda}} ,                                                                                 \\
		 & \mathbi{O}_{i,\beta}^{(\lambda)} =(X_{1, \beta}^{(\lambda)}, X_{2, \beta}^{(\lambda)}, \dots, X_{i-1, \beta}^{(\lambda)}, Y_{1, \beta}^{(\lambda)}, Y_{2, \beta}^{(\lambda)}, \dots, Y_{2^\lambda, \beta}^{(\lambda)}) .
	\end{aligned}
\end{equation}
According to the standard polar code construction, the $2^{m-\lambda}$ random vectors
\begin{equation*}
	\begin{aligned}
		\left\{  (X_{1, \beta}^{(\lambda)}, X_{2, \beta}^{(\lambda)}, \dots, X_{2^\lambda, \beta}^{(\lambda)}, Y_{1, \beta}^{(\lambda)}, Y_{2, \beta}^{(\lambda)}, \dots, Y_{2^\lambda, \beta}^{(\lambda)})  \right\}_{\beta = 1}^{2^{m-\lambda}}
	\end{aligned}
\end{equation*}
are independent and identically distributed.
Moreover, the channel mapping from $X_{i,\beta}^{(\lambda)}$ to $\mathbi{O}_{i,\beta}^{(\lambda)}$ is the bit-channel $W_{i}^{(2^\lambda)}$ for every $1\le \beta \le 2^{m-\lambda}$, where $W_{i}^{(2^\lambda)}$ is defined recursively using the relation \eqref{eq:recur_bit_channels}.

Recall that $(y_1, \dots, y_n)$ is a realization of the random vector $(Y_1,\dots,Y_n)$. For each $0\le \lambda \le m$, $1\le i\le 2^\lambda$ and $1\le \beta \le 2^{m-\lambda}$, we introduce the shorthand notation $y_{i,\beta}^{(\lambda)} = y_{\beta + (i-1)2^{m-\lambda}}$, and we use $\hat x_{i,\beta}^{(\lambda)}$ to denote the decoded value of $X_{i,\beta}^{(\lambda)}$. Moreover, we define a vector
\begin{equation}\label{eq:decoded_output_vector}
	\begin{aligned}
		\hat{\mathbi{o}}_{i,\beta}^{(\lambda)} = (\hat x_{1,\beta}^{(\lambda)}, \hat x_{2,\beta}^{(\lambda)}, \dots, \hat x_{i-1,\beta}^{(\lambda)},  y_{1, \beta}^{(\lambda)}, y_{2, \beta}^{(\lambda)}, \dots, y_{2^\lambda, \beta}^{(\lambda)}) .
	\end{aligned}
\end{equation}
By the analysis above, we have
\begin{equation} \label{eq:PWI}
	\mathbb{P} \big( \mathbi{O}_{i,\beta}^{(\lambda)}=\hat{\mathbi{o}}_{i,\beta}^{(\lambda)} \big| X_{i,\beta}^{(\lambda)}=b \big) = W_{i}^{(2^\lambda)} \big( \hat{\mathbi{o}}_{i,\beta}^{(\lambda)} \big| b \big)   \qquad \text{for~} b\in\{0,1\} .
\end{equation}

Now we are ready to introduce the data structures used in the SCL decoder for standard polar codes.
Most of the data structures below are also used in the SCL decoder for ABS polar codes.
\begin{enumerate}[(i)]
	\item 4-dimensional \emph{probability array $\tD$.}
	      The entries in the array $\tD$ are indexed as
	      \begin{align*}
		      \tD[\lambda, s, \beta, b],\quad & 0\le \lambda \le m,\qquad~ 1\le s\le L,          \\
		                                      & 1\le \beta \le 2^{m-\lambda},\quad 0\le b \le 1.
	      \end{align*}
	      For each $0\le \lambda \le m, 1\le s \le L$,  we define a subarray of $\tD$ as
	      \begin{equation*}
		      \tD[\lambda, s]=  (\tD[\lambda, s, \beta, b],\quad 1\le \beta\le 2^{m-\lambda}, \quad b\in\{0,1\}),
	      \end{equation*}
	      and we use $\vec{\tD}[\lambda, s]$ to denote the pointer to the head address of $\tD[\lambda, s]$.
	      In the algorithms below, we will write $\tD[\lambda, s, \beta, b]$ and $\vec{\tD}[\lambda, s][\beta, b]$ interchangeably.
	      Each array $\tD[\lambda, s]$ is used to store a set of transition probabilities in \eqref{eq:PWI}.

	\item 1-dimensional \emph{integer array $\numD$.}
	      The entries of $\numD$ are $\numD[\lambda], 0\le \lambda \le m$. The entry $\numD[\lambda]$ takes value in the set $\{0,1,2,\dots,L\}$ for every $0\le \lambda \le m$. The value of $\numD[\lambda]$ has the following meaning: The arrays $\tD[\lambda, 1],\tD[\lambda, 2],\dots,\tD[\lambda, \numD[\lambda]]$ are currently occupied in the decoding procedure while the arrays $\tD[\lambda, \numD[\lambda]+1],\tD[\lambda, \numD[\lambda]+2],\dots,\tD[\lambda, L]$ are free to use. See Fig.~\ref{fig:ds_n_nd} for an illustration.

	      \begin{figure}
		      \centering
		      \includegraphics[scale = 1.2]{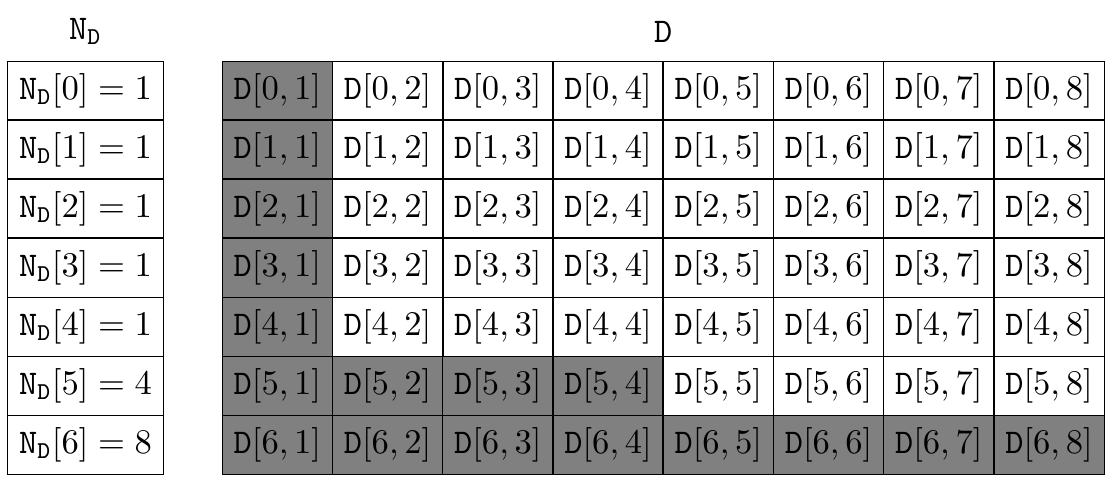}
		      \caption{An illustration of $\tD$ and $\numD$ for code length $n=64$ and list size $L = 8$. We put $\tD[\lambda, s]$ in a shaded cell if it is currently occupied; otherwise, we put it in a white cell. For example, $\numD[5] = 4$ means that  $\tD[5,1], \tD[5,2], \tD[5,3], \tD[5,4]$ have already been allocated to store some transition probabilities while $\tD[5,5], \tD[5,6], \tD[5,7], \tD[5,8]$ are free to use.}
		      \label{fig:ds_n_nd}
	      \end{figure}

	\item 3-dimensional \emph{bit array $\tB$.}
	      The entries in the array $\tB$ are indexed as
	      $$
		      \tB[\lambda, s, \beta],\quad 0\le \lambda \le m, \quad  1\le s\le 2L, \quad  1\le \beta \le 2^{m-\lambda}.
	      $$
	      For each $0\le \lambda \le m, 1\le s \le 2L$,  we define a subarray of $\tB$ as
	      \begin{equation*}
		      \tB[\lambda, s]=  (\tB[\lambda, s, \beta],\quad 1\le \beta\le 2^{m-\lambda}),
	      \end{equation*}
	      and we use $\vec{\tB}[\lambda, s]$ to denote the pointer to the head address of $\tB[\lambda, s]$.
	      In the algorithms below, we will write $\tB[\lambda, s, \beta]$ and  $\vec{\tB}[\lambda, s][\beta]$ interchangeably. Each array $\tB[\lambda, s]$ is used to store a set of decoding results of the intermediate vectors.

	\item 1-dimensional \emph{integer array $\numB$.}
	      The entries of $\numB$ are $\numB[\lambda], 0\le \lambda \le m$. The entry $\numB[\lambda]$ takes value in the set $\{0,1,2,\dots,2L\}$ for every $0\le \lambda \le m$. The value of $\numB[\lambda]$ has the following meaning: The arrays $\tB[\lambda, 1],\tB[\lambda, 2],\dots,\tB[\lambda, \numB[\lambda]]$ are currently occupied in the decoding procedure while the arrays $\tB[\lambda, \numB[\lambda]+1],\tB[\lambda, \numB[\lambda]+2],\dots,\tB[\lambda, 2L]$ are free to use.

	\item 1-dimensional \emph{probability array $\score$.}
	      The entries of $\score$ are $\score[\ell], 1\le \ell \le L_c$, where $L_c\in\{1,2,\dots,L\}$ is the current list size.
	      Each $\score[\ell]$ records the current transition probability of the $\ell$th candidate in the decoding list.
	      When the current list size is larger than the prescribed upper bound $L$, we prune the list according to the value of $\score[\ell]$.

	\item 2-dimensional \emph{pointer arrays $\tP, \barP$.} Their entries are
	      $$
		      \tP=(\tP[\ell,\lambda],~ 1\le \ell\le L,~ 0\le \lambda \le m), \qquad
		      \barP=(\barP[\ell,\lambda],~ 1\le \ell\le L,~ 0\le \lambda \le m) .
	      $$
	      We use $\tP[\ell,\lambda]$ to store the pointer $\vec{\tD}[\lambda, \numD[\lambda]+1]$, so that we can store the transition probabilities in the array $\tD[\lambda, \numD[\lambda]+1]$ and access them in the future. We usually assign values (i.e., pointers) to $\tP[\ell,\lambda]$ through the function $\texttt{allocate\_prob}$ in Algorithm~\ref{algo:st_allocate_prob}. The function $\texttt{allocate\_prob}$ is called in Line~4 of Algorithm~\ref{algo:ST_Decode}, Line~3 of Algorithm~\ref{algo:calcu_-_proba}, and Line~3 of Algorithm~\ref{algo:calcu_+_proba}. The array $\barP$ is a supplement to $\tP$. We use $\barP$ when the entries in $\tP$ are occupied.

	\item 2-dimensional \emph{pointer arrays $\tR, \barR$.}
	      Their entries are
	      $$
		      \tR=(\tR[\ell,\lambda],~ 1\le \ell\le L,~ 0\le \lambda \le m), \qquad
		      \barR=(\barR[\ell,\lambda],~ 1\le \ell\le L,~ 0\le \lambda \le m) .
	      $$
	      We use $\tR[\ell,\lambda]$ to store the pointer $\vec{\tB}[\lambda, \numB[\lambda]+1]$, so that we can store the decoding results of intermediate vectors in the array $\tB[\lambda, \numB[\lambda]+1]$ and access them in the future. We usually assign values (i.e., pointers) to $\tR[\ell,\lambda]$ through the function $\texttt{allocate\_bit}$ in Algorithm~\ref{algo:st_allocate_bit}. The function $\texttt{allocate\_bit}$ is called in Line~13 of Algorithm~\ref{algo:st_decode_channel} and Lines~10,16 of Algorithm~\ref{algo:st_decode_boundary_channel}. The array $\barR$ is a supplement to $\tR$. We use $\barR$ when the entries in $\tR$ are occupied.

	\item \emph{priority queue $\texttt{PriQue}$.}
	      $\texttt{PriQue}$ is a maximum priority queue with size $2L$ such that the element with the maximum value is always removed first from the queue.
	      We use $\texttt{PriQue}$ to record and prune candidate decoding paths.
	      Each element in the queue is a triple $(\ell, b, \prob)$ with the following meaning:
	      When we decode $U_i$ in the last layer $\lambda = m$, the (posterior) probability of $U_i = b$ in the $\ell$th decoding path is $\prob$. The queue $\texttt{PriQue}$ has 4 interfaces:
	      i)   $\texttt{PriQue.push}(\ell, b, \prob)$  pushes the element $(\ell, b, \prob)$ to the queue;
	      ii)  $\texttt{PriQue.pop}()$  removes the element $(\ell, b, \prob)$ with the maximum $\prob$ in the queue;
	      iii) $\texttt{PriQue.clear}()$ removes all the remaining elements in the  queue;
	      iv)  $\texttt{PriQue.size}()$ returns the current number of elements in the queue.
\end{enumerate}

We associate each candidate in the decoding list with a list element.
There are at most $L$ list elements in total. For $1\le \ell\le L$, the $\ell$th list element has the following fields:
\begin{equation}\label{eq:st_list_field}
	\begin{aligned}
		 & (\tP[\ell,0], \tP[\ell,1], \dots, \tP[\ell, m]), \\
		 & (\tR[\ell,0], \tR[\ell,1], \dots, \tR[\ell, m]), \\
		 & \score[\ell].
	\end{aligned}
\end{equation}

\begin{algorithm}[ht]
	\DontPrintSemicolon
	\caption{$\texttt{allocate\_prob}(\lambda)$}
	\label{algo:st_allocate_prob}
	\KwIn{layer $\lambda\in\{0, 1, 2, \dots, m\}$}
	\KwOut{a pointer to the allocated memory}

	$\numD[\lambda]\gets \numD[\lambda] + 1$

	\Return $\vec{\tD}[\lambda, \numD[\lambda]]$

\end{algorithm}

\begin{algorithm}[ht]
	\DontPrintSemicolon
	\caption{$\texttt{allocate\_bit}(\lambda)$}
	\label{algo:st_allocate_bit}
	\KwIn{layer $\lambda\in\{0, 1, 2, \dots, m\}$}
	\KwOut{a pointer to the allocated memory}

	$\numB[\lambda]\gets \numB[\lambda] + 1$

	\Return $\vec{\tB}[\lambda, \numB[\lambda]]$

\end{algorithm}

The function $\texttt{allocate\_prob}$ in Algorithm~\ref{algo:st_allocate_prob} and the function
$\texttt{allocate\_bit}$ in Algorithm~\ref{algo:st_allocate_bit} are used to allocate memory spaces throughout the decoding procedure.
$\texttt{allocate\_prob}(\lambda)$ returns the pointer to the next usable array in $\tD[\lambda, 1],\tD[\lambda, 2],\dots,\tD[\lambda, L]$ and updates the value of $\numD[\lambda]$. Similarly, $\texttt{allocate\_bit}(\lambda)$ returns the pointer to the next usable array in $\tB[\lambda, 1],\tB[\lambda, 2],\dots,\tB[\lambda, 2L]$ and updates the value of $\numB[\lambda]$.

We present the main function $\texttt{ST\_decode}((y_1, y_2,\dots, y_n))$ in Algorithm~\ref{algo:ST_Decode}.
Note that we only update the value of the current list size in the last layer $\lambda = m$, and we have only one list element in the beginning.
The first 3 lines initialize the parameters. In Line~4, we assign the pointer $\vec{\tD}[0,1]$ to $\tP[1, 0]$ and update the value of $\numD[0]$ to be $1$. In Lines~5--7, we store the transition probabilities of the whole channel output vector in the array $\tD[0,1]$. Line~8 executes recursive decoding which we will explain later. After recursive decoding, we obtain $L_c$ list elements. In the $\ell$th list element, $\score[\ell]$ is the transition probability which measures the likelihood of this list element, and the decoding result is stored in the array $(\tR[\ell,0][1],\tR[\ell,0][2],\dots,\tR[\ell,0][n])$. In Lines~9--17, we pick the list element with the maximum $\score[\ell]$ and return the corresponding decoding result.

\begin{algorithm}[ht]
	\DontPrintSemicolon
	\caption{\texttt{ST\_Decode}$((y_1,y_2,\dots,y_n))$}
	\label{algo:ST_Decode}
	\KwIn{the received vector $(y_1,y_2,\dots,y_n)\in\cY^n$}
	\KwOut{the decoded codeword $(\hat x_1,\hat x_2,\dots,\hat x_n)\in\{0,1\}^n$}

	\For{\em $\lambda \in \{1, 2, \dots, m\}$}{
		$\numD[\lambda] \gets 0, \quad \numB[\lambda] \gets 0$
	}

	$L_c\gets 1$

	$\tP[1, 0]\gets \texttt{allocate\_prob}(0)$

	\For{\em $\beta \in \{1, 2, \dots, n$\}}{
		\For{\em $b\in \{0,1\}$}{
			$\tP[1,0][\beta,b]\gets W(y_\beta|b)$
		}
	}

	\texttt{decode\_channel$(0, 1)$}
	\Comment{Algorithm~\ref{algo:st_decode_channel}, recursive decoding}

	$\max\_\score\gets 0$

	$\max\_\ell\gets 0$

	\For{$\ell\in\{1,2,\dots, L_c\}$}{
		\If{$\score[\ell]\ge \max\_\score$}{
			$\max\_\score\gets \score[\ell]$

			$\max\_\ell\gets \ell$
		}
	}

	\For{$\beta=1,2,\dots,n$}
	{
		$\hat x_\beta \gets \tR[\max\_\ell,0][\beta]$
	}

	\Return $(\hat x_1,\hat x_2,\dots,\hat x_n)$
\end{algorithm}

Before explaining the recursive decoding function $\texttt{decode\_channel}$ in Algorithm~\ref{algo:st_decode_channel}, let us introduce some additional notation. Recall that we defined a vector $\hat{\mathbi{o}}_{i,\beta}^{(\lambda)}$ in \eqref{eq:decoded_output_vector} which consists of both the decoding results of intermediate vectors and the channel outputs. This notation is designed for the SC decoder because we only have a single decoding result in the whole SC decoding procedure. However, we have multiple decoding results in the SCL decoder, so we need the following modification of the notation $\hat{\mathbi{o}}_{i,\beta}^{(\lambda)}$.
For each $1 \le \ell \le L_c$, we use $\hat x_{i,\beta}^{(\ell,\lambda)}$ to denote the decoded value of $X_{i,\beta}^{(\lambda)}$ in the $\ell$th list element, and we define a vector
\begin{equation}  \label{eq:listdecoded}
	\hat{\mathbi{o}}_{i,\beta}^{(\ell,\lambda)} = (\hat x_{1,\beta}^{(\ell,\lambda)}, \hat x_{2,\beta}^{(\ell,\lambda)}, \dots, \hat x_{i-1,\beta}^{(\ell,\lambda)},  y_{1, \beta}^{(\lambda)}, y_{2, \beta}^{(\lambda)}, \dots, y_{2^\lambda, \beta}^{(\lambda)}) .
\end{equation}
Then \eqref{eq:PWI} becomes
$$
	\mathbb{P} \big( \mathbi{O}_{i,\beta}^{(\lambda)}=\hat{\mathbi{o}}_{i,\beta}^{(\ell,\lambda)} \big| X_{i,\beta}^{(\lambda)}=b \big) = W_{i}^{(2^\lambda)} \big( \hat{\mathbi{o}}_{i,\beta}^{(\ell,\lambda)} \big| b \big)   \qquad \text{for~} b\in\{0,1\} .
$$

\begin{lemma}
	Suppose that $0\le \lambda \le m$ and $1\le i\le 2^\lambda$.
	Before we call the function $\texttt{decode\_channel}$ in Algorithm~\ref{algo:st_decode_channel} with input parameters $(\lambda, i)$, the pointer $\tP[\ell, \lambda]$ satisfies that
	\begin{equation}\label{eq:prob}
		\tP[\ell, \lambda][\beta, b] = W_{i}^{(2^\lambda)}(\hat{\mathbi o}_{i,\beta}^{(\ell,\lambda)} | b)
		\quad  \text{for all~} 1\le \ell\le L_c,~ 1\le \beta \le 2^{m-\lambda}
		\text{~and~} b\in\{0,1\} .
	\end{equation}
	After the function $\texttt{decode\_channel}(\lambda, i)$ in Algorithm~\ref{algo:st_decode_channel} returns, the pointer $\tR[\ell, \lambda]$ satisfies that
	\begin{equation} \label{eq:bit}
		\tR[\ell, \lambda][\beta] = \hat{x}_{i,\beta}^{(\ell, \lambda)}
		\quad  \text{for all~} 1\le \ell\le L_c \text{~and~}
		1\le \beta\le 2^{m-\lambda}.
	\end{equation}
\end{lemma}
\begin{proof}
	We prove \eqref{eq:bit} first, and we prove it by induction. Lines~1--2 of Algorithm~\ref{algo:st_decode_channel} deal with the base case $\lambda=m$, where we decode $U_i$ in the message vector $(U_1,U_2,\dots,U_n)$ by calling the function $\texttt{decode\_boundary\_channel}(i)$ in Algorithm~\ref{algo:st_decode_boundary_channel}.
	By \eqref{eq:subscript}, when $\lambda=m$, we have $X_{i,1}^{(m)}=X_i^{(n)}$. By \eqref{eq:xnn}, we further obtain that $X_{i,1}^{(m)}=U_i$.
	If $U_i$ is a frozen bit, then Line~17 of Algorithm~\ref{algo:st_decode_boundary_channel} immediately implies \eqref{eq:bit}.
	If $U_i$ is an information bit, we first use $\barR[\ell, m][1]$ to store the decoding result of $U_i$ in the $\ell$th list element\footnote{The variable $b$ in Line~11 of Algorithm~\ref{algo:st_decode_boundary_channel} is the decoding result of $U_i$ in the $\ell$th list element. We will explain Algorithm~\ref{algo:st_decode_boundary_channel} later.}; see Line~11 of Algorithm~\ref{algo:st_decode_boundary_channel}. Next we swap $\barR$ and $\tR$ in Line~13, so \eqref{eq:bit} is satisfied.

	For the inductive step, we assume that \eqref{eq:bit} holds for $\lambda+1$ and prove it for $\lambda$. By this induction hypothesis, after executing Line~6 of Algorithm~\ref{algo:st_decode_channel}, we have
	$$
		\tR[\ell, \lambda+1][\beta] = \hat{x}_{2i-1,\beta}^{(\ell, \lambda+1)}
		\quad  \text{for all~} 1\le \ell\le L_c \text{~and~}
		1\le \beta\le 2^{m-\lambda-1}.
	$$
	After executing Line~8 and Line~12 of Algorithm~\ref{algo:st_decode_channel}, we have
	$$
		\temppointer[\beta] = \hat{x}_{2i-1,\beta}^{(\ell, \lambda+1)}
		\quad  \text{for all~} 1\le \ell\le L_c \text{~and~}
		1\le \beta\le 2^{m-\lambda-1}.
	$$
	Again by the induction hypothesis, after executing Line~10, we have
	$$
		\tR[\ell, \lambda+1][\beta] = \hat{x}_{2i,\beta}^{(\ell, \lambda+1)}
		\quad  \text{for all~} 1\le \ell\le L_c \text{~and~}
		1\le \beta\le 2^{m-\lambda-1}.
	$$
	Since we set $n_c=2^\lambda$ in Line~4, we have $n/(2n_c)=2^{m-\lambda-1}$.
	Therefore, Lines~15--16 become
	\begin{equation} \label{eq:R+-}
		\begin{aligned}
			\tR[\ell, \lambda][\beta] = \hat{x}_{2i-1,\beta}^{(\ell, \lambda+1)} + \hat{x}_{2i,\beta}^{(\ell, \lambda+1)} , \quad
			\tR[\ell, \lambda][\beta + 2^{m-\lambda-1}] = \hat{x}_{2i,\beta}^{(\ell, \lambda+1)}
			\\
			\text{for all~} 1\le \ell\le L_c \text{~and~}
			1\le \beta\le 2^{m-\lambda-1}.
		\end{aligned}
	\end{equation}
	\eqref{eq:def_intermediate}--\eqref{eq:subscript} together imply that
	$$
		X_{i,\beta}^{(\lambda)} = X_{2i-1,\beta}^{(\lambda+1)} + X_{2i,\beta}^{(\lambda+1)}, \quad
		X_{i,\beta+2^{m-\lambda-1}}^{(\lambda)} = X_{2i,\beta}^{(\lambda+1)}  \quad
		\text{for all~} 1\le \beta\le 2^{m-\lambda-1}.
	$$
	This further implies that
	\begin{equation} \label{eq:fthat}
		\begin{aligned}
			\hat{x}_{i,\beta}^{(\ell, \lambda)} = \hat{x}_{2i-1,\beta}^{(\ell, \lambda+1)} + \hat{x}_{2i,\beta}^{(\ell, \lambda+1)}, \quad
			\hat{x}_{i,\beta+2^{m-\lambda-1}}^{(\ell, \lambda)} = \hat{x}_{2i,\beta}^{(\ell, \lambda+1)} \\
			\text{for all~} 1\le \ell\le L_c \text{~and~}
			1\le \beta\le 2^{m-\lambda-1}.
		\end{aligned}
	\end{equation}
	Combining this with \eqref{eq:R+-}, we complete the proof of \eqref{eq:bit}.

	Next we prove \eqref{eq:prob} by induction. This time the base case is $\lambda=0$, and this case only occurs once in Line~8 of Algorithm~\ref{algo:ST_Decode} during the whole decoding procedure. Note that the channel $W_1^{(1)}$ is $W$ itself. Therefore, Lines~5--7 of Algorithm~\ref{algo:ST_Decode} immediately imply \eqref{eq:prob} for $\lambda=0$.

	For the inductive step, we assume that \eqref{eq:prob} holds for $\lambda$ and prove it for $\lambda+1$. By this induction hypothesis, \eqref{eq:prob} holds for $\lambda$ when we execute Line~5 of Algorithm~\ref{algo:st_decode_channel}. In other words, the array associated with the pointer $\tP[\ell, \lambda]$ stores the transition probabilities of $W_i^{(2^\lambda)}$. By \eqref{eq:recur_bit_channels}, $W_{2i-1}^{(2^{\lambda+1})}$ is the ``$-$" transform of $W_i^{(2^\lambda)}$. The function \texttt{calculate\_$-$\_transform}$(\lambda+1)$ calculates the ``$-$" transform of $W_i^{(2^\lambda)}$ and stores the results in the array associated with the pointer $\tP[\ell, \lambda+1]$, so \eqref{eq:prob} holds before we call $\texttt{decode\_channel}$ in Line~6 of Algorithm~\ref{algo:st_decode_channel}. Again by \eqref{eq:recur_bit_channels}, $W_{2i}^{(2^{\lambda+1})}$ is the ``$+$" transform of $W_i^{(2^\lambda)}$. The function \texttt{calculate\_$+$\_transform}$(\lambda+1)$ in Line~9 of Algorithm~\ref{algo:st_decode_channel} calculates the ``$+$" transform of $W_i^{(2^\lambda)}$ and stores the results in the array associated with the pointer $\tP[\ell, \lambda+1]$, so \eqref{eq:prob} holds before we call $\texttt{decode\_channel}$ in Line~10 of Algorithm~\ref{algo:st_decode_channel}.

	During the whole decoding procedure, the function $\texttt{decode\_channel}$ is only called in Line~8 of Algorithm~\ref{algo:ST_Decode} and Lines~6,10 of Algorithm~\ref{algo:st_decode_channel}. We have proved that \eqref{eq:prob} holds for all three places. This completes the proof of the lemma.
\end{proof}

Now let us explain how Algorithm~\ref{algo:st_decode_boundary_channel} works when $U_i$ is an information bit. First, we explore both cases $U_i=0$ and $U_i=1$ for every list element; see Lines~2--4. The variable $b$ in Lines~3--4 represents the (possible) value of $U_i$. Since we explore two possible paths for each existing list element, we have expanded the list size by a factor of $2$ after executing Lines~2--4. If the current list size is larger than $L$, then we need to prune the list, and this is done in Lines~5--13. In Line~5, we update the current list size $L_c$ to be the smaller value among $L$ and the size of $\texttt{PriQue}$. Then in Lines~6--11, we execute \texttt{PriQue.pop()} $L_c$ times to obtain $L_c$ elements in the queue with the largest value of $\score[\ell]$. By Line~4, $\score[\ell]$ stores the transition probability $\tP[\ell, m][1,b]$, which measures the likelihood of the $\ell$th list element. Therefore, we obtain $L_c$ list elements with the largest likelihood after executing Lines~6--11.

\begin{algorithm}[ht]
	\DontPrintSemicolon
	\caption{\texttt{decode\_channel$(\lambda, i)$}}
	\label{algo:st_decode_channel}
	\KwIn{layer $\lambda\in \{0,1,2,\dots, m\}$ and index $i\in \{1, 2, \dots, 2^\lambda\}$}

	\uIf{$\lambda = m$}{
		\texttt{decode\_boundary\_channel}$(i)$
		\Comment{Algorithm~\ref{algo:st_decode_boundary_channel}}
	}\Else{
		$n_c\gets 2^\lambda$
		\Comment{$W_{2i-1}^{(2n_c)} = (W_{i}^{(n_c)})^-$}

		\texttt{calculate\_$-$\_transform}$(\lambda+1)$

		\texttt{decode\_channel}$(\lambda+1, 2i-1)$

		\For{$\ell \in\{1,2,\dots, L_c\}$}{
			$\tR[\ell, \lambda]\gets\tR[\ell, \lambda+1]$
		}

		\texttt{calculate\_$+$\_transform}$(\lambda+1)$
		\Comment{$W_{2i}^{(2n_c)} = (W_{i}^{(n_c)})^+$}

		\texttt{decode\_channel}$(\lambda+1, 2i)$

		\For{$\ell\in\{1,2, \dots L_c\}$}{
			$\temppointer\gets \tR[\ell, \lambda]$

			$\tR[\ell, \lambda]\gets \texttt{allocate\_bit}(\lambda)$

			\For{$\beta \in\{1, 2, \dots, n/(2n_c)\}$}{

				$\tR[\ell, \lambda][\beta]\gets \temppointer[\beta] + \tR[\ell,\lambda+1][\beta]$

				$\tR[\ell, \lambda][\beta + n/(2n_c)]\gets \tR[\ell,\lambda+1][\beta]$
			}
		}

		$\numB[\lambda+1]\gets 0$
	}

	$\numD[\lambda]\gets 0$

	\Return

\end{algorithm}

\begin{algorithm}[ht]
	\DontPrintSemicolon
	\caption{\texttt{decode\_boundary\_channel$(i)$}}
	\label{algo:st_decode_boundary_channel}
	\KwIn{index $i$ in the last layer $(\lambda = m)$}

	\eIf(\Comment{$U_i$ is an information bit}){$i\in\cA$}{

		\For{$\ell \in \{1,2,\dots, L_c\}$}{
			\For{$b\in\{0,1\}$}{
				$\texttt{PriQue.push}(\ell, b,\tP[\ell, m][1,b])$
			}
		}

		$L_c\gets \min\{L, \texttt{PriQue.size()}\}$

		\For{$\ell \in\{1,2,\dots, L_c\}$}{
			$(\ell', b, \score[\ell])\gets \texttt{PriQue.pop()}$

			\For{$\lambda \in\{0,1,2,\dots, m-1\}$}{
				$(\barP[\ell,\lambda], \barR[\ell,\lambda]) \gets (\tP[\ell',\lambda], \tR[\ell',\lambda])$
			}

			$\barR[\ell, m]\gets \texttt{allocate\_bit}(m)$

			$\barR[\ell, m][1]\gets b$
		}

		$\texttt{PriQue.clear()}$
		\Comment{Remove all the remaining elements}

		$\texttt{swap}(\barP, \tP)$,\quad $\texttt{swap}(\barR, \tR)$

	}(\Comment{$U_i$ is a frozen bit}){

		\For{$\ell \in \{1,2,\dots, L_c\}$}{
			$\tR[\ell, m]\gets \texttt{allocate\_bit}(m)$

			$\tR[\ell, m][1]\gets$ frozen value of $U_i$
		}
	}
	\Return
\end{algorithm}


\begin{algorithm}[ht]
	\DontPrintSemicolon
	\caption{\texttt{calculate\_$-$\_transform}$(\lambda)$}
	\label{algo:calcu_-_proba}
	\KwIn{layer $1\le \lambda \le m$}
	\KwOut{Update the entries pointed by $\tP[\ell, \lambda]$, $1\le \ell\le L_c $ }

	$\bar{n}_c \gets 2^{m-\lambda}$

	\For{$\ell\in\{1, 2, \dots, L_c\}$}{

		$\tP[\ell,\lambda]\gets \texttt{allocate\_prob}(\lambda)$

		\For{\em $\beta \in \{1,2,\dots,\bar{n}_c\}, a \in\{0,1\}$}{
			$\beta' \gets \beta + \bar{n}_c$

			$\tP[\ell, \lambda][\beta,a] \gets$ $\frac12\sum_{b\in\{0,1\}} \tP[\ell, \lambda-1][\beta, a+b]\tP[\ell, \lambda-1][\beta',b]$
		}
	}

	\Return

\end{algorithm}

\begin{algorithm}[ht]
	\DontPrintSemicolon
	\caption{\texttt{calculate\_$+$\_transform}$(\lambda)$}
	\label{algo:calcu_+_proba}
	\KwIn{layer $1\le \lambda \le m$}
	\KwOut{Update the entries pointed by $\tP[\ell, \lambda]$, $1\le \ell\le L_c $ }

	$\bar{n}_c \gets 2^{m-\lambda}$

	\For{$\ell\in\{1, 2, \dots, L_c\}$}{
		$\tP[\ell,\lambda]\gets \texttt{allocate\_prob}(\lambda)$

		\For{\em $\beta \in \{1,2,\dots,\bar{n}_c\},b\in\{0,1\}$}{
			$a \gets \tR[\ell, \lambda-1][\beta]$

			$\beta' \gets \beta + \bar{n}_c$

			$\tP[\ell, \lambda][\beta,b] \gets$ $\frac12 \tP[\ell, \lambda-1][\beta, a+b]\tP[\ell, \lambda-1][\beta',b]$
		}
	}

	\Return

\end{algorithm}

The next lemma shows that the data structures $\tD$ and $\tB$ are large enough to store the transition probabilities and the decoding results of the intermediate vectors throughout the decoding procedure.

\begin{lemma}\label{lemma:st_space}
	Throughout the whole decoding procedure, we have $\numD[\lambda]\le L$ and $\numB[\lambda]\le 2L$ for all $0\le \lambda \le m$.
	The space complexity of the SCL decoder is $O(Ln)$.
\end{lemma}
\begin{proof}
	For every $0 \le \lambda\le m$ and every $1\le i\le 2^{\lambda}$, the function $\texttt{decode\_channel}(\lambda, i)$ is called only once.
	Moreover, the function $\texttt{decode\_channel}(\lambda, i+1)$ is always called after the function $\texttt{decode\_channel}(\lambda, i)$ returns. Each time we call the function $\texttt{decode\_channel}(\lambda, i)$, we only need to store the transition probabilities for $L_c\le L$ different decoding paths, and we always reset $\numD[\lambda]$ to $0$ before the function $\texttt{decode\_channel}(\lambda, i)$ returns, so $\numD[\lambda]\le L$.

	We need to store the decoding results of intermediate vectors for $L_c\le L$ list elements when we call the function $\texttt{decode\_channel}(\lambda+1, 2i-1)$ in Line~6 of Algorithm~\ref{algo:st_decode_channel}. Similarly, we need to store the decoding results of intermediate vectors for another $L'_c\le L$ list elements\footnote{We use $L'_c$ here because the current list size may change over the decoding procedure.} when we call the function $\texttt{decode\_channel}(\lambda+1, 2i)$ in Line~10 of Algorithm~\ref{algo:st_decode_channel}. Therefore, before we reset $\numB[\lambda+1]$ to $0$ in Line~17, we have $\numB[\lambda+1]=L_c+L'_c\le 2L$. This proves that $\numB[\lambda]$ can not exceed $2L$ for all $0\le \lambda \le m$.

	Next we prove the $O(Ln)$ space complexity of the SCL decoder. The number of entries in the array $\tD$ is upper bounded by
	$$
		2L \sum_{\lambda=0}^m 2^{m-\lambda}= 2L(1+2+4+\dots+2^m) < 2L\cdot 2^{m+1}=4Ln .
	$$
	Similarly, the number of entries in $\tB$ is upper bounded by
	$$
		2L \sum_{\lambda=0}^m 2^{m-\lambda} < 4Ln .
	$$
	The number of entries in both $\numD$ and $\numB$ is $O(\log(n))$. The number of entries in both $\score$ and $\texttt{PriQue}$ is $O(L)$. The number of entries in the pointer arrays $\tP,\barP,\tR,\barR$ is $O(L\log(n))$. Adding these up gives us the $O(Ln)$ space complexity.
\end{proof}

\begin{proposition}
	The decoding time complexity of standard polar codes is $O(Ln\log(n))$.
\end{proposition}

\subsection{SCL decoder for standard polar codes based on the Double-Bits polar transform}\label{sect:ST_decoder_DB}
In this subsection, we present a new SCL decoder for standard polar codes based on the Double-Bits polar transform in Fig.~\ref{fig:DBpt}.
We still use the notation in \eqref{eq:xnn}--\eqref{eq:subscript} and \eqref{eq:listdecoded}.
By \eqref{eq:st_adc}, the channel mapping from $(X_{i,\beta}^{(\lambda)}, X_{i+1,\beta}^{(\lambda)})$ to $\mathbi{O}_{i,\beta}^{(\lambda)}$ is the adjacent-bits-channel $V_{i}^{(2^\lambda)}$ for every $1\le \beta \le 2^{m-\lambda}$, i.e.,
\begin{equation} \label{eq:PVI}
	\mathbb{P} \big( \mathbi{O}_{i,\beta}^{(\lambda)}=\hat{\mathbi{o}}_{i,\beta}^{(\ell,\lambda)} \big| X_{i,\beta}^{(\lambda)}=a, X_{i+1,\beta}^{(\lambda)}=b \big) = V_{i}^{(2^\lambda)} \big( \hat{\mathbi{o}}_{i,\beta}^{(\ell,\lambda)} \big| a,b \big)   \qquad \text{for~} a,b\in\{0,1\} .
\end{equation}
Below we list the data structures of the new SCL decoder for standard polar codes based on the DB polar transform.
\begin{enumerate}[(i)]
	\item 5-dimensional \emph{probability array $\tD$.}
	      The entries in the array $\tD$ are indexed as
	      \begin{align*}
		      \tD[\lambda, s, \beta, a, b],\quad & 1\le \lambda \le m,\qquad~ 1\le s\le L,              \\
		                                         & 1\le \beta \le 2^{m-\lambda},\quad 0\le a, b \le 1 .
	      \end{align*}
	      For each $1\le \lambda \le m, 1\le s \le L$,  we define a subarray of $\tD$ as
	      $$
		      \tD[\lambda, s] = (\tD[\lambda, s, \beta, a, b],\quad 1\le \beta\le 2^{m-\lambda}, \quad 0\le a,b\le 1),
	      $$
	      and we use $\vec{\tD}[\lambda, s]$ to denote the pointer to the head address of $\tD[\lambda, s]$.
	      In the algorithms below, we will write $\tD[\lambda, s, \beta, a, b]$ and $\vec{\tD}[\lambda, s][\beta, a, b]$ interchangeably.
	      Each array $\tD[\lambda, s]$ is used to store a set of transition probabilities in \eqref{eq:PVI}.

	\item 1-dimensional \emph{integer array $\numD$.}
	      The entries of $\numD$ are $\numD[\lambda], 1\le \lambda \le m$. This array is defined in the same way as the previous subsection.

	\item 3-dimensional \emph{bit array $\tB$.}
	      The entries in the array $\tB$ are indexed as
	      \begin{equation} \label{eq:B_index}
		      \tB[\lambda, s, \beta],\quad 1\le \lambda \le m, \quad  1\le s\le 4L, \quad  1\le \beta \le 2^{m-\lambda}.
	      \end{equation}
	      For each $1\le \lambda \le m, 1\le s \le 4L$,  we define a subarray of $\tB$ as
	      \begin{equation*}
		      \tB[\lambda, s]=  (\tB[\lambda, s, \beta],\quad 1\le \beta\le 2^{m-\lambda}),
	      \end{equation*}
	      and we use $\vec{\tB}[\lambda, s]$ to denote the pointer to the head address of $\tB[\lambda, s]$.
	      In the algorithms below, we will write $\tB[\lambda, s, \beta]$ and  $\vec{\tB}[\lambda, s][\beta]$ interchangeably. Each array $\tB[\lambda, s]$ is used to store a set of decoding results of the intermediate vectors.

	\item 1-dimensional \emph{integer array $\numB$.}
	      The entries of $\numB$ are $\numB[\lambda], 1\le \lambda \le m$. The entry $\numB[\lambda]$ takes value in the set $\{0,1,2,\dots,4L\}$ for every $1\le \lambda \le m$. The meaning of $\numB[\lambda]$ is the same as the previous subsection.

	\item 1-dimensional \emph{probability array $\score$,}
	      defined in the same way as the previous subsection.

	\item 2-dimensional \emph{pointer arrays $\tP, \barP$.} Their entries are
	      $$
		      \tP=(\tP[\ell,\lambda],~ 1\le \ell\le L,~ 1\le \lambda \le m), \qquad
		      \barP=(\barP[\ell,\lambda],~ 1\le \ell\le L,~ 1\le \lambda \le m) .
	      $$
	      They are used in the same way as the previous subsection.

	\item 2-dimensional \emph{pointer arrays $\tR, \barR$.}
	      Their entries are
	      $$
		      \tR=(\tR[\ell,\lambda],~ 1\le \ell\le L,~ 1\le \lambda \le m), \qquad
		      \barR=(\barR[\ell,\lambda],~ 1\le \ell\le L,~ 1\le \lambda \le m) .
	      $$
	      They are used in the same way as the previous subsection.

	\item 2-dimensional \emph{pointer arrays $\tH, \barH$.}
	      Their entries are
	      $$
		      \tH=(\tH[\ell,\lambda],~ 1\le \ell\le L,~ 1\le \lambda \le m), \qquad
		      \barH=(\barH[\ell,\lambda],~ 1\le \ell\le L,~ 1\le \lambda \le m) .
	      $$
	      These two pointer arrays serve as backups of $\tR$ and $\barR$. We use $\tH, \barH$ when all the entries in $\tR$ and $\barR$ are occupied.

	\item \emph{priority queue $\texttt{PriQue}$.}
	      $\texttt{PriQue}$ is defined essentially in the same way as the previous subsection. The only difference is that each element in the queue changes from a triple $(\ell, b, \prob)$ to a quadruple $(\ell, a, b, \prob)$.
	      The quadruple $(\ell, a, b, \prob)$ has the following meaning:  When we decode $U_i$ and $U_{i+1}$ in the last layer $\lambda = m$, the (posterior) probability of $(U_i = a, U_{i+1}=b)$ in the $\ell$th decoding path is $\prob$.
\end{enumerate}

Below we list the main differences between the data structures in this subsection and the previous subsection.
\begin{enumerate}[(1)]
	\item The range of $\lambda$ in all the data structures changes from $0\le \lambda \le m$ (previous subsection) to $1\le \lambda \le m$ (this subsection).

	\item The dimension of the probability array $\tD$ changes from $4$ (previous subsection) to $5$ (this subsection).

	\item The range of the index $s$ in the array $\tB$ changes from $1\le s\le 2L$ (previous subsection) to $1\le s\le 4L$ (this subsection).

	\item We have two more pointer arrays $\tH, \barH$ in this subsection.

	\item Each element in the priority queue $\texttt{PriQue}$ changes from a triple $(\ell, b, \prob)$ to a quadruple $(\ell, a, b, \prob)$.
\end{enumerate}

For the SCL decoder presented in this subsection, the $\ell$th list element has the following fields:
\begin{equation}\label{eq:stDB_list_field}
	\begin{aligned}
		 & (\tP[\ell,1], \tP[\ell,2], \dots, \tP[\ell, m]), \\
		 & (\tR[\ell,1], \tR[\ell,2], \dots, \tR[\ell, m]), \\
		 & (\tH[\ell,1], \tH[\ell,2], \dots, \tH[\ell, m]), \\
		 & \score[\ell] .
	\end{aligned}
\end{equation}

We still use the function $\texttt{allocate\_prob}(\lambda)$ in Algorithm~\ref{algo:st_allocate_prob} although the range of $\lambda$ is $\{1,2,\dots, m\}$ in this subsection.
However, we will use the function $\texttt{allocate\_bit}$ in Algorithm~\ref{algo:allocate_bit} for the new decoder in this subsection, which is different from the function with the same name in Algorithm~\ref{algo:st_allocate_bit}. The main difference is that the function $\texttt{allocate\_bit}$ in Algorithm~\ref{algo:allocate_bit} has an extra input parameter $k$, which takes value in $\{1,2\}$. In this subsection, the decoder makes decisions according to the transition probabilities of adjacent-bits-channels. Each adjacent-bits-channel has two input bits. In some cases we only decode one bit while in other cases we need to decode both bits. The input parameter $k$ in Algorithm~\ref{algo:allocate_bit} corresponds to the number of input bits we need to decode for each adjacent-bits-channel. We do not have the parameter $k$ in Algorithm~\ref{algo:st_allocate_bit} because each bit-channel only has one input bit.


\begin{algorithm}[ht]
	\DontPrintSemicolon
	\caption{$\texttt{allocate\_bit}(\lambda, k)$}
	\label{algo:allocate_bit}
	\KwIn{layer $\lambda\in\{1, 2, \dots, m\}$ and an integer $k\in\{1,2\}$}
	\KwOut{a pointer to the allocated memory}

	$s \gets \numB[\lambda] + 1$

	$\numB[\lambda]\gets \numB[\lambda] + k$

	\Return $\vec{\tB}[\lambda, s]$

\end{algorithm}

We present the main function $\texttt{decode}$ in Algorithm~\ref{algo:ABS_Decoding}.
The first 3 lines initialize the parameters. In Line~4, we assign the pointer $\vec{\tD}[1,1]$ to $\tP[1, 1]$ and update the value of $\numD[1]$ to be $1$. In Lines~5--7, we calculate the transition probabilities $V_1^{(2)}(y_\beta,y_{\beta+n/2}|a,b)$ for $1\le\beta\le n/2$ and $a,b\in\{0,1\}$ using \eqref{eq:v_init} and store $V_1^{(2)}(y_\beta,y_{\beta+n/2}|a,b)$ in $\tP[1,1][\beta,a,b]$. Line~8 executes recursive decoding which we will explain later.
After recursive decoding, we obtain $L_c$ list elements. In the $\ell$th list element, $\score[\ell]$ is the transition probability which measures the likelihood of this list element. In Lines~9--14, we pick the list element with the maximum $\score[\ell]$. Recall that $\hat x_{i,\beta}^{(\ell,\lambda)}$ is the decoded value of $X_{i,\beta}^{(\lambda)}$ in the $\ell$th list element.
As we will prove in Lemma~\ref{lm:566} below,
after recursive decoding, we have
$$
	\tR[\ell,1][\beta] = \hat x_{1,\beta}^{(\ell,1)} ,
	\quad
	\tR[\ell,1][\beta+n/2] = \hat x_{2,\beta}^{(\ell,1)}
	\quad
	\text{for~} 1\le\ell\le L_c \text{~and~} 1\le\beta\le n/2 .
$$
Since the codeword vector $(X_1,\dots,X_n)$ and the intermediate vectors $(X_{1,1}^{(1)},\dots,X_{1,n/2}^{(1)}), (X_{2,1}^{(1)},\dots,X_{2,n/2}^{(1)})$ satisfy
$$
	X_{\beta} = X_{1,\beta}^{(1)} + X_{2,\beta}^{(1)} , \quad
	X_{\beta+n/2} = X_{2,\beta}^{(1)}
	\quad
	\text{for~} 1\le\beta\le n/2 ,
$$
we further have
$$
	\hat x_\beta^{(\ell)} = \tR[\ell,1][\beta] + \tR[\ell,1][\beta + n/2] , \quad
	\hat x_{\beta+n/2}^{(\ell)} = \tR[\ell,1][\beta + n/2]
	\quad
	\text{for~} 1\le\beta\le n/2 ,
$$
where $(\hat x_1^{(\ell)},\hat x_2^{(\ell)},\dots,\hat x_n^{(\ell)})$ is the decoding result of the codeword vector in the $\ell$th list element. This is how we calculate the final decoding result in Lines~15--17.

\begin{algorithm}[ht]
	\DontPrintSemicolon
	\caption{\texttt{Decode}$((y_1,y_2,\dots,y_n)) \quad$}
	\label{algo:ABS_Decoding}
	\KwIn{the received vector $(y_1,y_2,\dots,y_n)\in\cY^n$}
	\KwOut{the decoded codeword $(\hat x_1,\hat x_2,\dots,\hat x_n)\in\{0,1\}^n$}

	\For{\em $\lambda \in \{1, 2, \dots, m\}$}{
		$\numD[\lambda] \gets \numB[\lambda] \gets 0$
	}

	$L_c\gets 1$

	$\tP[1, 1]\gets \texttt{allocate\_prob}(1)$

	\For{\em $\beta \in \{1, 2, \dots, n/2\}$}{
		\For{\em $a\in \{0,1\}$, $b\in \{0,1\}$}{
			$\tP[1,1][\beta,a,b]\gets W(y_\beta|a+b)\cdot W(y_{\beta+n/2}|b)$
		}
	}

	\texttt{decode\_channel$(1, 1)$}
	\Comment{Recursive decoding}

	$\max\_\score\gets 0$

	$\max\_\ell\gets 0$

	\For{$\ell\in\{1,2,\dots, L_c\}$}{
		\If{$\score[\ell]\ge \max\_\score$}{
			$\max\_\score\gets \score[\ell]$

			$\max\_\ell\gets \ell$
		}
	}

	\For{$\beta=1,2,\dots,n/2$}
	{
		$\hat x_\beta \gets \tR[\max\_\ell,1][\beta] + \tR[\max\_\ell,1][\beta + n/2]$

		$\hat x_{\beta+n/2} \gets \tR[\max\_\ell,1][\beta + n/2]$
	}

	\Return $(\hat x_1,\hat x_2,\dots,\hat x_n)$
\end{algorithm}

The recursive decoding function $\texttt{decode\_channel}$ in Algorithm~\ref{algo:STDB_decode_channel} has two branches. If $\lambda=m$, we call the function $\texttt{decode\_boundary\_channel}$  in  Algorithm~\ref{algo:decode_boundary_channel}.
If $\lambda < m$, we call the function $\texttt{decode\_original\_channel}$ in Algorithm~\ref{algo:st_decode_ori_channel}.
In  Algorithms~\ref{algo:decode_boundary_channel}--\ref{algo:st_decode_ori_channel}, we only decode $(X_{i,\beta}^{(\lambda)}, 1\le \beta\le 2^{m-\lambda})$ if $i\le 2^{\lambda}-2$; we decode both $(X_{i,\beta}^{(\lambda)}, 1\le \beta\le 2^{m-\lambda})$ and $(X_{i+1,\beta}^{(\lambda)}, 1\le \beta\le 2^{m-\lambda})$ if $i=2^\lambda-1$.
The following lemma further explains how  Algorithms~\ref{algo:STDB_decode_channel},\ref{algo:decode_boundary_channel}--\ref{algo:st_decode_ori_channel} work.

\begin{lemma}  \label{lm:566}
	Suppose that $1\le \lambda \le m$ and $1\le i\le 2^\lambda-1$.
	Before we call the function $\texttt{decode\_channel}$ in Algorithm~\ref{algo:STDB_decode_channel} with input parameters $(\lambda, i)$, the pointer $\tP[\ell, \lambda]$ satisfies that
	\begin{equation}\label{eq:1prob}
		\tP[\ell, \lambda][\beta, a, b] = V_{i}^{(2^\lambda)}(\hat{\mathbi o}_{i,\beta}^{(\ell,\lambda)} | a, b)
		\quad  \text{for all~} 1\le \ell\le L_c,~ 1\le \beta \le 2^{m-\lambda}
		\text{~and~} a,b\in\{0,1\} .
	\end{equation}
	After the function $\texttt{decode\_channel}(\lambda, i)$ in Algorithm~\ref{algo:STDB_decode_channel} returns, the pointer $\tR[\ell, \lambda]$ satisfies that
	\begin{equation} \label{eq:1bit}
		\tR[\ell, \lambda][\beta] = \hat{x}_{i,\beta}^{(\ell, \lambda)}
		\quad  \text{for all~} 1\le \ell\le L_c \text{~and~}
		1\le \beta\le 2^{m-\lambda}.
	\end{equation}
	Moreover, if $i=2^\lambda-1$, then the pointer $\tR[\ell, \lambda]$ further satisfies that
	\begin{equation} \label{eq:2bit}
		\tR[\ell, \lambda][\beta+2^{m-\lambda}] = \hat{x}_{i+1,\beta}^{(\ell, \lambda)}
		\quad  \text{for all~} 1\le \ell\le L_c \text{~and~}
		1\le \beta\le 2^{m-\lambda}.
	\end{equation}
\end{lemma}

\begin{proof}
	We first prove \eqref{eq:1bit}--\eqref{eq:2bit} by induction. Algorithm~\ref{algo:decode_boundary_channel} deals with the base case $\lambda=m$. Recall from \eqref{eq:subscript} that $X_{i,1}^{(m)}=X_i^{(n)}$ and $X_{i+1,1}^{(m)}=X_{i+1}^{(n)}$ when $\lambda=m$. By \eqref{eq:xnn}, we further obtain $X_{i,1}^{(m)}=U_i$ and $X_{i+1,1}^{(m)}=U_{i+1}$. The discussion below is divided into two cases. {\bf Case (1) $i\le n-2$:}
	If $U_i$ is a frozen bit, then Line~11 of Algorithm~\ref{algo:decode_boundary_channel} immediately implies \eqref{eq:1bit}. If $U_i$ is an information bit, then we explore both decoding paths $U_i=0$ and $U_i=1$ for every list element, where the variable $a$ in Lines~4--6 represents the (possible) value of $U_i$. The question mark ``?" in Line~6 means that we do not need to decode $U_{i+1}$ when $i\le n-2$. Since we expand the current list size by a factor of $2$ in Lines~4--6, the current list size might exceed the prescribed upper bound $L$. In this case, we prune the list according to $\score[\ell]$ in Lines~29--44. The variables $a$ and $b$ in Lines~29--44 represent the decoded values of $U_i$ and $U_{i+1}$ in each list element, respectively. In Line~42, we use $\barR[\ell, m][1]$ to temporarily store the decoding result of $U_i$ in the $\ell$th list element, and we use $\barR[\ell, m][2]$ to temporarily store the decoding result of $U_{i+1}$ in the $\ell$th list element. Next we swap $\barR$ and $\tR$ in Line~44, so \eqref{eq:1bit}--\eqref{eq:2bit} are satisfied.
		{\bf Case (2) $i= n-1$:} This case is handled in Lines~12--28. Note that $U_{n-1}$ and $U_n$ in Lines~17,21,25 refer to their frozen values (or true values).
	The argument for Case (2) is similar to Case (1), and we do not repeat it here.

	For the inductive step, we assume that \eqref{eq:1bit}--\eqref{eq:2bit} hold for $\lambda+1$ and prove them for $\lambda$. By this induction hypothesis, after executing Lines~1--5 of Algorithm~\ref{algo:st_decode_ori_channel}, we have
	$$
		\tH[\ell, \lambda+1][\beta] = \hat{x}_{2i-1,\beta}^{(\ell, \lambda+1)}
		\quad  \text{for all~} 1\le \ell\le L_c \text{~and~}
		1\le \beta\le 2^{m-\lambda-1}.
	$$
	In Line~1, we set $n_c=2^\lambda$. We again divide the discussion into two cases.
		{\bf Case (1) $i\le n_c-2$:} In this case, we only need to prove \eqref{eq:1bit}. After executing Lines~7--10 and Line~14, we have
	$$
		\temppointer[\beta] = \hat{x}_{2i,\beta}^{(\ell, \lambda+1)}
		\quad  \text{for all~} 1\le \ell\le L_c \text{~and~}
		1\le \beta\le 2^{m-\lambda-1}.
	$$
	Since $n_c=2^\lambda$, we have $n/(2n_c)=2^{m-\lambda-1}$. Therefore, Lines~18--19 of Algorithm~\ref{algo:st_decode_ori_channel} become \eqref{eq:R+-}. Combining \eqref{eq:R+-} with \eqref{eq:fthat}, we finish the proof of \eqref{eq:1bit} for Case (1).
		{\bf Case (2) $i= n_c-1$:} In this case, we  need to prove both \eqref{eq:1bit} and \eqref{eq:2bit}. The proof of \eqref{eq:1bit} is exactly the same as Case (1). To prove \eqref{eq:2bit}, we observe that if $i= n_c-1$, then $2i+1=2n_c-1=2^{\lambda+1}-1$. Then by the induction hypothesis, after executing Line~24, we have
	\begin{align*}
		\tR[\ell, \lambda+1][\beta] = \hat{x}_{2i+1,\beta}^{(\ell, \lambda+1)} , \quad
		\tR[\ell, \lambda+1][\beta+2^{m-\lambda-1}] = \hat{x}_{2i+2,\beta}^{(\ell, \lambda+1)} \\
		\text{for all~} 1\le \ell\le L_c \text{~and~}
		1\le \beta\le 2^{m-\lambda-1}.
	\end{align*}
	Therefore, Lines~32--34 become
	\begin{equation} \label{eq:lmd+1}
		\begin{aligned}
			\tR[\ell, \lambda][\beta+2^{m-\lambda}] = \hat{x}_{2i+1,\beta}^{(\ell, \lambda+1)} + \hat{x}_{2i+2,\beta}^{(\ell, \lambda+1)}, \quad
			\tR[\ell, \lambda][\beta+2^{m-\lambda-1}+2^{m-\lambda}] = \hat{x}_{2i+2,\beta}^{(\ell, \lambda+1)} \\
			\text{for all~} 1\le \ell\le L_c \text{~and~}
			1\le \beta\le 2^{m-\lambda-1}.
		\end{aligned}
	\end{equation}
	Replacing $i$ with $i+1$ in \eqref{eq:fthat} we obtain
	\begin{align*}
		\hat{x}_{i+1,\beta}^{(\ell, \lambda)} = \hat{x}_{2i+1,\beta}^{(\ell, \lambda+1)} + \hat{x}_{2i+2,\beta}^{(\ell, \lambda+1)}, \quad
		\hat{x}_{i+1,\beta+2^{m-\lambda-1}}^{(\ell, \lambda)} = \hat{x}_{2i+2,\beta}^{(\ell, \lambda+1)} \\
		\text{for all~} 1\le \ell\le L_c \text{~and~}
		1\le \beta\le 2^{m-\lambda-1}.
	\end{align*}
	Combining this with \eqref{eq:lmd+1}, we complete the proof of \eqref{eq:2bit}.

	Next we prove \eqref{eq:1prob} by induction. This time the base case is $\lambda=1$, and this case only occurs once in Line~8 of Algorithm~\ref{algo:ABS_Decoding} during the whole decoding procedure.
	By \eqref{eq:v_init}, we have $V_1^{(2)}(y_\beta, y_{\beta+n/2}| a,b) = W(y_\beta|a+b)\cdot W(y_{\beta+n/2}|b)$.
	Therefore, Lines~5--7 of Algorithm~\ref{algo:ABS_Decoding} immediately imply \eqref{eq:1prob} for $\lambda=1$.

	For the inductive step, we assume that \eqref{eq:1prob} holds for $\lambda$ and prove it for $\lambda+1$. By this induction hypothesis, \eqref{eq:1prob} holds for $\lambda$ when we execute Line~2 of Algorithm~\ref{algo:st_decode_ori_channel}. In other words, the array associated with the pointer $\tP[\ell, \lambda]$ stores the transition probabilities of $V_i^{(2^\lambda)}$. By Lemma~\ref{lemma:recur_ST_DB}, $V_{2i-1}^{(2^{\lambda+1})}$ is the ``$\oria$" transform of $V_i^{(2^\lambda)}$. The function \texttt{calculate\_$\oria$\_transform}$(\lambda+1)$ calculates the ``$\oria$" transform of $V_i^{(2^\lambda)}$ and stores the results in the array associated with the pointer $\tP[\ell, \lambda+1]$, so \eqref{eq:1prob} holds before we call $\texttt{decode\_channel}$ in Line~3 of Algorithm~\ref{algo:st_decode_ori_channel}. Again by Lemma~\ref{lemma:recur_ST_DB}, $V_{2i}^{(2^{\lambda+1})}$ is the ``$\orib$" transform of $V_i^{(2^\lambda)}$. The function \texttt{calculate\_$\orib$\_transform}$(\lambda+1)$ in Line~7 of Algorithm~\ref{algo:st_decode_ori_channel} calculates the ``$\orib$" transform of $V_i^{(2^\lambda)}$ and stores the results in the array associated with the pointer $\tP[\ell, \lambda+1]$, so \eqref{eq:1prob} holds before we call $\texttt{decode\_channel}$ in Line~8 of Algorithm~\ref{algo:st_decode_ori_channel}.
	Using exactly the same method, we can show that \eqref{eq:1prob} also holds before we call $\texttt{decode\_channel}$ in Line~24 of Algorithm~\ref{algo:st_decode_ori_channel}.

	During the whole decoding procedure, the function $\texttt{decode\_channel}$ is only called in Line~8 of Algorithm~\ref{algo:ABS_Decoding} and Lines~3,8,24 of Algorithm~\ref{algo:st_decode_ori_channel}. We have proved that \eqref{eq:1prob} holds for all four places. This completes the proof of the lemma.
\end{proof}

\begin{algorithm}[ht]
	\DontPrintSemicolon
	\caption{\texttt{decode\_channel$(\lambda, i)$}}
	\label{algo:STDB_decode_channel}
	\KwIn{layer $\lambda\in \{1,2,\dots, m\}$ and index $i\in \{1, 2, \dots, 2^\lambda-1\}$}

	\uIf{$\lambda = m$}{
		\texttt{decode\_boundary\_channel}$(i)$
		\Comment{Algorithm~\ref{algo:decode_boundary_channel}}
	}
	\Else{
		\texttt{decode\_original\_channel}$(\lambda, i)$
		\Comment{Algorithm~\ref{algo:st_decode_ori_channel}}
	}

	$\numD[\lambda]\gets 0$

	\Return

\end{algorithm}

Before proceeding further, let us explain the meaning of the boolean variable ``flag" in Algorithm~\ref{algo:decode_boundary_channel}. flag takes value $0$ if we do not expand the decoding list in the decoding procedure, and it takes value $1$ otherwise. In Algorithm~\ref{algo:decode_boundary_channel}, we do not expand the decoding list if and only if we only decode frozen bits. There are two such cases, one in Lines~7--11 and the other in Lines~13--17. We set the variable flag to be $0$ in both cases. In all the other cases, we need to decode at least one information bit, and we need to expand the list size by a factor of at least $2$, so we set the variable flag to be $1$ in all the other cases.
If flag$=0$, then the list size does not change, and we do not need to prune the list. Therefore, we only prune the list when flag$=1$; see Line~29.

\begin{remark}
	The calculations in Line 6 of Algorithm~\ref{algo:calcu_oria_proba} correspond to the "$\oria$" transform in Fig.~\ref{fig:DBpt} and the first equation in \eqref{eq:DBpt}.
	The calculations in Line 7 of Algorithm~\ref{algo:calcu_orib_proba} correspond to the "$\orib$" transform in Fig.~\ref{fig:DBpt} and the second equation in \eqref{eq:DBpt}.
	The calculations in Lines 7-8 of Algorithm~\ref{algo:calcu_oric_proba} correspond to the "$\oric$" transform in Fig.~\ref{fig:DBpt} and the third equation in \eqref{eq:DBpt}.
	This is why we say that the SCL decoder presented in this subsection is based on the DB polar transform.
\end{remark}

\begin{algorithm}[ht]
	\DontPrintSemicolon
	\caption{\texttt{calculate\_$\oria$\_transform}$(\lambda)$}
	\label{algo:calcu_oria_proba}
	\KwIn{layer $2\le \lambda \le m$}
	\KwOut{Update the entries pointed by $\tP[\ell, \lambda]$, $1\le \ell\le L_c $ }

	$\bar{n}_c \gets 2^{m-\lambda}$

	\For{$\ell\in\{1, 2, \dots, L_c\}$}{

		$\tP[\ell,\lambda]\gets \texttt{allocate\_prob}(\lambda)$

		\For{\em $\beta \in \{1,2,\dots,\bar{n}_c\}, r_1,r_2\in\{0,1\}$}{
			$\beta' \gets \beta + \bar{n}_c$

			$\tP[\ell, \lambda][\beta,r_1,r_2] \gets$ $\frac14\sum_{r_3, r_4\in\{0,1\}}$
			$\tP[\ell, \lambda-1][\beta, r_1+r_2, r_3+r_4]\tP[\ell, \lambda-1][\beta', r_2, r_4]$
		}
	}

	\Return
\end{algorithm}

\begin{algorithm}[ht]
	\DontPrintSemicolon
	\caption{\texttt{calculate\_$\orib$\_transform}$(\lambda)$}
	\label{algo:calcu_orib_proba}
	\KwIn{layer $2\le \lambda \le m$}
	\KwOut{Update the entries pointed by $\tP[\ell, \lambda]$, $1\le \ell\le L_c $ }

	$\bar{n}_c \gets 2^{m-\lambda}$


	\For{$\ell\in\{1, 2, \dots, L_c\}$}{

		$\tP[\ell,\lambda]\gets \texttt{allocate\_prob}(\lambda)$

		\For{\em $\beta \in \{1,2,\dots,\bar{n}_c\}, r_2,r_3\in\{0,1\}$}{
			$r_1\gets \tH[\ell, \lambda][\beta]$

			$\beta' \gets \beta + \bar{n}_c$

			$\tP[\ell, \lambda][\beta,r_2,r_3] \gets$ $\frac14\sum_{r_4\in\{0,1\}}$
			$\tP[\ell, \lambda-1][\beta, r_1+r_2, r_3+r_4]\tP[\ell, \lambda-1][\beta', r_2, r_4]$
		}
	}

	\Return

\end{algorithm}

\begin{algorithm}[ht]
	\DontPrintSemicolon
	\caption{\texttt{calculate\_$\oric$\_transform}$(\lambda)$}
	\label{algo:calcu_oric_proba}
	\KwIn{layer $2\le \lambda \le m$}
	\KwOut{Update the entries pointed by $\tP[\ell, \lambda]$, $1\le \ell\le L_c $ }

	$\bar{n}_c \gets 2^{m-\lambda}$


	\For{$\ell\in\{1, 2, \dots, L_c\}$}{

		$\tP[\ell,\lambda]\gets \texttt{allocate\_prob}(\lambda)$

		\For{\em $\beta \in \{1,2,\dots,\bar{n}_c\}, r_3,r_4\in\{0,1\}$}{
			$r_1\gets \tH[\ell, \lambda][\beta],\quad r_2\gets \tR[\ell, \lambda-1][\beta]$

			$\beta' \gets \beta + \bar{n}_c$

			$\tP[\ell, \lambda][\beta,r_3,r_4] \gets$\\
			$\frac14\tP[\ell, \lambda-1][\beta, r_1+r_2, r_3+r_4]\tP[\ell, \lambda-1][\beta', r_2, r_4]$
		}
	}

	\Return
\end{algorithm}

\begin{algorithm}[ht]
	\DontPrintSemicolon
	\caption{\texttt{decode\_boundary\_channel$(i)$}}
	\label{algo:decode_boundary_channel}
	\KwIn{index $i$ in the last layer $(\lambda = m)$}

	flag $\gets 1$

	\eIf(\Comment{Only decode $U_i$}){$i\le n-2$}{

		\eIf(\Comment{$U_i$ is an information bit}){$i\in\cA$}{

			\For{$\ell \in \{1,2,\dots, L_c\}$, $a\in\{0,1\}$}{
				$\prob \gets\frac12\sum_{b\in\{0,1\}}\tP[\ell, m][1,a,b]$

				$\texttt{PriQue.push}(\ell, a, ``?",\prob)$
			}
		}(\Comment{$U_i$ is a frozen bit}){
			flag $\gets 0$

			\For{$\ell \in \{1,2,\dots, L_c\}$}{
				$\tR[\ell, m]\gets \texttt{allocate\_bit}(m, 1)$

				$\tR[\ell, m][1]\gets$ frozen value of $U_i$
			}
		}

	}(\Comment{Decode both $U_{n-1}$ and $U_n$.}){
		\uIf{$n-1, n\notin \cA$}{
			flag $\gets 0$
			\Comment{$U_{n-1}$ and $U_{n}$ are both frozen bits}

			\For{$\ell\in\{1,2,\dots, L_c\}$}{
				$\tR[\ell,m]\gets \texttt{allocate\_bit}(m,2)$

				$(\tR[\ell,m][1],\tR[\ell, m][2])\gets(U_{n-1}, U_{n})$
			}
		}\ElseIf{$n-1\in \cA$, $n\notin \cA$}{
			\Comment{information bit $U_{n-1}$, frozen bit $U_{n}$ }

			\For{$\ell\in\{1,2,\dots, L_c\}$, $a\in\{0,1\}$}{
				$\texttt{PriQue.push}(\ell, a, U_{n},\tP[\ell,m][1, a, U_{n}])$
			}

		}\ElseIf{$n-1\notin \cA$, $n\in \cA$}{
			\Comment{frozen bit $U_{n-1}$, information bit $U_{n}$}

			\For{$\ell\in\{1,2,\dots, L_c\}$, $b\in\{0,1\}$}{
				$\texttt{PriQue.push}(\ell, U_{n-1}, b,\tP[\ell,m][1, U_{n-1}, b])$
			}
		}\Else(\Comment{$U_{n-1}$ and $U_{n}$ are both information bits}){

			\For{$\ell\in\{1,2,\dots, L_c\}$, $a,b\in\{0,1\}$}{
				$\texttt{PriQue.push}(\ell, a, b,\tP[\ell,m][1, a, b])$
			}
		}
	}
	\If{\em flag $ = 1$}{
		$L_c\gets \min\{L, \texttt{PriQue.size()}\}$

		\For{$\ell \in\{1,2,\dots, L_c\}$}{
			$(\ell', a, b, \score[\ell])\gets \texttt{PriQue.pop()}$

			\For{$\lambda \in\{1,2,\dots, m-1\}$}{
				$\barP[\ell,\lambda]\gets \tP[\ell',\lambda]$

				$\barR[\ell,\lambda]\gets \tR[\ell',\lambda]$

				$\barH[\ell,\lambda]\gets \tH[\ell',\lambda]$
			}

			\eIf{$i < n-1$}{
				$\barR[\ell, m]\gets \texttt{allocate\_bit}(m, 1)$

				$\barR[\ell, m][1]\gets a$
			}{
				$\barR[\ell, m]\gets \texttt{allocate\_bit}(m, 2)$

				$(\barR[\ell,m][1],\barR[\ell, m][2]) \gets (a,b)$

			}
		}

		$\texttt{PriQue.clear()}$

		$\texttt{swap}(\barP, \tP)$,\quad $\texttt{swap}(\barR, \tR)$,\quad $\texttt{swap}(\barH, \tH)$
	}
	\Return
\end{algorithm}

\begin{algorithm}[ht]
	\DontPrintSemicolon
	\caption{\texttt{decode\_original\_channel$(\lambda, i)$}}
	\label{algo:st_decode_ori_channel}
	\KwIn{$\lambda \in \{1,2, \dots, m\}$ and index $i$ satisfying $1\le i \le 2^{\lambda}-1$}

	$n_c\gets 2^{\lambda}$
	\Comment{$V_{2i-1}^{(2n_c)} = (V_{i}^{(n_c)})^\oria$}

	\texttt{calculate\_$\oria$\_transform}$(\lambda+1)$

	\texttt{decode\_channel}$(\lambda+1, 2i-1)$

	\For{$\ell\in\{1, 2, \dots, L_c\}$}{
		$\tH[\ell, \lambda + 1]\gets \tR[\ell, \lambda+1]$
	}

	\Comment{$V_{2i}^{(2n_c)} = (V_{i}^{(n_c)})^\orib$}

	\texttt{calculate\_$\orib$\_transform}$(\lambda+1)$

	\texttt{decode\_channel}$(\lambda+1, 2i)$

	\For{$\ell\in\{1, 2, \dots, L_c\}$}{
		$\tR[\ell, \lambda]\gets \tR[\ell, \lambda+1]$
	}
	\vspace*{.1in}

	\eIf{$i \le n_c-2$}{
		\Comment{Only decode one bit ${X}_{i,\beta}^{(\lambda)}$ for each $\beta$}

		\For{$\ell\in\{1, 2, \dots, L_c\}$}{
			$\temppointer \gets \tR[\ell, \lambda]$

			$\tR[\ell, \lambda] \gets \texttt{allocate\_bit}(\lambda, 1)$

			\For{\em $\beta \in \{1,2,\dots,n/(2n_c)\}$}{
				$\beta' \gets \beta + n/(2n_c)$

				$\tR[\ell, \lambda][\beta] \gets \tH[\ell, \lambda+1][\beta] + \temppointer[\beta]$

				$\tR[\ell, \lambda][\beta'] \gets \temppointer[\beta]$

			}
		}
	}{
		\Comment{Decode two bits ${X}_{n_c-1,\beta}^{(\lambda)}, {X}_{n_c,\beta}^{(\lambda)}$ for each $\beta$}

		\Comment{$V_{2i+1}^{(2n_c)} = (V_{i}^{(n_c)})^\oric$}

		\texttt{calculate\_$\oric$\_transform}$(\lambda+1)$

		\texttt{decode\_channel}$(\lambda+1, 2i+1)$

		\For{$\ell\in\{1, 2, \dots, L_c\}$}{
			$\temppointer \gets \tR[\ell, \lambda]$

			$\tR[\ell, \lambda] \gets \texttt{allocate\_bit}(\lambda, 2)$

			\For{\em $\beta \in \{1,2,\dots,n/(2n_c)\}$}{
				$\beta' \gets \beta + n/(2n_c)$

				$\tR[\ell, \lambda][\beta] \gets \tH[\ell, \lambda+1][\beta] + \temppointer[\beta]$

				$\tR[\ell, \lambda][\beta'] \gets \temppointer[\beta]$

				$\tR[\ell, \lambda][\beta + n/(n_c)] \gets$\\
				$\qquad\tR[\ell, \lambda+1][\beta] + \tR[\ell, \lambda+1][\beta']$

				$\tR[\ell, \lambda][\beta' + n/(n_c)] \gets \tR[\ell, \lambda+1][\beta']$

			}
		}
	}

	$\numB[\lambda+1]\gets 0$

	\Return

\end{algorithm}

The next lemma shows that the data structures $\tD$ and $\tB$ are large enough to store the transition probabilities and the decoding results of the intermediate vectors throughout the decoding procedure.

\begin{lemma}\label{lemma:stdb_space}
	Throughout the whole decoding procedure, we have $\numD[\lambda]\le L$ and $\numB[\lambda]\le 4L$ for all $1\le \lambda\le m$. The space complexity of the SCL decoder is $O(Ln)$.
\end{lemma}

\begin{proof}
	The proof of $\numD[\lambda]\le L$ is the same as Lemma~\ref{lemma:st_space}, and we do not repeat it.
	Now we prove $\numB[\lambda]\le 4L$. As we can see from Algorithms~\ref{algo:STDB_decode_channel},\ref{algo:decode_boundary_channel}--\ref{algo:st_decode_ori_channel}, each time we call the function \texttt{decode\_channel$(\lambda, i)$}, the value of $\numB[\lambda]$ increases by $L_c$ if $i\le 2^\lambda-2$, and it increases by $2L_c$ if $i=2^\lambda-1$.
	Since the input $i$ in Algorithm~\ref{algo:st_decode_ori_channel} satisfies $i\le 2^\lambda-1$, we have $2i-1<2i\le 2^{\lambda+1}-2$. Therefore, after executing Line~3 of Algorithm~\ref{algo:st_decode_ori_channel}, the value of $\numB[\lambda+1]$ increases by at most $L$. Similarly, after executing Line~8 of Algorithm~\ref{algo:st_decode_ori_channel}, the value of $\numB[\lambda+1]$ also increases by at most $L$. If $i=2^\lambda-1$, we will execute Line~24 of Algorithm~\ref{algo:st_decode_ori_channel}. In this case, $2i+1=2^{\lambda+1}-1$, so the value of $\numB[\lambda+1]$ increases by at most $2L$. Therefore, before we reset $\numB[\lambda+1]$ to $0$ in Line~35 of Algorithm~\ref{algo:st_decode_ori_channel}, its value is at most $4L$. This proves $\numB[\lambda]\le 4L$.

	The proof of the space complexity is the same as Lemma~\ref{lemma:st_space}.
\end{proof}

In TABLE~\ref{tb:stdb_space}, we list the upper bound of $\numB[\lambda+1]$ at the starting point and the end of the function $\texttt{decode\_channel}(\lambda+1, j)$. The starting point refers to the moment we call $\texttt{decode\_channel}(\lambda+1, j)$, and the end refers to the moment this function returns. These upper bounds come from the proof of Lemma~\ref{lemma:stdb_space}.

%
%

\begin{table}[H]
	\centering
	\begin{tabular}{|r|l|c|c|}
		\hline
		\multicolumn{2}{|c|}{cases}              & start      & end      \\
		\hline
		\multirow{2}*{$1\le i\le 2^{\lambda}-1$} & $j = 2i-1$ & 0   & L  \\
		\cline{2-4}
		~                                        & $j = 2i$   & L   & 2L \\
		\hline
		$ i=2^{\lambda}-1$                       & $j = 2i+1$ & 2L  & 4L \\
		\hline
	\end{tabular}
	\caption{The upper bound of $\numB[\lambda+1]$ at the starting point and the end of the function $\texttt{decode\_channel}(\lambda+1, j)$}
	\label{tb:stdb_space}
\end{table}

\begin{proposition}
	The decoding time complexity of standard polar codes based on the DB polar transform is $O(Ln\log(n))$.
\end{proposition}

\subsection{SCL decoder for ABS polar codes}\label{sect:ABS_decoder}

\begin{figure*}
	\centering
	\includegraphics[scale = 0.6]{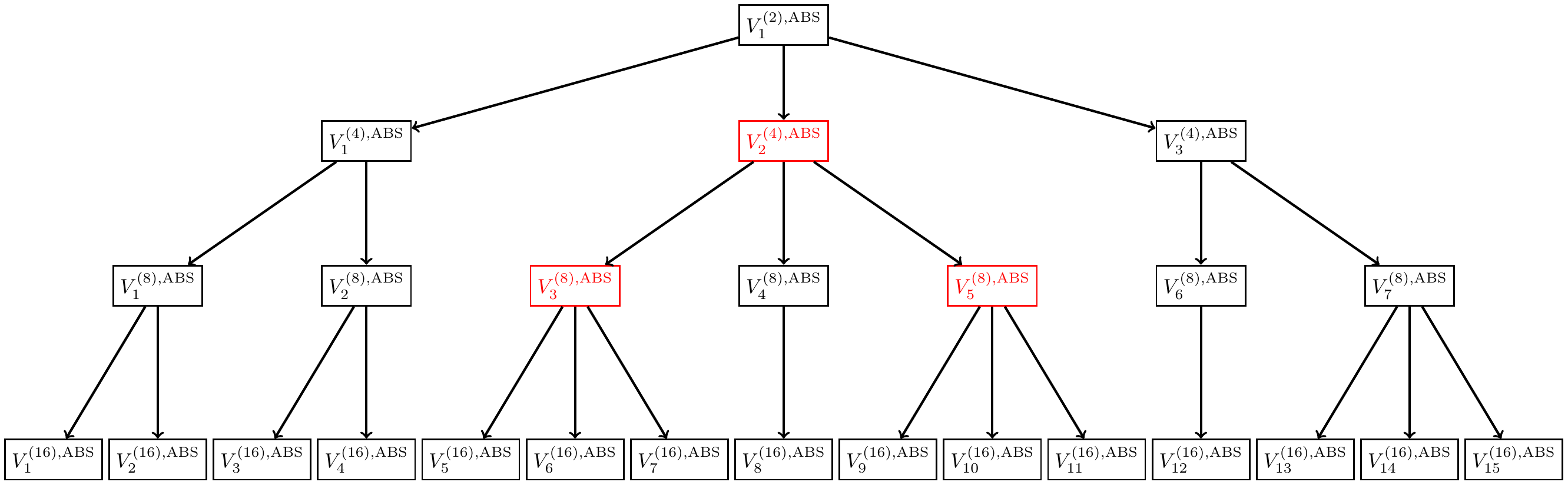}
	\caption{Recursive decoding of the ABS polar code defined in Fig.~\ref{fig:example_enc}.
	We put $V_{i}^{(2^\lambda),\ABS}$ in a black block (e.g., $V_1^{(2),\ABS}$) if $2i\notin\cI^{(2^{\lambda+1})}$.
	In this case, \texttt{decode\_channel}$(\lambda, i)$ in Algorithm~\ref{algo:decode_channel} calls $\texttt{decode\_original\_channel}(\lambda,i)$.
	We put $V_{i}^{(2^\lambda),\ABS}$ in a red block (e.g., $V_2^{(4),\ABS}$) if $2i\in\cI^{(2^{\lambda+1})}$.
	In this case, \texttt{decode\_channel}$(\lambda, i)$ calls $\texttt{decode\_swapped\_channel}(\lambda,i)$. An arrow from $V_{i}^{(2^\lambda),\ABS}$ to $V_{i_1}^{(2^{\lambda_1}),\ABS}$ means that \texttt{decode\_channel}$(\lambda_1, i_1)$ is called in the execution of \texttt{decode\_channel}$(\lambda, i)$. For example, we call \texttt{decode\_channel} with input parameters $(2,1),(2,2)$ and $(2,3)$ in the execution of \texttt{decode\_channel}$(1,1)$.
	}
	\label{fig:example_dec}
\end{figure*}

\begin{algorithm}[ht]
	\DontPrintSemicolon
	\caption{\texttt{decode\_channel$(\lambda, i)$} }
	\label{algo:decode_channel}
	\KwIn{layer $\lambda\in \{1,2,\dots, m\}$ and index $i\in \{1, 2, \dots, 2^\lambda-1\}$}

	\uIf{$\lambda = m$}{
		\texttt{decode\_boundary\_channel}$(i)$
		\Comment{Algorithm~\ref{algo:decode_boundary_channel}}
	}
	\ElseIf{$2i\notin \cI^{(2^{\lambda + 1})}$}{
		\texttt{decode\_original\_channel}$(\lambda, i)$
		\Comment{Algorithm~\ref{algo:decode_ori_channel}}
	}\ElseIf{$2i\in \cI^{(2^{\lambda + 1})}$}{
		\texttt{decode\_swapped\_channel}$(\lambda, i)$
		\Comment{Algorithm~\ref{algo:decode_swp_channel}}
	}

	$\numD[\lambda]\gets 0$

	\Return

\end{algorithm}

In this subsection, we present the new SCL decoder for ABS polar codes.
This decoder is based on the DB polar transform in Fig.~\ref{fig:DBpt} and the SDB polar transform in Fig.~\ref{fig:SDBpt}.
Since we have the permutation matrices $\mP_{2}^{\ABS}, \mP_4^{\ABS} ,\dots, \mP_{n}^{\ABS}$ in the ABS polar code construction, we need to replace the recursive relation \eqref{eq:def_intermediate} with
\begin{equation}\label{eq:abs_intermediate}
	\begin{aligned}
		      & (X_{1}^{(2^\lambda)}, X_{2}^{(2^\lambda)},\dots, X_{n}^{(2^\lambda)})
		=     & (X_{1}^{(2^{\lambda+1})}, X_{2}^{(2^{\lambda+1})},\dots, X_{n}^{(2^{\lambda+1})})
		\cdot & \big((\mP_{2^{\lambda+1}}^{\ABS}(\mI_{2^\lambda}\otimes\mG_2^{\polar}))\otimes\mI_{2^{m-\lambda-1}}\big)
	\end{aligned}
\end{equation}
in order to define the intermediate vectors in ABS polar codes.
We still use the notation in \eqref{eq:subscript} and \eqref{eq:listdecoded}.

The data structures in this subsection are essentially the same as the ones in the previous subsection. There are only two minor differences:
\begin{enumerate}[(i)]
	\item We change the range of the index $s$ in \eqref{eq:B_index} from $1\le s\le 4L$ to $1\le s\le 6L$.

	\item In the integer array $\numB$, each entry $\numB[\lambda]$ takes value in $\{0,1,2,\dots,6L\}$ instead of $\{0,1,2,\dots,4L\}$.
\end{enumerate}
For the SCL decoder presented in this subsection, the fields of the $\ell$th list element are the same as the ones listed in \eqref{eq:stDB_list_field}.

The following functions are shared by the decoder in this subsection and the decoders in previous sections:
\begin{enumerate}[(1)]
	\item  \texttt{allocate\_prob} in Algorithm~\ref{algo:st_allocate_prob}

	\item  \texttt{allocate\_bit} in Algorithm~\ref{algo:allocate_bit}

	\item  $\texttt{decode}$ in Algorithm~\ref{algo:ABS_Decoding}. This is the main function of the decoder.

	\item  \texttt{calculate\_$\oria$\_transform} in Algorithm~\ref{algo:calcu_oria_proba}

	\item  \texttt{calculate\_$\orib$\_transform} in Algorithm~\ref{algo:calcu_orib_proba}

	\item  \texttt{calculate\_$\oric$\_transform} in Algorithm~\ref{algo:calcu_oric_proba}

	\item  \texttt{decode\_boundary\_channel} in Algorithm~\ref{algo:decode_boundary_channel}
\end{enumerate}

The following functions are solely used in this subsection. More precisely, either they appeared in previous subsections with the same name but with different implementations or they did not appear in previous subsections at all.

\begin{enumerate}[(1)]
	\item  \texttt{decode\_channel} in Algorithm~\ref{algo:decode_channel}. In Section~\ref{sect:ST_decoder_DB}, we also have the function \texttt{decode\_channel} in Algorithm~\ref{algo:STDB_decode_channel}, but the implementations in these two algorithms are different.

	\item \texttt{decode\_original\_channel} in Algorithm~\ref{algo:decode_ori_channel}. In Section~\ref{sect:ST_decoder_DB}, we also have the function \texttt{decode\_original\_channel} in Algorithm~\ref{algo:st_decode_ori_channel}, but the implementations in these two algorithms are different.

	\item \texttt{decode\_swapped\_channel} in Algorithm~\ref{algo:decode_swp_channel}. This function did not appear in previous subsections.

	\item  \texttt{calculate\_$\swpa$\_transform} in Algorithm~\ref{algo:calcu_swpa_proba}. This function did not appear in previous subsections.

	\item  \texttt{calculate\_$\swpb$\_transform} in Algorithm~\ref{algo:calcu_swpb_proba}. This function did not appear in previous subsections.

	\item  \texttt{calculate\_$\swpc$\_transform} in Algorithm~\ref{algo:calcu_swpc_proba}. This function did not appear in previous subsections.
\end{enumerate}

Although this subsection and the previous subsection share the same main function $\texttt{decode}$ in Algorithm~\ref{algo:ABS_Decoding}, the function \texttt{decode\_channel} in Line~8 of Algorithm~\ref{algo:ABS_Decoding} has different implementations in these two subsections.
More specifically, the function $\texttt{decode\_channel}$ in Algorithm~\ref{algo:decode_channel} has one more branch than $\texttt{decode\_channel}$ in Algorithm~\ref{algo:STDB_decode_channel}. The additional branch decodes swapped adjacent bits.

Algorithm~\ref{algo:decode_ori_channel} and Algorithm~\ref{algo:st_decode_ori_channel} are the implementations of $\texttt{decode\_original\_channel}$ for this subsection and the previous subsection, respectively. The difference between Algorithm~\ref{algo:decode_ori_channel} and Algorithm~\ref{algo:st_decode_ori_channel} is that we calculate the $\oria$ transform only when $2(i-1)\notin\cI^{(2^{\lambda+1})}$ in Algorithm~\ref{algo:decode_ori_channel}; see Lines~2--6. In contrast, we always calculate the $\oria$ transform in Algorithm~\ref{algo:st_decode_ori_channel}; see Lines~2--5. The reason behind this difference is given in Lemma~\ref{lemma:recur_ABS}: When $2(i-1)\in\cI^{(2^{\lambda+1})}$, we only have $V_{2i-1}^{(2^{\lambda+1}), \ABS} = \big(V_{i-1}^{(2^\lambda),\ABS}\big)^{\swpc}$, but $V_{2i-1}^{(2^{\lambda+1}), \ABS} = \big(V_{i}^{(2^\lambda),\ABS}\big)^{\oria}$ does not hold, so we do not calculate the $\oria$ transform in this case.

In Fig.~\ref{fig:example_dec}, we use the ABS polar code defined in Fig.~\ref{fig:example_enc} as a concrete example to illustrate the recursive structure of the function $\texttt{decode\_channel}$ in Algorithm~\ref{algo:decode_channel}.

\begin{lemma}  \label{lm:swp6}
	Suppose that $1\le \lambda \le m$ and $1\le i\le 2^\lambda-1$.
	Before we call the function $\texttt{decode\_channel}$ in Algorithm~\ref{algo:decode_channel} with input parameters $(\lambda, i)$, the pointer $\tP[\ell, \lambda]$ satisfies that
	\begin{equation}\label{eq:swp_prob}
		\tP[\ell, \lambda][\beta, a, b] = V_{i}^{(2^\lambda)}(\hat{\mathbi o}_{i,\beta}^{(\ell,\lambda)} | a, b)
		\quad  \text{for all~} 1\le \ell\le L_c,~ 1\le \beta \le 2^{m-\lambda}
		\text{~and~} a,b\in\{0,1\} .
	\end{equation}
	After the function $\texttt{decode\_channel}(\lambda, i)$ in Algorithm~\ref{algo:decode_channel} returns, the pointer $\tR[\ell, \lambda]$ satisfies that
	\begin{equation} \label{eq:swp_1bit}
		\tR[\ell, \lambda][\beta] = \hat{x}_{i,\beta}^{(\ell, \lambda)}
		\quad  \text{for all~} 1\le \ell\le L_c \text{~and~}
		1\le \beta\le 2^{m-\lambda}.
	\end{equation}
	Moreover, if $i=2^\lambda-1$, then the pointer $\tR[\ell, \lambda]$ further satisfies that
	\begin{equation} \label{eq:swp_2bit}
		\tR[\ell, \lambda][\beta+2^{m-\lambda}] = \hat{x}_{i+1,\beta}^{(\ell, \lambda)}
		\quad  \text{for all~} 1\le \ell\le L_c \text{~and~}
		1\le \beta\le 2^{m-\lambda}.
	\end{equation}
\end{lemma}

\begin{proof}
	The proof of \eqref{eq:swp_prob} is the same as that of \eqref{eq:1prob}. Here we only prove \eqref{eq:swp_1bit}--\eqref{eq:swp_2bit} by induction. The proof of the base case $\lambda=m$ relies on the analysis of Algorithm~\ref{algo:decode_boundary_channel}, which was already done in the proof of Lemma~\ref{lm:566}. For the inductive step, we assume that \eqref{eq:swp_1bit}--\eqref{eq:swp_2bit} hold for $\lambda+1$ and prove them for $\lambda$. This requires us to analyze Algorithm~\ref{algo:decode_ori_channel} for $2i\notin \cI^{(2^{\lambda + 1})}$ and analyze Algorithm~\ref{algo:decode_swp_channel} for $2i\in \cI^{(2^{\lambda + 1})}$. Algorithm~\ref{algo:decode_ori_channel} and Algorithm~\ref{algo:st_decode_ori_channel} are essentially the same. Since we have already analyzed Algorithm~\ref{algo:st_decode_ori_channel} in the proof of Lemma~\ref{lm:566}, we omit the analysis of Algorithm~\ref{algo:decode_ori_channel} here. We will focus on the analysis of Algorithm~\ref{algo:decode_swp_channel} for the rest of this proof.

	By the induction hypothesis, after executing Lines~1--5 of Algorithm~\ref{algo:decode_swp_channel}, we have
	$$
		\tH[\ell, \lambda+1][\beta] = \hat{x}_{2i-1,\beta}^{(\ell, \lambda+1)}
		\quad  \text{for all~} 1\le \ell\le L_c \text{~and~}
		1\le \beta\le 2^{m-\lambda-1}.
	$$
	After executing Lines~7--10 and Lines~17,27, we have
	$$
		\temppointer[\beta] = \hat{x}_{2i,\beta}^{(\ell, \lambda+1)}
		\quad  \text{for all~} 1\le \ell\le L_c \text{~and~}
		1\le \beta\le 2^{m-\lambda-1}.
	$$
	After executing Lines~12--13, we have
	\begin{equation} \label{eq:part1}
		\tR[\ell, \lambda+1][\beta] = \hat{x}_{2i+1,\beta}^{(\ell, \lambda+1)}
		\quad  \text{for all~} 1\le \ell\le L_c \text{~and~}
		1\le \beta\le 2^{m-\lambda-1}.
	\end{equation}
	In Line~1, we set $n_c=2^\lambda$. We again divide the discussion into two cases.
		{\bf Case (1) $i\le n_c-2$:} In this case, we only need to prove \eqref{eq:swp_1bit}.
	Since $n_c=2^\lambda$, we have $n/(2n_c)=2^{m-\lambda-1}$. Therefore, Lines~21--22 of Algorithm~\ref{algo:decode_swp_channel} become
	\begin{equation} \label{eq:swp_R1}
		\begin{aligned}
			\tR[\ell, \lambda][\beta] = \hat{x}_{2i-1,\beta}^{(\ell, \lambda+1)} + \hat{x}_{2i+1,\beta}^{(\ell, \lambda+1)} , \quad
			\tR[\ell, \lambda][\beta+2^{m-\lambda-1}] = \hat{x}_{2i+1,\beta}^{(\ell, \lambda+1)} \\
			\text{for all~} 1\le \ell\le L_c \text{~and~}
			1\le \beta\le 2^{m-\lambda-1}.
		\end{aligned}
	\end{equation}
	Equations \eqref{eq:abs_intermediate} and \eqref{eq:subscript} together imply that if $2i\in \cI^{(2^{\lambda + 1})}$, then
	\begin{align}
		 & X_{i,\beta}^{(\lambda)} = X_{2i-1,\beta}^{(\lambda+1)} + X_{2i+1,\beta}^{(\lambda+1)}, \quad
		X_{i,\beta+2^{m-\lambda-1}}^{(\lambda)} = X_{2i+1,\beta}^{(\lambda+1)}  \quad
		\text{for all~} 1\le \beta\le 2^{m-\lambda-1} , \label{eq:rdX1}                                 \\
		 & X_{i+1,\beta}^{(\lambda)} = X_{2i,\beta}^{(\lambda+1)} + X_{2i+2,\beta}^{(\lambda+1)}, \quad
		X_{i+1,\beta+2^{m-\lambda-1}}^{(\lambda)} = X_{2i+2,\beta}^{(\lambda+1)}  \quad
		\text{for all~} 1\le \beta\le 2^{m-\lambda-1}.
		\label{eq:rdX2}
	\end{align}
	\eqref{eq:rdX1} further implies that
	\begin{align*}
		\hat{x}_{i,\beta}^{(\ell, \lambda)} = \hat{x}_{2i-1,\beta}^{(\ell, \lambda+1)} + \hat{x}_{2i+1,\beta}^{(\ell, \lambda+1)}, \quad
		\hat{x}_{i,\beta+2^{m-\lambda-1}}^{(\ell, \lambda)} = \hat{x}_{2i+1,\beta}^{(\ell, \lambda+1)} \\
		\text{for all~} 1\le \ell\le L_c \text{~and~}
		1\le \beta\le 2^{m-\lambda-1}.
	\end{align*}
	Combining this with \eqref{eq:swp_R1}, we complete the proof of \eqref{eq:swp_1bit} for Case (1).
		{\bf Case (2) $i= n_c-1$:} In this case, we need to prove both \eqref{eq:swp_1bit} and \eqref{eq:swp_2bit}. The proof of \eqref{eq:swp_1bit} is exactly the same as Case (1). To prove \eqref{eq:swp_2bit}, we observe that if $i= n_c-1$, then $2i+1=2n_c-1=2^{\lambda+1}-1$. Then by the induction hypothesis, after executing Lines~12--13, we have not only \eqref{eq:part1} but also
	\begin{align*}
		\tR[\ell, \lambda+1][\beta+2^{m-\lambda-1}] = \hat{x}_{2i+2,\beta}^{(\ell, \lambda+1)} \quad
		\text{for all~} 1\le \ell\le L_c \text{~and~}
		1\le \beta\le 2^{m-\lambda-1}.
	\end{align*}
	Therefore, Lines~33--35 become
	\begin{equation} \label{eq:swp_R2}
		\begin{aligned}
			\tR[\ell, \lambda][\beta+2^{m-\lambda}] = \hat{x}_{2i,\beta}^{(\ell, \lambda+1)} + \hat{x}_{2i+2,\beta}^{(\ell, \lambda+1)}, \quad
			\tR[\ell, \lambda][\beta+2^{m-\lambda-1}+2^{m-\lambda}] = \hat{x}_{2i+2,\beta}^{(\ell, \lambda+1)} \\
			\text{for all~} 1\le \ell\le L_c \text{~and~}
			1\le \beta\le 2^{m-\lambda-1}.
		\end{aligned}
	\end{equation}
	Equation \eqref{eq:rdX2} implies that
	\begin{align*}
		\hat{x}_{i+1,\beta}^{(\ell, \lambda)} = \hat{x}_{2i,\beta}^{(\ell, \lambda+1)} + \hat{x}_{2i+2,\beta}^{(\ell, \lambda+1)}, \quad
		\hat{x}_{i+1,\beta+2^{m-\lambda-1}}^{(\ell, \lambda)} = \hat{x}_{2i+2,\beta}^{(\ell, \lambda+1)} \\
		\text{for all~} 1\le \ell\le L_c \text{~and~}
		1\le \beta\le 2^{m-\lambda-1}.
	\end{align*}
	Combining this with \eqref{eq:swp_R2}, we complete the proof of \eqref{eq:swp_2bit}.
\end{proof}


The next lemma shows that the data structures $\tD$ and $\tB$ are large enough to store the transition probabilities and the decoding results of the intermediate vectors throughout the decoding procedure.
\begin{lemma}\label{lemma:abs_space}
	Throughout the whole decoding procedure, we have  $\numD[\lambda]\le L$ and $\numB[\lambda]\le 6L$ for all $1\le \lambda\le m$.  The space complexity of the SCL decoder is $O(Ln)$.
\end{lemma}
\begin{proof}
	In TABLE~\ref{tb:abs_space}, we use the method in the proof of Lemma~\ref{lemma:stdb_space} to obtain the upper bound of $\numB[\lambda+1]$ at the starting point and the end of the function $\texttt{decode\_channel}(\lambda+1, j)$. The upper bounds in TABLE~\ref{tb:abs_space} immediately imply $\numB[\lambda]\le 6L$. The proof of $\numD[\lambda]\le L$ is the same as Lemma~\ref{lemma:st_space}, and we do not repeat it.

	The proof of the space complexity is the same as Lemma~\ref{lemma:st_space}.
\end{proof}

\begin{table}[H]
	\centering
	\begin{tabular}{|c|r|l|c|c|}
		\hline
		\multicolumn{3}{|c|}{cases}                                                       & start                                     & end                                 \\
		\hline
		\multirow{4}*{$2i\in\cI^{(2^{\lambda+1})}$}                                       & \multirow{2}*{$1\le i\le 2^{\lambda}-1$}  & $j = 2i-1$                & 0  & L  \\
		\cline{3-5}
		~                                                                                 & ~                                         & $j = 2i$                  & L  & 2L \\
		\cline{2-5}
		~                                                                                 & $1\le i< 2^{\lambda}-1$                   & \multirow{2}*{$j = 2i+1$} & 2L & 3L \\
		\cline{2-2}\cline{4-5}
		~                                                                                 & $i = 2^{\lambda}-1$                       & ~                         & 2L & 4L \\
		\hline
		\multirow{2}*{$2(i-1)\in\cI^{(2^{\lambda+1})}$}                                   & $1 \le i\le 2^{\lambda}-1$                & $j = 2i$                  & 3L & 4L \\
		\cline{2-5}
		~                                                                                 & $i = 2^{\lambda}-1$                       & $j = 2i+1$                & 4L & 6L \\
		\hline
		\multirow{2}*{$2(i-1)\notin\cI^{(2^{\lambda+1})}, 2i\notin\cI^{(2^{\lambda+1})}$} & \multirow{2}*{$1 \le i\le 2^{\lambda}-1$} & $j = 2i-1$                & 0  & L  \\
		\cline{3-5}
		~                                                                                 & ~                                         & $j = 2i$                  & L  & 2L \\
		\cline{2-5}
		~                                                                                 & $i=2^{\lambda}-1$                         & $j = 2i+1$                & 2L & 4L \\
		\hline
	\end{tabular}
	\caption{The upper bound of $\numB[\lambda+1]$ at the starting point and the end of the function $\texttt{decode\_channel}(\lambda+1, j)$}
	\label{tb:abs_space}
\end{table}

\begin{proposition}
	The decoding time complexity of ABS polar codes is $O(Ln\log(n))$.
\end{proposition}

\begin{algorithm}[ht]
	\DontPrintSemicolon
	\caption{\texttt{decode\_original\_channel$(\lambda, i)$}}
	\label{algo:decode_ori_channel}
	\KwIn{$\lambda \in \{1,2, \dots, m\}$ and index $i$ satisfying $1\le i \le 2^{\lambda}-1$ and $2i\notin \cI^{(2^{\lambda + 1})}$}

	$n_c\gets 2^{\lambda}$

	\If(\Comment{$V_{2i-1}^{(2n_c)} = (V_{i}^{(n_c)})^\oria$}){\em $2(i-1)\notin \cI^{(2^{\lambda + 1})}$}{

		\texttt{calculate\_$\oria$\_transform}$(\lambda+1)$

		\texttt{decode\_channel}$(\lambda+1, 2i-1)$

		\For{$\ell\in\{1, 2, \dots, L_c\}$}{
			$\tH[\ell, \lambda + 1]\gets \tR[\ell, \lambda+1]$
		}
	}

	\Comment{$V_{2i}^{(2n_c)} = (V_{i}^{(n_c)})^\orib$}

	\texttt{calculate\_$\orib$\_transform}$(\lambda+1)$

	\texttt{decode\_channel}$(\lambda+1, 2i)$

	\For{$\ell\in\{1, 2, \dots, L_c\}$}{
		$\tR[\ell, \lambda]\gets \tR[\ell, \lambda+1]$
	}
	\vspace*{.1in}

	\eIf{$i \le n_c-2$}{
		\Comment{Only decode one bit ${X}_{i,\beta}^{(\lambda)}$ for each $\beta$}

		\For{$\ell\in\{1, 2, \dots, L_c\}$}{
			$\temppointer \gets \tR[\ell, \lambda]$

			$\tR[\ell, \lambda] \gets \texttt{allocate\_bit}(\lambda, 1)$

			\For{\em $\beta \in \{1,2,\dots,n/(2n_c)\}$}{
				$\beta' \gets \beta + n/(2n_c)$

				$\tR[\ell, \lambda][\beta] \gets \tH[\ell, \lambda+1][\beta] + \temppointer[\beta]$

				$\tR[\ell, \lambda][\beta'] \gets \temppointer[\beta]$

			}
		}
	}{
		\Comment{Decode two bits ${X}_{n_c-1,\beta}^{(\lambda)}, {X}_{n_c,\beta}^{(\lambda)}$ for each $\beta$}

		\Comment{$V_{2i+1}^{(2n_c)} = (V_{i}^{(n_c)})^\oric$}

		\texttt{calculate\_$\oric$\_transform}$(\lambda+1)$

		\texttt{decode\_channel}$(\lambda+1, 2i+1)$

		\For{$\ell\in\{1, 2, \dots, L_c\}$}{
			$\temppointer \gets \tR[\ell, \lambda]$

			$\tR[\ell, \lambda] \gets \texttt{allocate\_bit}(\lambda, 2)$

			\For{\em $\beta \in \{1,2,\dots,n/(2n_c)\}$}{
				$\beta' \gets \beta + n/(2n_c)$

				$\tR[\ell, \lambda][\beta] \gets \tH[\ell, \lambda+1][\beta] + \temppointer[\beta]$

				$\tR[\ell, \lambda][\beta'] \gets \temppointer[\beta]$

				$\tR[\ell, \lambda][\beta + n/(n_c)] \gets$\\
				$\qquad\tR[\ell, \lambda+1][\beta] + \tR[\ell, \lambda+1][\beta']$

				$\tR[\ell, \lambda][\beta' + n/(n_c)] \gets \tR[\ell, \lambda+1][\beta']$

			}
		}
	}

	$\numB[\lambda+1]\gets 0$

	\Return

\end{algorithm}

\begin{algorithm}[ht]
	\DontPrintSemicolon
	\caption{\texttt{decode\_swapped\_channel$(\lambda, i)$}}
	\label{algo:decode_swp_channel}
	\KwIn{$\lambda \in \{1,2, \dots, m\}$ and index $i$ satisfying $1\le i \le 2^{m-\lambda}-1$ and $2i\in \cI^{(2^{\lambda + 1})}$.}

	$n_c\gets 2^{\lambda}$
	\Comment{$V_{2i-1}^{(2n_c)} = (V_{i}^{(n_c)})^\swpa$}

	\texttt{calculate\_$\swpa$\_transform}$(\lambda+1)$

	\texttt{decode\_channel}$(\lambda+1, 2i-1)$

	\For{$\ell\in\{1, 2, \dots, L_c\}$}{
		$\tH[\ell, \lambda + 1]\gets \tR[\ell, \lambda+1]$
	}
	\vspace*{.1in}

	\Comment{$V_{2i}^{(2n_c)} = (V_{i}^{(n_c)})^\swpb$}

	\texttt{calculate\_$\swpb$\_transform}$(\lambda+1)$

	\texttt{decode\_channel}$(\lambda+1, 2i)$

	\For{$\ell\in\{1, 2, \dots, L_c\}$}{
		$\tR[\ell, \lambda]\gets \tR[\ell, \lambda+1]$
	}
	\vspace*{.1in}

	\Comment{$V_{2i+1}^{(2n_c)} = (V_{i}^{(n_c)})^\swpc$}

	\texttt{calculate\_$\swpc$\_transform}$(\lambda+1)$

	\texttt{decode\_channel}$(\lambda+1, 2i+1)$

	\eIf{$i\le n_c-2$}{
		\Comment{Only decode one bit ${X}_{i,\beta}^{(\lambda)}$ for each $\beta$}

		\For{$\ell\in\{1, 2, \dots, L_c\}$}{
			$\temppointer \gets \tR[\ell, \lambda]$

			$\tR[\ell, \lambda] \gets \texttt{allocate\_bit}(\lambda, 1)$

			\For{\em $\beta \in \{1,2,\dots,n/(2n_c)\}$}{
				$\beta' \gets \beta + n/(2n_c)$

				$\tR[\ell, \lambda][\beta] \gets \tH[\ell, \lambda+1][\beta] + \tR[\ell, \lambda+1][\beta]$

				$\tR[\ell, \lambda][\beta'] \gets \tR[\ell, \lambda+1][\beta]$
			}

			$\tH[\ell, \lambda+1] \gets \temppointer$
		}

	}{
		\Comment{Decode two bits ${X}_{n_c-1, \beta}^{(\lambda)}, {X}_{n_c,\beta}^{(\lambda)}$ for each $\beta$}

		\For{$\ell\in\{1, 2, \dots, L_c\}$}{
			$\temppointer \gets \tR[\ell, \lambda]$

			$\tR[\ell, \lambda] \gets \texttt{allocate\_bit}(\lambda, 2)$

			\For{\em $\beta \in \{1,2,\dots,n/(2n_c)\}$}{
				$\beta' \gets \beta + n/(2n_c)$

				$\tR[\ell, \lambda][\beta] \gets \tH[\ell, \lambda+1][\beta] + \tR[\ell,\lambda+1][\beta]$

				$\tR[\ell, \lambda][\beta'] \gets \tR[\ell,\lambda+1][\beta]$

				$\tR[\ell, \lambda][\beta + n/(n_c)] \gets$\\
				$\qquad\temppointer[\beta] + \tR[\ell, \lambda+1][\beta']$

				$\tR[\ell, \lambda][\beta' + n/(n_c)] \gets \tR[\ell, \lambda+1][\beta']$

			}
		}

		$\numB[\lambda+1]\gets 0$
	}

	\Return
\end{algorithm}

\begin{algorithm}[ht]
	\DontPrintSemicolon
	\caption{\texttt{calculate\_$\swpa$\_transform}$(\lambda)$}
	\label{algo:calcu_swpa_proba}
	\KwIn{layer $2\le \lambda \le m$}
	\KwOut{Update the entries pointed by $\tP[\ell, \lambda]$, $1\le \ell\le L_c $ }

	$\bar{n}_c \gets 2^{m-\lambda}$

	\For{$\ell\in\{1, 2, \dots, L_c\}$}{

		$\tP[\ell,\lambda]\gets \texttt{allocate\_prob}(\lambda)$

		\For{\em $\beta \in \{1,2,\dots,\bar{n}_c\}, r_1,r_2\in\{0,1\}$}{
			$\beta' \gets \beta + \bar{n}_c$

			$\tP[\ell, \lambda][\beta,r_1, r_2] \gets$ $\frac14\sum_{r_3, r_4\in\{0,1\}}$
			$\tP[\ell, \lambda-1][\beta, r_1+r_3, r_2+r_4]\tP[\ell, \lambda-1][\beta', r_3, r_4]$
		}
	}

	\Return
\end{algorithm}

\begin{algorithm}[ht]
	\DontPrintSemicolon
	\caption{\texttt{calculate\_$\swpb$\_transform}$(\lambda)$}
	\label{algo:calcu_swpb_proba}
	\KwIn{layer $2\le \lambda \le m$}
	\KwOut{Update the entries pointed by $\tP[\ell, \lambda]$, $1\le \ell\le L_c $ }

	$\bar{n}_c \gets 2^{m-\lambda}$


	\For{$\ell\in\{1, 2, \dots, L_c\}$}{

		$\tP[\ell,\lambda]\gets \texttt{allocate\_prob}(\lambda)$

		\For{\em $\beta \in \{1,2,\dots,\bar{n}_c\}, r_2,r_3\in\{0,1\}$}{
			$r_1\gets \tH[\ell, \lambda][\beta]$

			$\beta' \gets \beta + \bar{n}_c$

			$\tP[\ell, \lambda][\beta,r_2,r_3] \gets$ $\frac14\sum_{r_4\in\{0,1\}}$
			$\tP[\ell, \lambda-1][\beta, r_1+r_3, r_2+r_4]\tP[\ell, \lambda-1][\beta', r_3, r_4]$
		}
	}

	\Return

\end{algorithm}

\begin{algorithm}[ht]
	\DontPrintSemicolon
	\caption{\texttt{calculate\_$\swpc$\_transform}$(\lambda)$}
	\label{algo:calcu_swpc_proba}
	\KwIn{layer $2\le \lambda \le m$}
	\KwOut{Update the entries pointed by $\tP[\ell, \lambda]$, $1\le \ell\le L_c $ }

	$\bar{n}_c \gets 2^{m-\lambda}$


	\For{$\ell\in\{1, 2, \dots, L_c\}$}{

		$\tP[\ell,\lambda]\gets \texttt{allocate\_prob}(\lambda)$

		\For{\em $\beta \in \{1,2,\dots,\bar{n}_c\}, r_3,r_4\in\{0,1\}$}{
			$r_1\gets \tH[\ell, \lambda][\beta],\quad r_2\gets \tR[\ell, \lambda-1][\beta]$

			$\beta' \gets \beta + \bar{n}_c$

			$\tP[\ell, \lambda][\beta,r_3,r_4] \gets$\\
			$\frac14\tP[\ell, \lambda-1][\beta, r_1+r_3, r_2+r_4]\tP[\ell, \lambda-1][\beta', r_3, r_4]$
		}
	}

	\Return
\end{algorithm}

\clearpage

\section{Simulation results} \label{sect:simu}

\subsection{Scaling exponent over binary erasure channels}  \label{sect:scaling_BEC}

In this subsection, we empirically calculate the scaling exponents of ABS polar codes and standard polar codes over a BEC with erasure probability $0.5$.

When the original channel $W$ is a general BMS channel, we can only obtain an approximation of the transition probabilities of the adjacent-bits-channels through quantization, as discussed in Section~\ref{sect:quantization}. However, when the original channel $W$ is a BEC, we are able to calculate the exact parameters of the adjacent-bits-channels. To that end, we introduce a class of channels called double-bits-erasure-channels (DBEC). The input alphabet of a DBEC is $\{0,1\}^2$, and the output alphabet is $\{0,1,?\}^3$. For a given input $(u_1,u_2)\in\{0,1\}^2$, the output of the DBEC can only take the following five values
\begin{itemize}
	\item $(u_1,u_1+u_2,u_2)$ with probability $p$,
	\item $(u_1,?,?)$ with probability $q$,
	\item $(?,u_1+u_2,?)$ with probability $r$,
	\item $(?,?,u_2)$ with probability $s$,
	\item $(?,?,?)$ with probability $t$.
\end{itemize}
$p$ is the probability of preserving all information in the inputs. $q,r,s$ are the probabilities of preserving one bit of information in the inputs. $t$ is the probability of erasing all the information. Such a DBEC is denoted as DBEC$(p,q,r,s,t)$, where the parameters satisfy $p+q+r+s+t=1$. Note that DBEC has been studied in the literature under other names. For example, the authors of \cite{Duursma22} call it tetrahedral erasure channel.

One can show that if the original channel $W$ is a BEC, then all the adjacent-bits-channels in the ABS polar code construction are DBEC. More precisely, using \eqref{eq:v_abs_init}, we can show that if $W$ is a BEC with erasure probability $\epsilon$, then
$$
	V_{1}^{(2),\ABS} = \text{DBEC}( (1-\epsilon)^2,0,(1-\epsilon)\epsilon,(1-\epsilon)\epsilon,\epsilon^2 ) .
	$$
	Moreover, if an adjacent-bits-channel $V=\text{DBEC}(p,q,r,s,t)$, then
	\begin{align*}
		V^{\oria}= & \text{DBEC}((p+q)^2,            0,           (p+q)(r+s+t),   (r+s+t)(p+q), (r+s+t)^2) ,          \\
		V^{\orib}= & \text{DBEC}(p^2+2rp+2sp,        2q-q^2+2pt,  2rs,            r^2+s^2,      2t(r+s)+t^2) ,        \\
		V^{\oric}= & \text{DBEC}((p+r+s)^2,          0,           (p+r+s)(q+t),   (q+t)(p+r+s), (q+t)^2) ,            \\
		V^{\swpa}= & \text{DBEC}(p^2,                q^2+2pq,     r^2+2pr,        s^2+2ps,      2t-t^2+2qr+2qs+2rs) , \\
		V^{\swpb}= & \text{DBEC}(p^2+2pr+2ps,        r^2+s^2,     2rs,            2q-q^2+2pt,   t^2+2rt+2rs) ,        \\
		V^{\swpc}= & \text{DBEC}(2p-p^2+2qr+2qs+2rs, q^2+2tq,     r^2+2rt,        s^2+2st,      t^2) .
	\end{align*}
	Combining this with Lemma~\ref{lemma:recur_ABS}, we can explicitly calculate the parameters of all the adjacent-bits-channels in the ABS polar code construction when the original channel $W$ is a BEC. After that, we use \eqref{eq:v_to_w} to calculate the erasure probabilities of each bit-channel: If $V_{i}^{(n),\ABS}=\text{DBEC}(p,q,r,s,t)$, then $W_{i}^{(n),\ABS}$ is an erasure channel with erasure probability $r+s+t$, and $W_{i+1}^{(n),\ABS}$ is an erasure channel with erasure probability $q+t$.

	Let $W$ be a BEC with erasure probability $0.5$. For $n\in\{2^6, 2^7, \dots, 2^{20}\}$, we define
	\begin{align*}
		 & f_{\polar}(n):=\frac{1}{n} |\{i:1\le i\le n,~~ 0.01\le I(W_i^{(n)})\le 0.99 \}| ,    \\
		 & f_{\ABS}(n):=\frac{1}{n} |\{i:1\le i\le n,~~ 0.01\le I(W_i^{(n),\ABS})\le 0.99 \}| .
	\end{align*}
	By definition, $f_{\polar}(n)$ is the fraction of ``unpolarized" bit-channels in the length-$n$ standard polar code constructed for the BEC $W$, and $f_{\ABS}(n)$ is the fraction of ``unpolarized" bit-channels in the length-$n$ ABS polar code constructed for the BEC $W$. A bit-channel is said to be unpolarized if its capacity is between $0.01$ and $0.99$. The values of $f_{\polar}(n)$ and $f_{\ABS}(n)$ for $n\in\{2^6, 2^7, \dots, 2^{20}\}$ are listed in TABLE~\ref{fract}.

	\begin{table}
		\centering
		\begin{tabular}{r|cc}
			\hline
			$n$     & $f_{\polar}(n)$ & $f_{\ABS}(n)$ \\
			\hline
			64      & 0.53125000      & 0.50000000    \\
			128     & 0.43750000      & 0.42187500    \\
			256     & 0.37500000      & 0.34375000    \\
			512     & 0.30078125      & 0.27343750    \\
			1024    & 0.25390625      & 0.22070312    \\
			2048    & 0.20605469      & 0.18164062    \\
			4096    & 0.17041016      & 0.15136719    \\
			8192    & 0.14208984      & 0.12329102    \\
			16384   & 0.11755371      & 0.09936523    \\
			32768   & 0.09674072      & 0.08087158    \\
			65536   & 0.07995605      & 0.06542969    \\
			131072  & 0.06613159      & 0.05333710    \\
			262144  & 0.05499268      & 0.04324722    \\
			524288  & 0.04529572      & 0.03502846    \\
			1048576 & 0.03742218      & 0.02853012    \\
			\hline
		\end{tabular}
		\caption{The fractions of ``unpolarized" bit-channels in standard polar codes and ABS polar codes constructed for a BEC with erasure probability $\epsilon = 0.5$.}
		\label{fract}
	\end{table}

	In order to estimate the scaling exponents, we approximate $f_{\polar}(n)$ as $f_{\polar}(n)\approx c_1 n^{-\gamma_1}$, and we approximate $f_{\ABS}(n)$ as $f_{\ABS}(n)\approx c_2 n^{-\gamma_2}$. By taking the logarithm on both sides of the equation and running linear regression, we obtain that $f_{\polar}(n) \approx 1.67 n^{-0.274}$ and $f_{\ABS}(n) \approx 1.76 n^{-0.297}$. Therefore, the scaling exponent for standard polar codes is $\mu_{\polar}\approx 1/0.274=3.65$, and the scaling exponent for ABS polar codes is $\mu_{\ABS}\approx 1/0.297=3.37$.

	The above empirical calculations of scaling exponents confirm that the polarization of ABS polar codes is indeed faster than standard polar codes. An interesting problem for future research is to obtain provable and tight upper bounds on the scaling exponent of ABS polar codes. Another related question is to analyze the code distance of ABS polar codes and compare it with standard polar codes.

	\subsection{Simulation results over binary-input AWGN channels}

	\begin{figure}
		\centering
		\begin{subfigure}{0.41\linewidth}
			\centering
			\includegraphics[width=\linewidth]{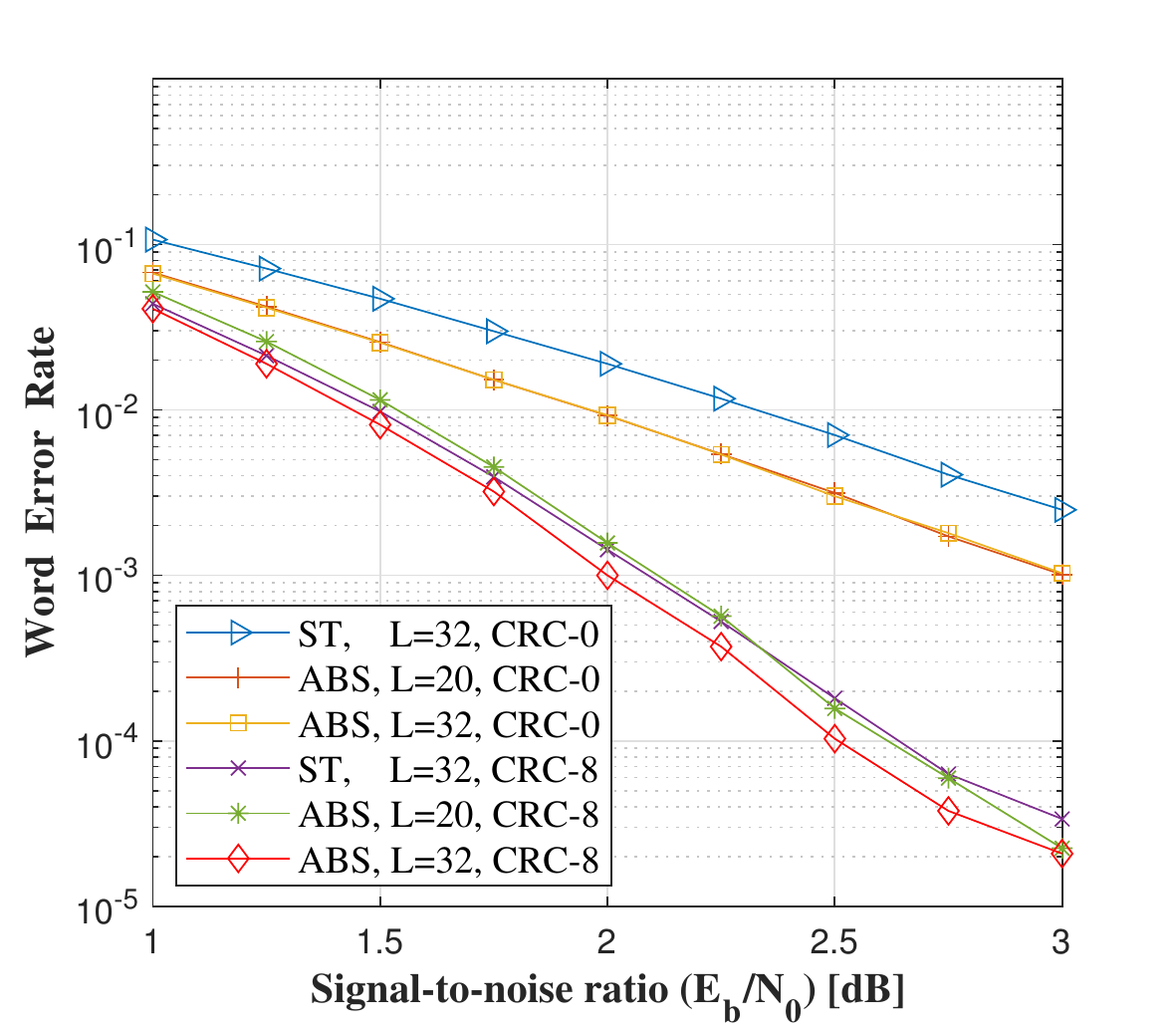}
			\caption{length 256, dimension 77}
		\end{subfigure}
		~\hspace*{0.2in}
		\begin{subfigure}{0.41\linewidth}
			\centering
			\includegraphics[width=\linewidth]{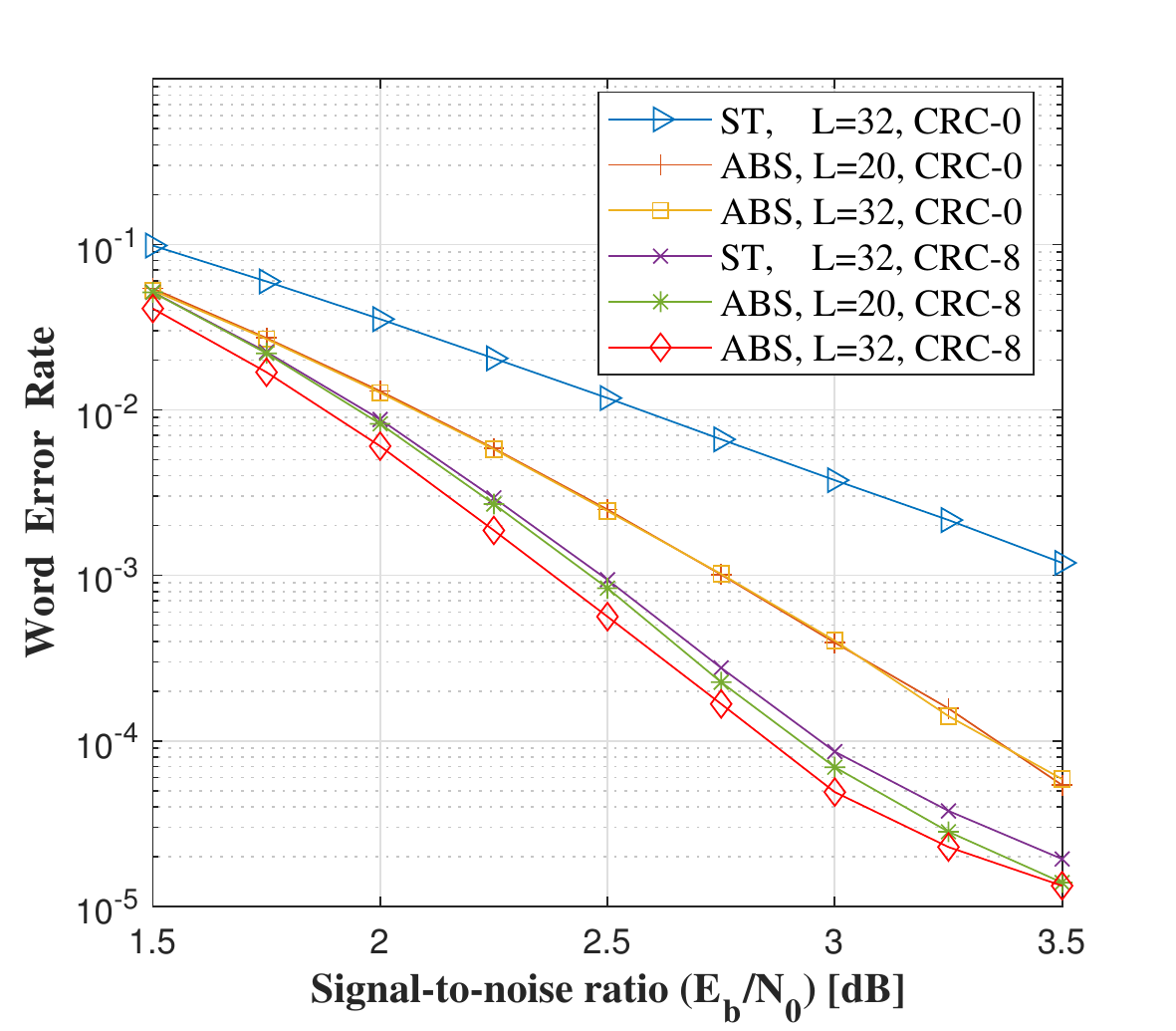}
			\caption{length 256, dimension 128}
		\end{subfigure}

		\vspace*{0.1in}

		\begin{subfigure}{0.41\linewidth}
			\centering
			\includegraphics[width=\linewidth]{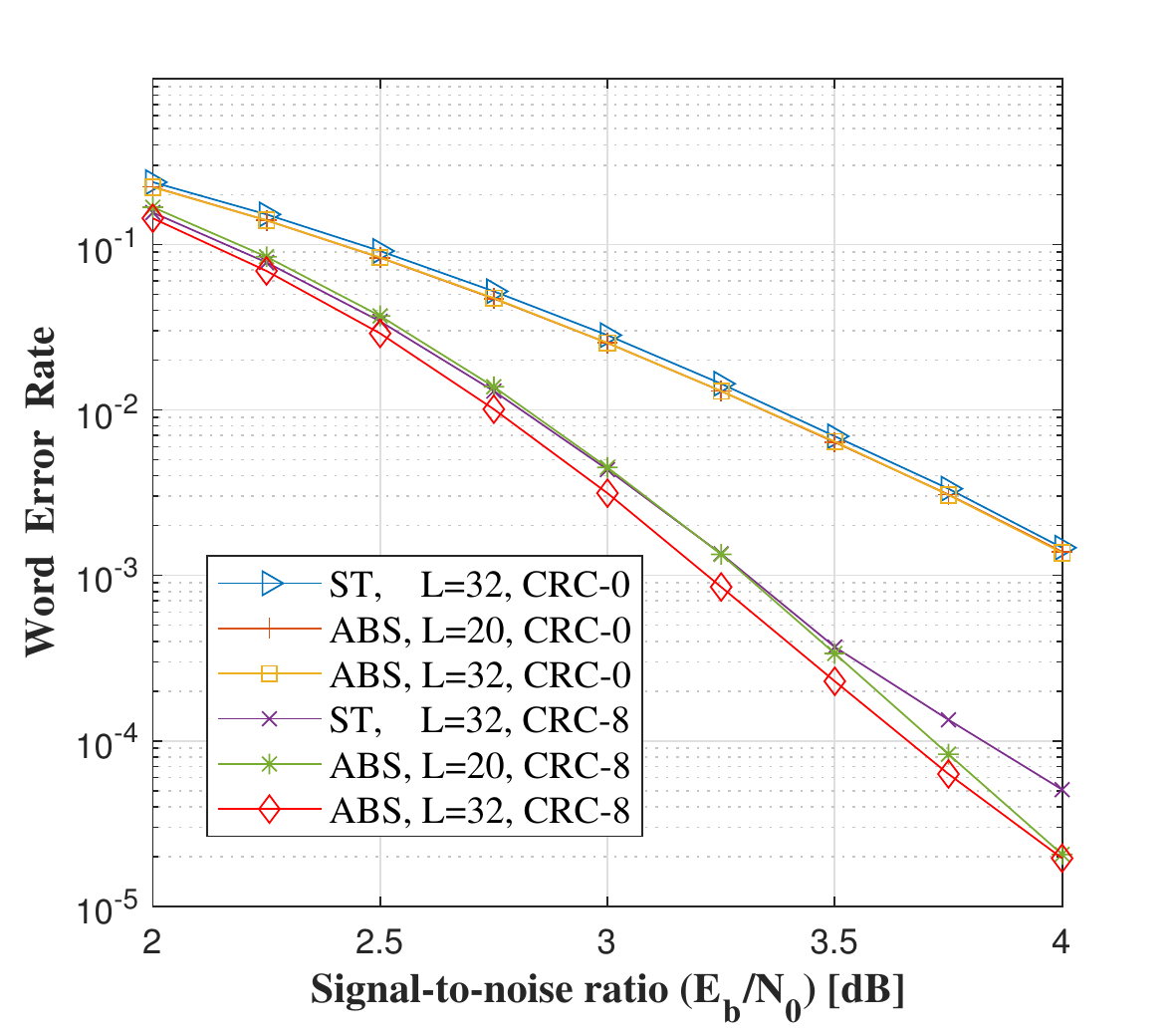}
			\caption{length 256, dimension 179}
		\end{subfigure}
		~\hspace*{0.2in}
		\begin{subfigure}{0.41\linewidth}
			\centering
			\includegraphics[width=\linewidth]{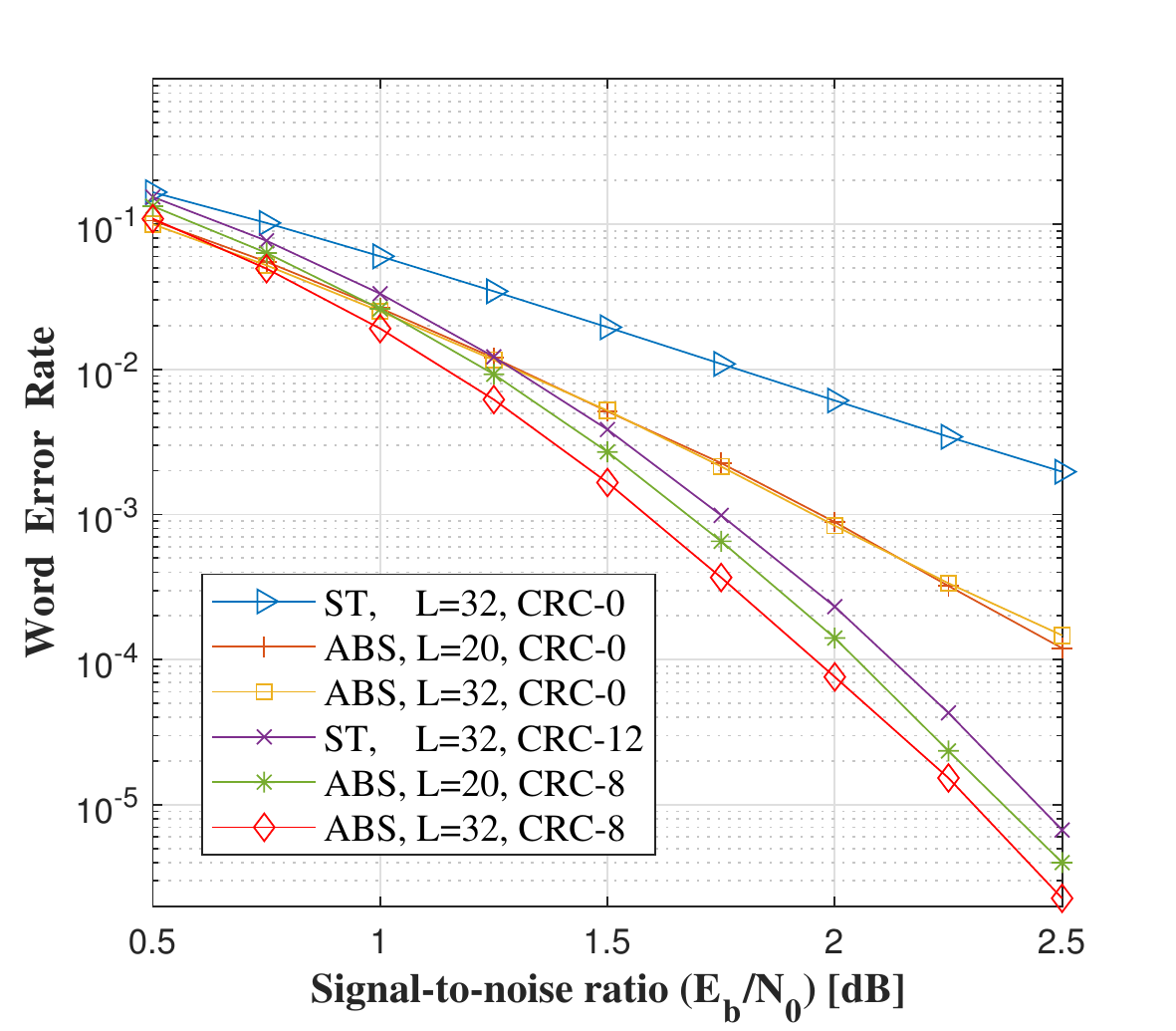}
			\caption{length 512, dimension 154}
		\end{subfigure}

		\vspace*{0.1in}

		\begin{subfigure}{0.41\linewidth}
			\centering
			\includegraphics[width=\linewidth]{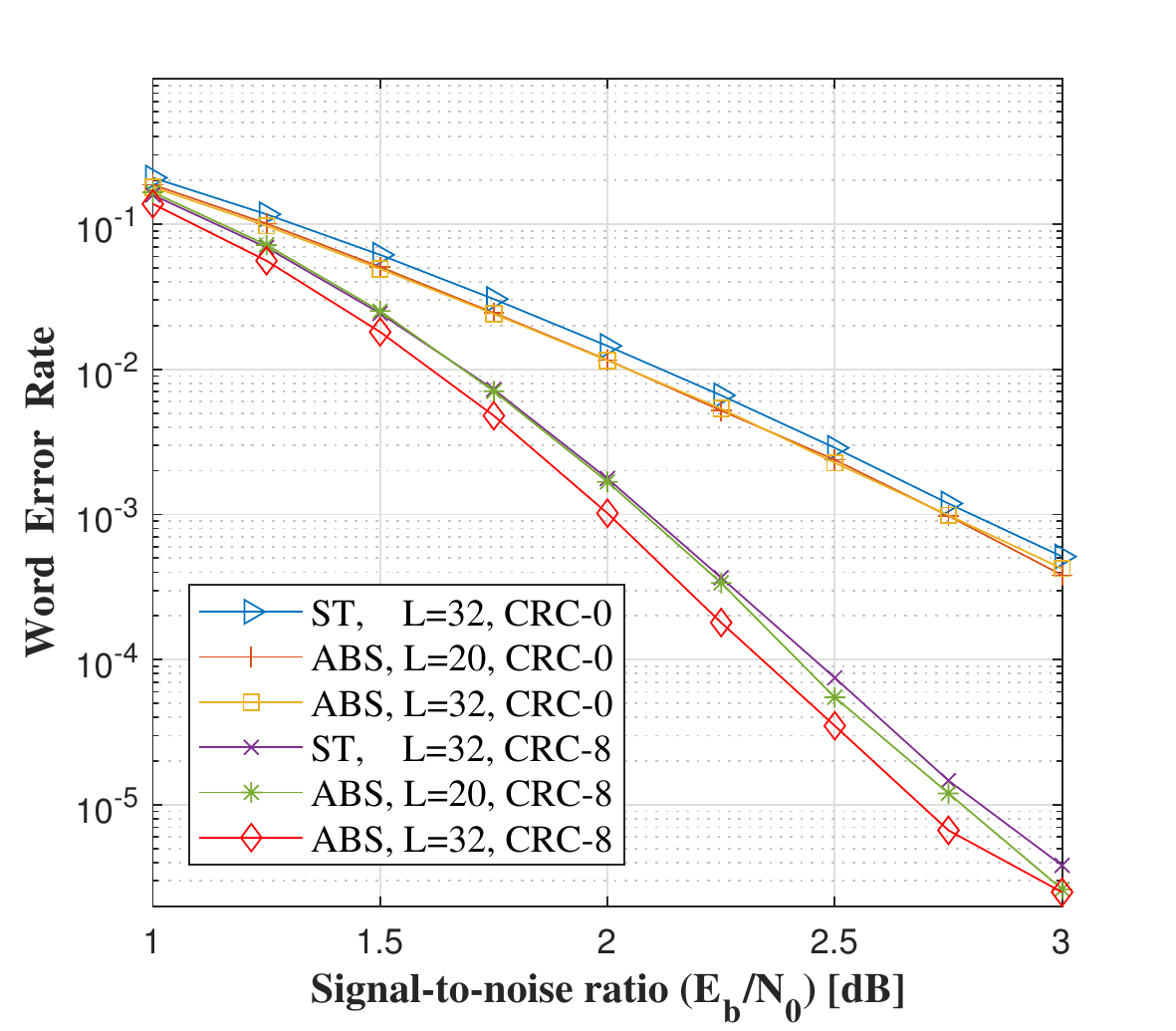}
			\caption{length 512, dimension 256}
		\end{subfigure}
		~\hspace*{0.2in}
		\begin{subfigure}{0.41\linewidth}
			\centering
			\includegraphics[width=\linewidth]{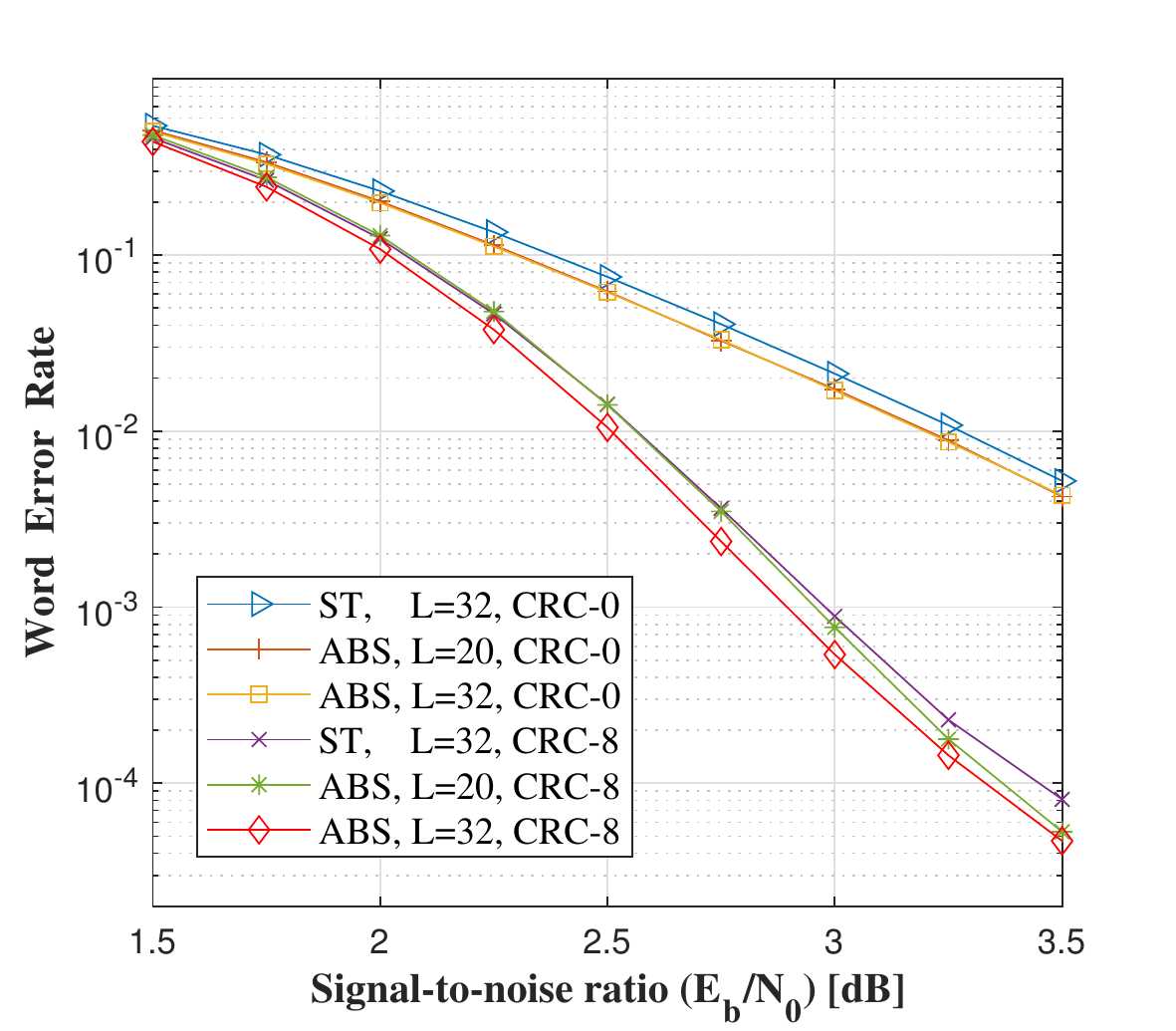}
			\caption{length 512, dimension 358}
		\end{subfigure}

		\caption{Comparison between ABS polar codes and standard polar codes over the binary-input AWGN channel.
			The legend ``ABS" refers to ABS polar codes, and ``ST" refers to standard polar codes.
			``CRC-0" means that we do not use CRC.
			The nonzero CRC length is chosen from the set $\{4,8,12,16,20,24\}$ to minimize the decoding error probability.
			The parameter $L$ is the list size.
			For standard polar codes, we always choose $L=32$.
			For ABS polar codes, we test two different list sizes $L=20$ and $L=32$.}
		\label{fig:cp1}
	\end{figure}

	\begin{figure}
		\centering
		\begin{subfigure}{0.41\linewidth}
			\centering
			\includegraphics[width=\linewidth]{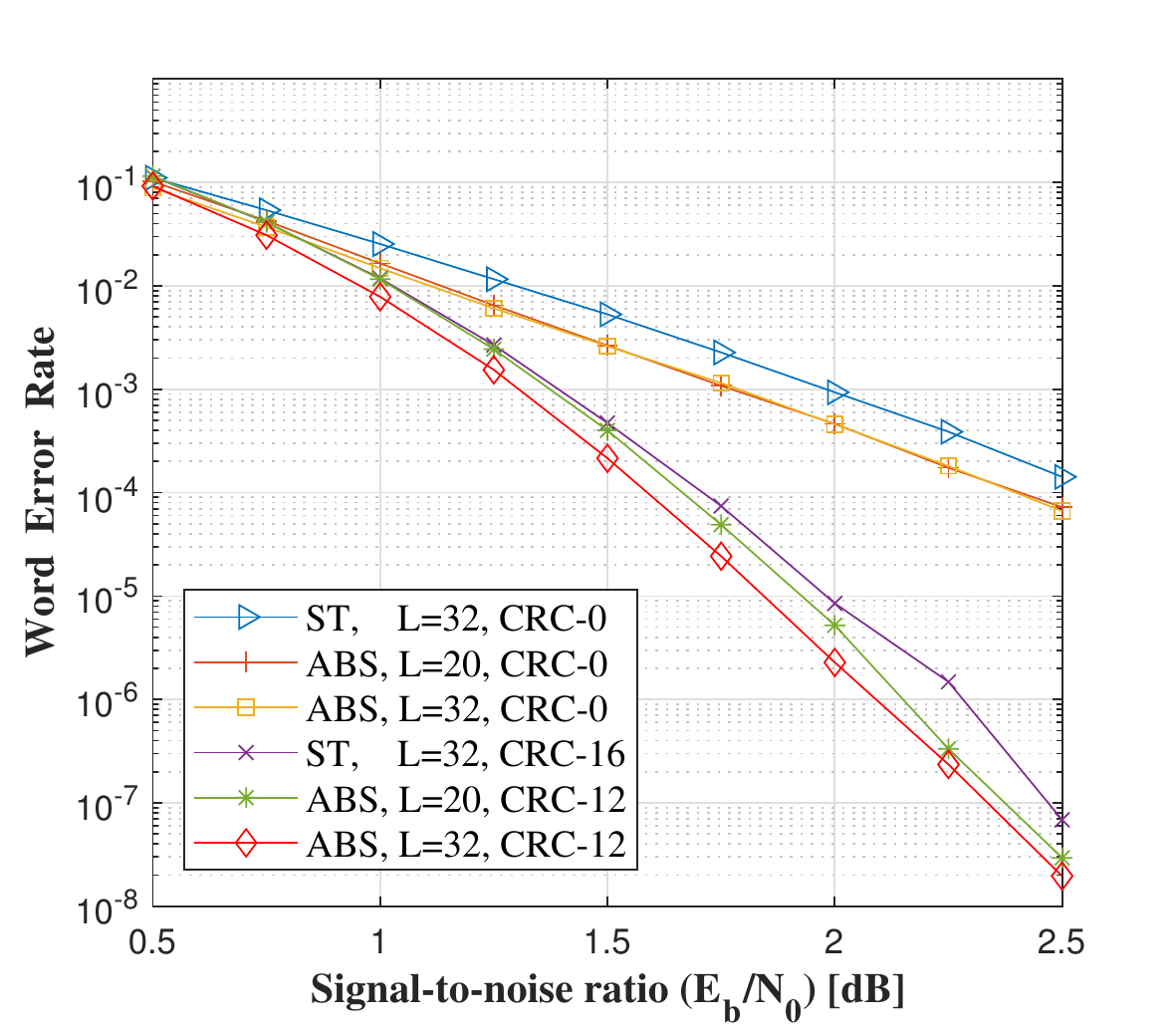}
			\caption{length 1024, dimension 307}
		\end{subfigure}
		~\hspace*{0.2in}
		\begin{subfigure}{0.41\linewidth}
			\centering
			\includegraphics[width=\linewidth]{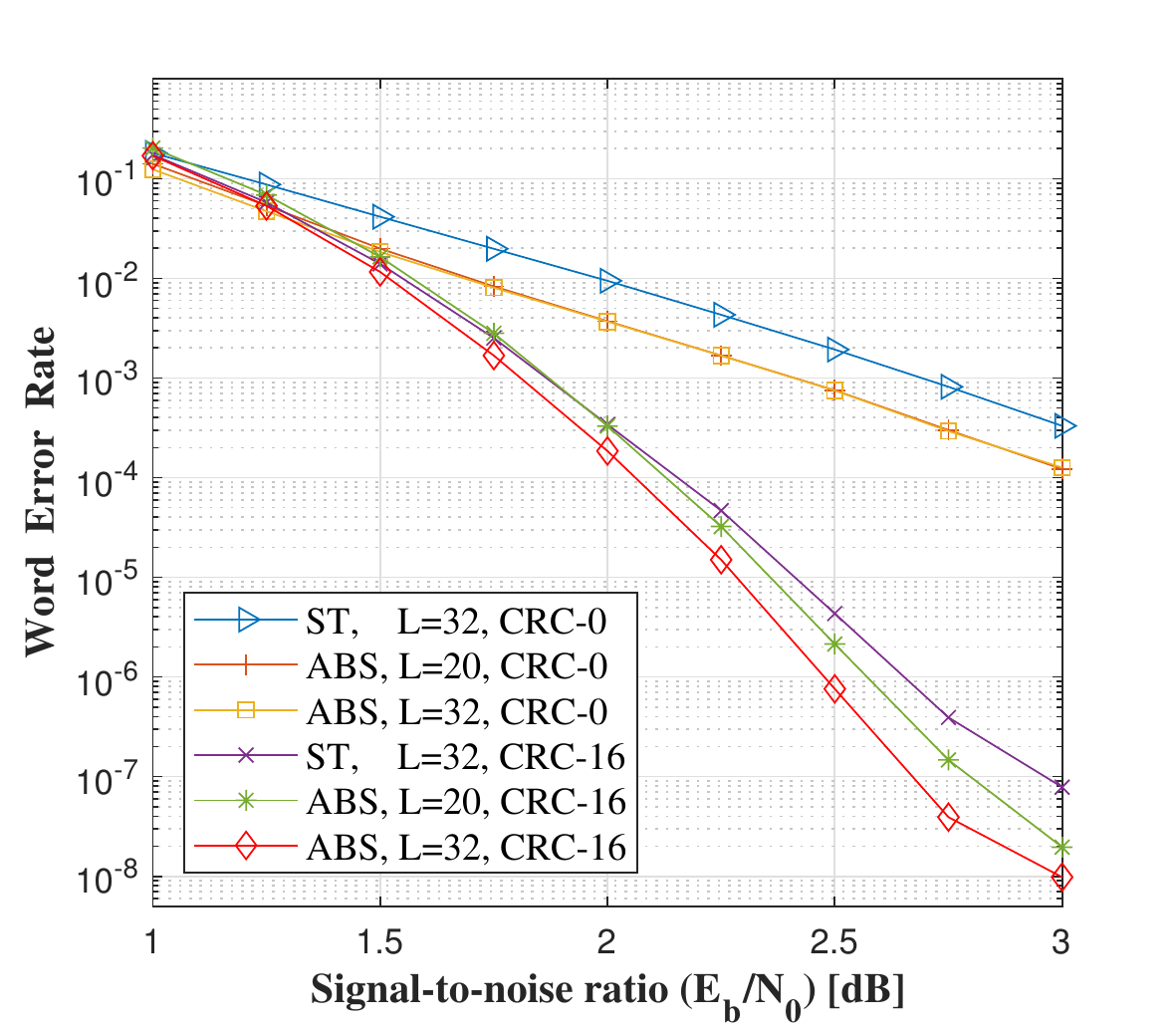}
			\caption{length 1024, dimension 512}
		\end{subfigure}

		\vspace*{0.1in}

		\begin{subfigure}{0.41\linewidth}
			\centering
			\includegraphics[width=\linewidth]{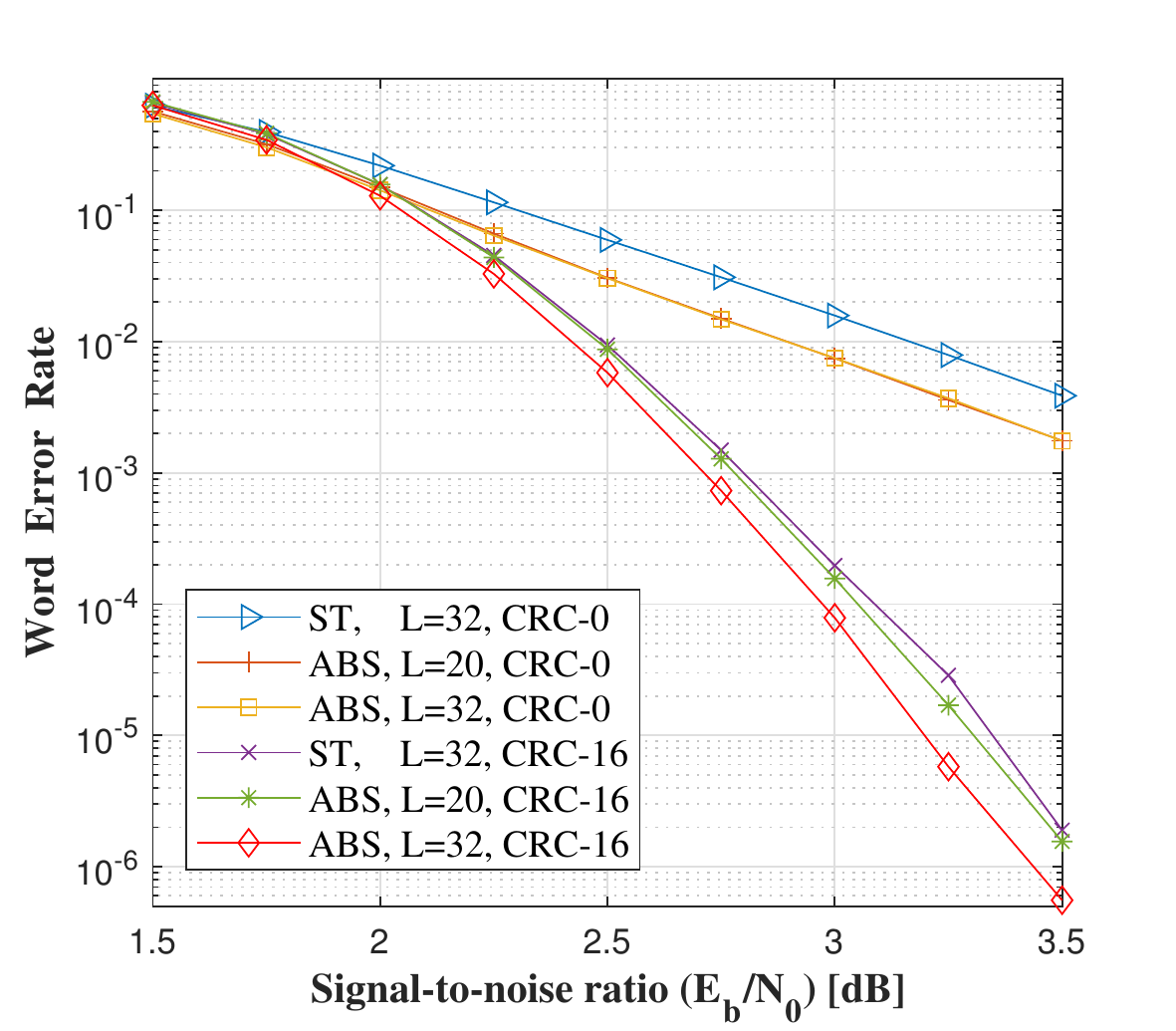}
			\caption{length 1024, dimension 717}
		\end{subfigure}
		~\hspace*{0.2in}
		\begin{subfigure}{0.41\linewidth}
			\centering
			\includegraphics[width=\linewidth]{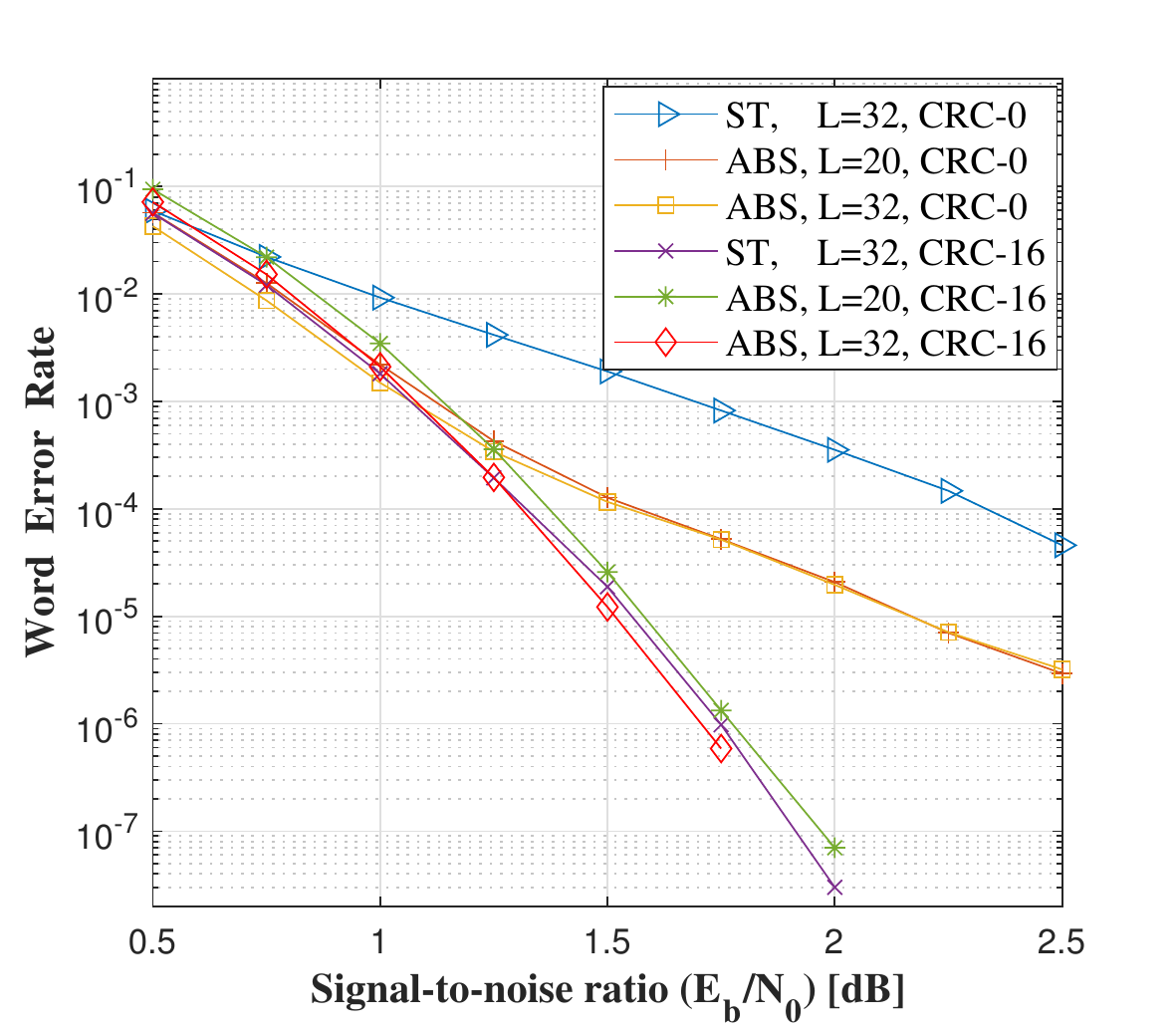}
			\caption{length 2048, dimension 614}
		\end{subfigure}

		\vspace*{0.1in}

		\begin{subfigure}{0.41\linewidth}
			\centering
			\includegraphics[width=\linewidth]{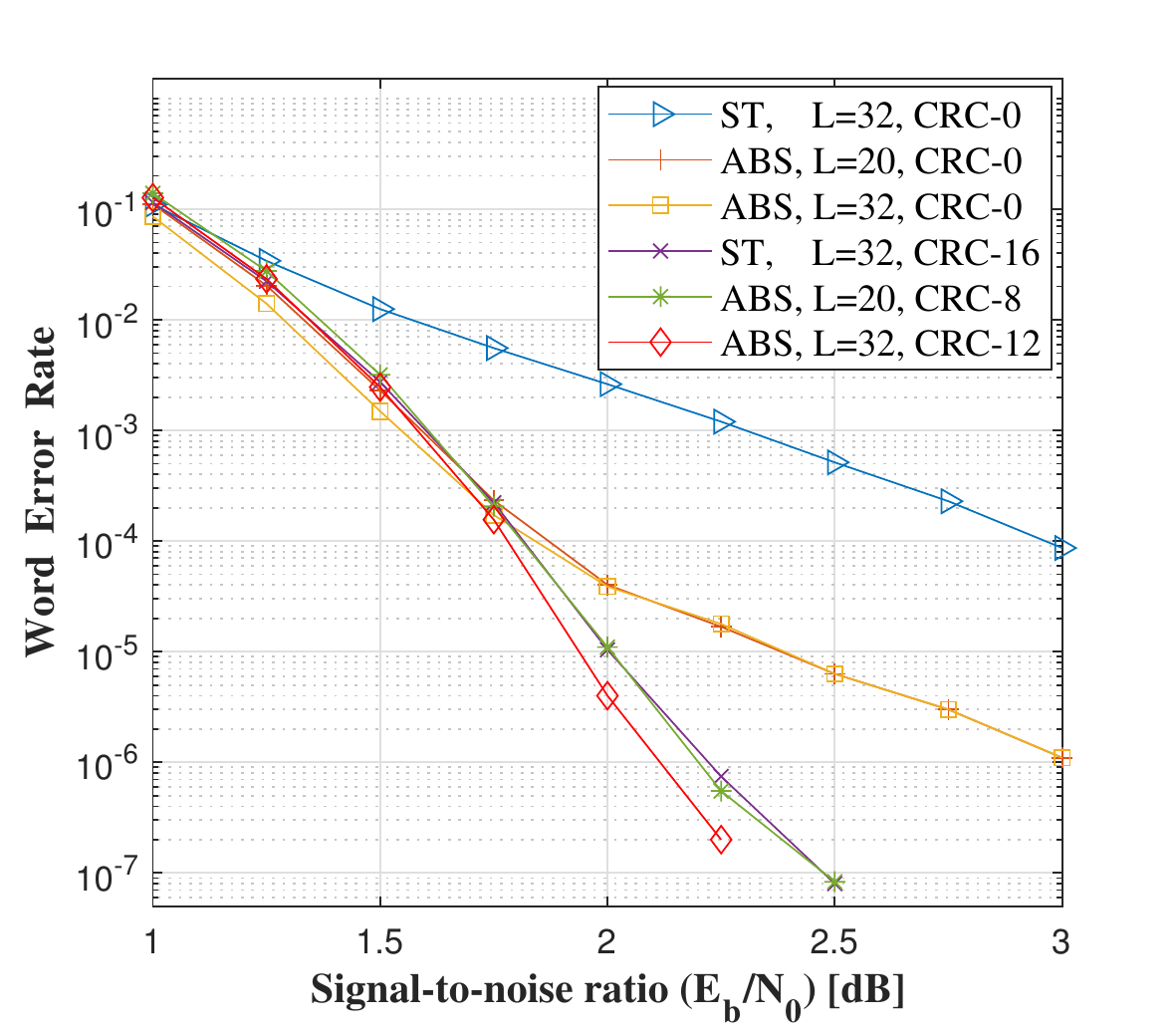}
			\caption{length 2048, dimension 1024}
		\end{subfigure}
		~\hspace*{0.2in}
		\begin{subfigure}{0.41\linewidth}
			\centering
			\includegraphics[width=\linewidth]{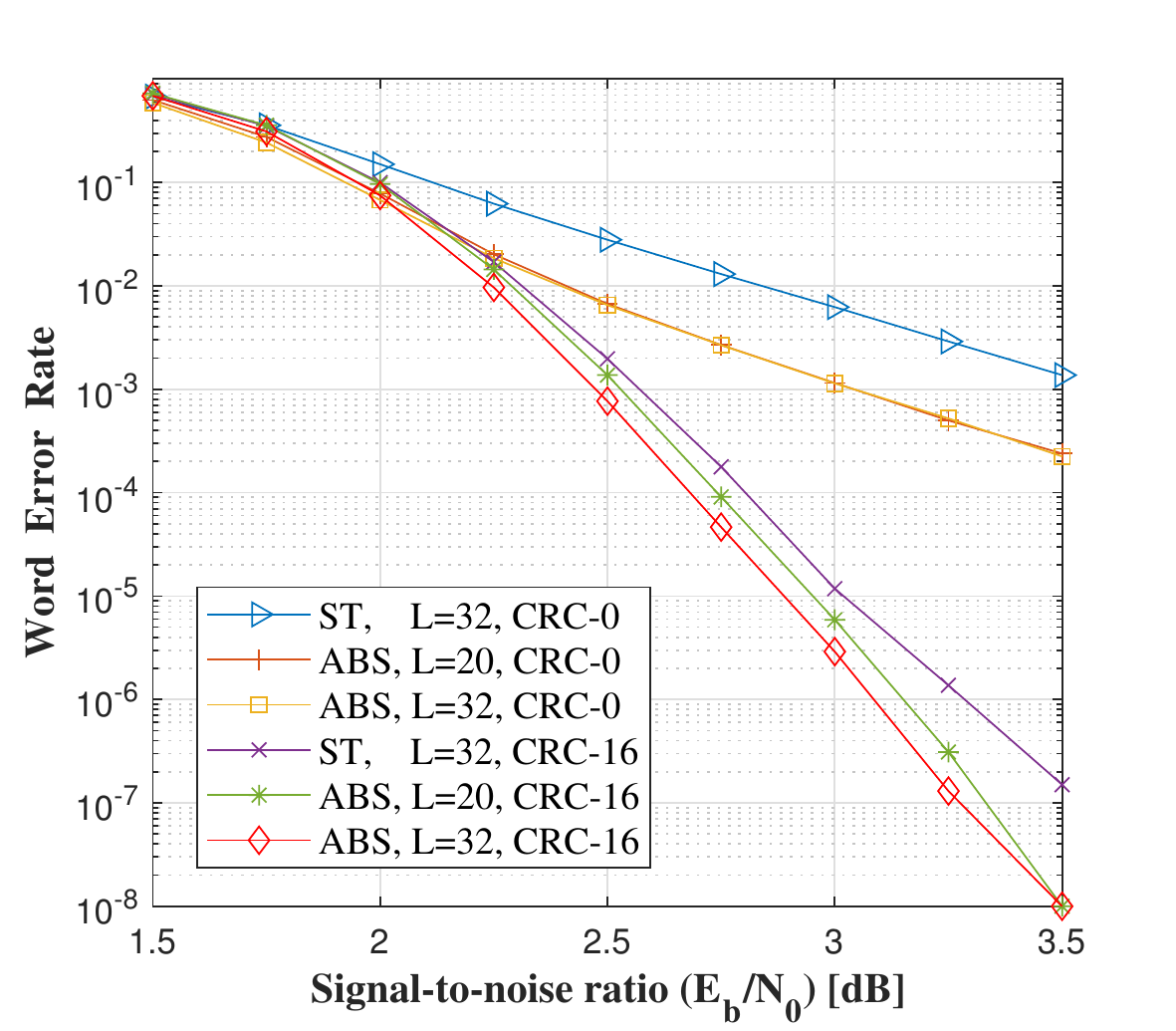}
			\caption{length 2048, dimension 1434}
		\end{subfigure}

		\caption{Comparison between ABS polar codes and standard polar codes over the binary-input AWGN channel.
			The legend ``ABS" refers to ABS polar codes, and ``ST" refers to standard polar codes.
			``CRC-0" means that we do not use CRC.
			The nonzero CRC length is chosen from the set $\{4,8,12,16,20,24\}$ to minimize the decoding error probability.
			The parameter $L$ is the list size.
			For standard polar codes, we always choose $L=32$.
			For ABS polar codes, we test two different list sizes $L=20$ and $L=32$.}
		\label{fig:cp2}
	\end{figure}

	\begin{table}
		\centering
		\begin{tabular}{c|cccccc}
			\hline
			$(n,k)$     & $(256,77)$   & $(256,128)$  & $(256,179)$  & $(512,154)$  & $(512,256)$   & $(512,358)$   \\
			\hline
			ST, $L=32$  & $0.963$ms    & 1.41ms       & 1.73ms       & 1.94ms       & 2.80ms        & 3.54ms        \\

			ABS, $L=20$ & 0.816ms      & 1.24ms       & 1.47ms       & 1.86ms       & 2.66ms        & 3.10ms        \\

			ABS, $L=32$ & 1.29ms       & 1.99ms       & 2.37ms       & 2.93ms       & 4.36ms        & 5.13ms        \\

			\hline
			$(n,k)$     & $(1024,307)$ & $(1024,512)$ & $(1024,717)$ & $(2048,614)$ & $(2048,1024)$ & $(2048,1434)$ \\
			\hline
			ST, $L=32$  & 4.21ms       & 5.75ms       & 7.15ms       & 9.05ms       & 11.7ms        & 14.6ms        \\

			ABS, $L=20$ & 4.32ms       & 5.90ms       & 6.67ms       & 10.6ms       & 12.6ms        & 14.0ms        \\

			ABS, $L=32$ & 6.63ms       & 9.41ms       & 10.8ms       & 16.7ms       & 20.1ms        & 23.2ms        \\
			\hline
		\end{tabular}
		\caption{Comparison of the decoding time over the binary-input AWGN channel with $E_b/N_0=2\dB$.
			The row starting with $(n,k)$ lists the code length and code dimension we have tested.
			The row starting with ``ST, $L=32$" lists the decoding time of the CRC-aided SCL decoder for standard polar codes with list size $32$.
			The row starting with ``ABS, $L=20$" lists the decoding time of the CRC-aided SCL decoder for ABS polar codes with list size $20$.
			The row starting with ``ABS, $L=32$" lists the decoding time of the CRC-aided SCL decoder for ABS polar codes with list size $32$.
			The time unit ``ms" is $10^{-3}$s.}
		\label{tb:time}
	\end{table}

	We conduct extensive simulations to compare the performance of the ABS polar codes and the standard polar codes over the binary-input AWGN channel with various choices of parameters.
	We have tested the performance for $4$ different choices of code length $256,512,1024,2048$.
	For each choice of code length, we test $3$ different code rates $0.3, 0.5$ and $0.7$.
	The comparison of decoding error probability is given in Fig.~\ref{fig:cp1} and Fig.~\ref{fig:cp2}.
	Specifically, Fig.~\ref{fig:cp1} contains the plots for code length $256$ and $512$; Fig.~\ref{fig:cp2} contains the plots for code length $1024$ and $2048$.
	The comparison of decoding time is given in Table~\ref{tb:time}.

	In Fig.~\ref{fig:cp1} and Fig.~\ref{fig:cp2}, for each choice of code length and code dimension, we compare the decoding error probability of the following $6$ decoders: (1) SCL decoder for standard polar codes with list size $32$ and no CRC; (2) SCL decoder for ABS polar codes with list size $20$ and no CRC; (3) SCL decoder for ABS polar codes with list size $32$ and no CRC; (4) SCL decoder for standard polar codes with list size $32$ and optimal CRC length; (5) SCL decoder for ABS polar codes with list size $20$ and optimal CRC length; (6) SCL decoder for ABS polar codes with list size $32$ and optimal CRC length.
	The optimal CRC length is chosen from the set $\{4,8,12,16,20,24\}$ to minimize the decoding error probability.
	For standard polar codes, we use the classic SCL decoder presented in Section~\ref{sect:polar_decoding}, {\bf not} the new SCL decoder presented in Section~\ref{sect:ST_decoder_DB}.
	For ABS polar codes, we use the SCL decoder presented in Section~\ref{sect:ABS_decoder}.

	Note that in a previous arXiv version and the ISIT version \cite{Li2022ISIT} of this paper, we used a different choice of CRC length. More specifically, for cases (4)--(6) in the above paragraph, we used CRC length $8$ for all choices of code length and code dimension in the previous versions. In contrast, we use the optimal CRC length in this version, and the optimal CRC length varies with the code length and the code dimension.

	From Fig.~\ref{fig:cp1} and Fig.~\ref{fig:cp2} we can see that the performance of ABS polar codes is consistently better than standard polar codes if we set the list size to be $32$ for the CRC-aided SCL decoders of both codes.
	More specifically, for all $12$ choices of $(n,k)$, the improvement of ABS polar codes over standard polar codes ranges from $0.15\dB$ to $0.3\dB$.
	Even if we reduce the list size of ABS polar codes to be $20$ and maintain the list size of standard polar codes to be $32$, ABS polar codes still demonstrate better performance for most choices of parameters, and the improvement over standard polar codes is up to $0.15\dB$ in this case.
	Next let us compare the performance of ABS polar codes and standard polar codes when neither of them uses CRC.
	When there is no CRC, the performance of ABS polar codes with list size $20$ is more or less the same as that of ABS polar codes with list size $32$.
	Again, ABS polar codes consistently outperform standard polar codes for all $12$ choices of $(n,k)$.
	This time the improvement over standard polar codes is up to $1.1\dB$.

	In Table~\ref{tb:time}, we only compare the decoding time of the SCL decoders with CRC length $8$.
	From Table~\ref{tb:time}, we can see that the decoding time of the SCL decoder for ABS polar codes with list size $20$ is more or less the same as the decoding time of the SCL decoder for standard polar codes with list size $32$.
	More precisely, for $8$ out of $12$ choices of $(n,k)$, the SCL decoder for ABS polar codes with list size $20$ runs faster.
	For the other $4$ choices of $(n,k)$, the SCL decoder for standard polar codes with list size $32$ runs faster.
	If we set the list size to be $32$ for both the standard polar codes and the ABS polar codes, then Table~\ref{tb:time} tells us that the decoding time of ABS polar codes is longer than that of standard polar codes by roughly $60\%$.

	In conclusion, when we use list size $32$ for the CRC-aided SCL decoders of both codes, ABS polar codes consistently outperform standard polar codes by $0.15\dB$---$0.3\dB$, but the decoding time of ABS polar decoder is longer than that of standard polar codes by roughly $60\%$.
	If we use list size $20$ for ABS polar codes and maintain the list size to be $32$ for standard polar codes, then the decoding time is more or less the same for these two codes, and ABS polar codes still outperform standard polar codes for most choices of parameters.
	In this case, the improvement over standard polar codes is up to $0.15\dB$.

	As a final remark, the implementations of all the algorithms in this paper are available at the website
	\texttt{https://github.com/PlumJelly/ABS-Polar}

	\section*{Acknowledgement}
	In the implementation of our decoding algorithm, we have learned a lot from the GitHub project \texttt{https://github.com/kshabunov/ecclab} maintained by Kirill Shabunov.
	Shabunov's GitHub project mainly presents the implementation of the Reed-Muller decoder proposed in \cite{Dumer06}.
	Due to the similarity between (ABS) polar codes and Reed-Muller codes, some of the accelerating techniques for Reed-Muller decoders can also be used to speed up (ABS) polar decoders.

	\appendices

	\section{The proof of Lemma~\ref{lemma:recur_ST_DB}} \label{ap:lm1}

	Let $(U_1,\dots,U_n),(X_1,\dots,X_n)$ and $(Y_1,\dots,Y_n)$ be the random vectors defined in Fig.~\ref{fig:bit_channels_polar}.
	Define a new vector $(\widetilde{U}_1,\dots,\widetilde{U}_n)$ as follows:
$$
\widetilde{U}_{2i-1} = U_{2i-1}+U_{2i} \text{~and~} \widetilde{U}_{2i}=U_{2i} \text{~for all~} 1\le i\le n/2.
$$
	Since $\mathbf{G}_n^{\polar}= \mathbf{G}_{n/2}^{\polar} \otimes \mathbf{G}_2^{\polar}$, we have
	\begin{align*}
		 & (X_1,X_3,X_5,\dots,X_{n-1})=(\widetilde{U}_1,\widetilde{U}_3,\widetilde{U}_5,\dots,\widetilde{U}_{n-1}) \mathbf{G}_{n/2}^{\polar}, \\
		 & (X_2,X_4,X_6,\dots,X_{n})=(\widetilde{U}_2,\widetilde{U}_4,\widetilde{U}_6,\dots,\widetilde{U}_{n}) \mathbf{G}_{n/2}^{\polar}.
	\end{align*}
	Therefore, the mapping from $\widetilde{U}_{2i-1},\widetilde{U}_{2i+1}$ to $\widetilde{U}_1,\widetilde{U}_3,\dots,\widetilde{U}_{2i-3}, Y_1,Y_3,\dots,Y_{n-1}$ is $V_i^{(n/2)}$, and the channel mapping from $\widetilde{U}_{2i},\widetilde{U}_{2i+2}$ to $\widetilde{U}_2,\widetilde{U}_4,\dots,\widetilde{U}_{2i-2}, Y_2,Y_4,\dots,Y_{n}$ is also $V_i^{(n/2)}$.
	Moreover, the two random vectors $(\widetilde{U}_1,\widetilde{U}_3,\dots,\widetilde{U}_{n-1},Y_1,Y_3,\dots,Y_{n-1})$ and $(\widetilde{U}_2,\widetilde{U}_4,\dots,\widetilde{U}_{n},Y_2,Y_4,\dots,Y_{n})$ are independent.
	As a consequence,
	\begin{align*}
		                 & V_{2i-1}^{(n)}(y_1,y_2,\dots,y_n,u_1,u_2,\dots,u_{2i-2}|u_{2i-1},u_{2i})                                                                                                                                                               \\
		=                & \bP_{Y_1,Y_2,\dots,Y_n,U_1,U_2,\dots,U_{2i-2}|U_{2i-1},U_{2i}}(y_1,y_2,\dots,y_n,u_1,u_2,\dots,u_{2i-2}|u_{2i-1},u_{2i})                                                                                                               \\
		=                & \frac{1}{4} \sum_{u_{2i+1},u_{2i+2}\in\{0,1\}} \bP_{Y_1,Y_2,\dots,Y_n,U_1,U_2,\dots,U_{2i-2}|U_{2i-1},U_{2i},U_{2i+1},U_{2i+2}}(y_1,y_2,\dots,y_n,                                                                                     \\
		                 & \hspace*{2.8in} u_1,u_2,\dots,u_{2i-2}|u_{2i-1},u_{2i},u_{2i+1},u_{2i+2})                                                                                                                                                              \\
		\overset{(a)}{=} & \frac{1}{4} \sum_{u_{2i+1},u_{2i+2}\in\{0,1\}} \bP_{Y_1,Y_2,\dots,Y_n,\widetilde{U}_1,\widetilde{U}_2,\dots,\widetilde{U}_{2i-2}|\widetilde{U}_{2i-1},\widetilde{U}_{2i},\widetilde{U}_{2i+1},\widetilde{U}_{2i+2}}(y_1,y_2,\dots,y_n, \\
		                 & \hspace*{2.8in} \widetilde{u}_1,\widetilde{u}_2,\dots,\widetilde{u}_{2i-2}|\widetilde{u}_{2i-1},\widetilde{u}_{2i},\widetilde{u}_{2i+1},\widetilde{u}_{2i+2})                                                                          \\
		=                & \frac{1}{4} \sum_{u_{2i+1},u_{2i+2}\in\{0,1\}} \Big( \bP_{Y_1,Y_3,\dots,Y_{n-1},\widetilde{U}_1,\widetilde{U}_3,\dots,\widetilde{U}_{2i-3}|\widetilde{U}_{2i-1},\widetilde{U}_{2i+1}}(y_1,y_3,\dots,y_{n-1},                           \\
		                 & \hspace*{3.2in} \widetilde{u}_1,\widetilde{u}_3,\dots,\widetilde{u}_{2i-3}|\widetilde{u}_{2i-1},\widetilde{u}_{2i+1})                                                                                                                  \\
		                 & \hspace*{1.2in} \cdot \bP_{Y_2,Y_4,\dots,Y_n,\widetilde{U}_2,\widetilde{U}_4,\dots,\widetilde{U}_{2i-2}|\widetilde{U}_{2i},\widetilde{U}_{2i+2}}(y_2,y_4,\dots,y_n,                                                                    \\
		                 & \hspace*{3.2in} \widetilde{u}_2,\widetilde{u}_4,\dots,\widetilde{u}_{2i-2}|\widetilde{u}_{2i},\widetilde{u}_{2i+2})  \Big)                                                                                                             \\
		=                & \frac{1}{4} \sum_{u_{2i+1},u_{2i+2}\in\{0,1\}} \Big( V_i^{(n/2)}(y_1,y_3,\dots,y_{n-1}, \widetilde{u}_1,\widetilde{u}_3,\dots,\widetilde{u}_{2i-3}|\widetilde{u}_{2i-1},\widetilde{u}_{2i+1})                                          \\
		                 & \hspace*{1.2in} \cdot V_i^{(n/2)} (y_2,y_4,\dots,y_n, \widetilde{u}_2,\widetilde{u}_4,\dots,\widetilde{u}_{2i-2}|\widetilde{u}_{2i},\widetilde{u}_{2i+2}) \Big)                                                                        \\
		=                & \frac{1}{4} \sum_{u_{2i+1},u_{2i+2}\in\{0,1\}} \Big( V_i^{(n/2)}(y_1,y_3,\dots,y_{n-1}, \widetilde{u}_1,\widetilde{u}_3,\dots,\widetilde{u}_{2i-3}|u_{2i-1}+u_{2i},u_{2i+1}+u_{2i+2})                                                  \\
		                 & \hspace*{1.2in} \cdot V_i^{(n/2)} (y_2,y_4,\dots,y_n, \widetilde{u}_2,\widetilde{u}_4,\dots,\widetilde{u}_{2i-2}|u_{2i},u_{2i+2}) \Big)                                                                                                \\
		=                & (V_i^{(n/2)})^\triangledown (y_1,y_2,\dots,y_n,\widetilde{u}_1,\widetilde{u}_2,\dots,\widetilde{u}_{2i-2}|u_{2i-1},u_{2i}),
	\end{align*}
	where $\widetilde{u}_1,\widetilde{u}_2,\dots,\widetilde{u}_{2i+2}$ in equality $(a)$ are defined as $\widetilde{u}_{2j-1}=u_{2j-1}+u_{2j}$ and $\widetilde{u}_{2j}=u_{2j}$ for $1\le j\le i+1$.
	Finally, by noting that there is a one-to-one mapping between $y_1,y_2,\dots,y_n,u_1,u_2,\dots,u_{2i-2}$ in the first line and $y_1,y_2,\dots,y_n,\widetilde{u}_1,\widetilde{u}_2,\dots,\widetilde{u}_{2i-2}$ in the last line, we conclude that $V_{2i-1}^{(n)} = (V_i^{(n/2)})^\triangledown$.
	The proofs of $V_{2i}^{(n)} = (V_i^{(n/2)})^\lozenge$ and $V_{2i+1}^{(n)} = (V_i^{(n/2)})^\vartriangle$ are similar.
	We include them here for the sake of completeness.
	\begin{align*}
		                 & V_{2i}^{(n)}(y_1,y_2,\dots,y_n,u_1,u_2,\dots,u_{2i-1}|u_{2i},u_{2i+1})                                                                                                                                                        \\
		=                & \bP_{Y_1,Y_2,\dots,Y_n,U_1,U_2,\dots,U_{2i-1}|U_{2i},U_{2i+1}}(y_1,y_2,\dots,y_n,u_1,u_2,\dots,u_{2i-1}|u_{2i},u_{2i+1})                                                                                                      \\
		=                & \frac{1}{4} \sum_{u_{2i+2}\in\{0,1\}} \bP_{Y_1,Y_2,\dots,Y_n,U_1,U_2,\dots,U_{2i-2}|U_{2i-1},U_{2i},U_{2i+1},U_{2i+2}}(y_1,y_2,\dots,y_n,                                                                                     \\
		                 & \hspace*{2.8in} u_1,u_2,\dots,u_{2i-2}|u_{2i-1},u_{2i},u_{2i+1},u_{2i+2})                                                                                                                                                     \\
		\overset{(a)}{=} & \frac{1}{4} \sum_{u_{2i+2}\in\{0,1\}} \bP_{Y_1,Y_2,\dots,Y_n,\widetilde{U}_1,\widetilde{U}_2,\dots,\widetilde{U}_{2i-2}|\widetilde{U}_{2i-1},\widetilde{U}_{2i},\widetilde{U}_{2i+1},\widetilde{U}_{2i+2}}(y_1,y_2,\dots,y_n, \\
		                 & \hspace*{2.8in} \widetilde{u}_1,\widetilde{u}_2,\dots,\widetilde{u}_{2i-2}|\widetilde{u}_{2i-1},\widetilde{u}_{2i},\widetilde{u}_{2i+1},\widetilde{u}_{2i+2})                                                                 \\
		=                & \frac{1}{4} \sum_{u_{2i+2}\in\{0,1\}} \Big( \bP_{Y_1,Y_3,\dots,Y_{n-1},\widetilde{U}_1,\widetilde{U}_3,\dots,\widetilde{U}_{2i-3}|\widetilde{U}_{2i-1},\widetilde{U}_{2i+1}}(y_1,y_3,\dots,y_{n-1},                           \\
		                 & \hspace*{3.2in} \widetilde{u}_1,\widetilde{u}_3,\dots,\widetilde{u}_{2i-3}|\widetilde{u}_{2i-1},\widetilde{u}_{2i+1})                                                                                                         \\
		                 & \hspace*{1.2in} \cdot \bP_{Y_2,Y_4,\dots,Y_n,\widetilde{U}_2,\widetilde{U}_4,\dots,\widetilde{U}_{2i-2}|\widetilde{U}_{2i},\widetilde{U}_{2i+2}}(y_2,y_4,\dots,y_n,                                                           \\
		                 & \hspace*{3.2in} \widetilde{u}_2,\widetilde{u}_4,\dots,\widetilde{u}_{2i-2}|\widetilde{u}_{2i},\widetilde{u}_{2i+2})  \Big)                                                                                                    \\
		=                & \frac{1}{4} \sum_{u_{2i+2}\in\{0,1\}} \Big( V_i^{(n/2)}(y_1,y_3,\dots,y_{n-1}, \widetilde{u}_1,\widetilde{u}_3,\dots,\widetilde{u}_{2i-3}|\widetilde{u}_{2i-1},\widetilde{u}_{2i+1})                                          \\
		                 & \hspace*{1.2in} \cdot V_i^{(n/2)} (y_2,y_4,\dots,y_n, \widetilde{u}_2,\widetilde{u}_4,\dots,\widetilde{u}_{2i-2}|\widetilde{u}_{2i},\widetilde{u}_{2i+2}) \Big)                                                               \\
		=                & \frac{1}{4} \sum_{u_{2i+2}\in\{0,1\}} \Big( V_i^{(n/2)}(y_1,y_3,\dots,y_{n-1}, \widetilde{u}_1,\widetilde{u}_3,\dots,\widetilde{u}_{2i-3}|u_{2i-1}+u_{2i},u_{2i+1}+u_{2i+2})                                                  \\
		                 & \hspace*{1.2in} \cdot V_i^{(n/2)} (y_2,y_4,\dots,y_n, \widetilde{u}_2,\widetilde{u}_4,\dots,\widetilde{u}_{2i-2}|u_{2i},u_{2i+2}) \Big)                                                                                       \\
		=                & (V_i^{(n/2)})^\lozenge (y_1,y_2,\dots,y_n,\widetilde{u}_1,\widetilde{u}_2,\dots,\widetilde{u}_{2i-2},u_{2i-1}|u_{2i},u_{2i+1}),
	\end{align*}
	where $\widetilde{u}_1,\widetilde{u}_2,\dots,\widetilde{u}_{2i+2}$ in equality $(a)$ are defined the same way as above.
	This proves $V_{2i}^{(n)} = (V_i^{(n/2)})^\lozenge$.
	\begin{align*}
		                 & V_{2i+1}^{(n)}(y_1,y_2,\dots,y_n,u_1,u_2,\dots,u_{2i}|u_{2i+1},u_{2i+2})                                                                                                                             \\
		=                & \bP_{Y_1,Y_2,\dots,Y_n,U_1,U_2,\dots,U_{2i}|U_{2i+1},U_{2i+2}}(y_1,y_2,\dots,y_n,u_1,u_2,\dots,u_{2i}|u_{2i+1},u_{2i+2})                                                                             \\
		=                & \frac{1}{4}  \bP_{Y_1,Y_2,\dots,Y_n,U_1,U_2,\dots,U_{2i-2}|U_{2i-1},U_{2i},U_{2i+1},U_{2i+2}}(y_1,y_2,\dots,y_n,                                                                                     \\
		                 & \hspace*{2.8in} u_1,u_2,\dots,u_{2i-2}|u_{2i-1},u_{2i},u_{2i+1},u_{2i+2})                                                                                                                            \\
		\overset{(a)}{=} & \frac{1}{4}  \bP_{Y_1,Y_2,\dots,Y_n,\widetilde{U}_1,\widetilde{U}_2,\dots,\widetilde{U}_{2i-2}|\widetilde{U}_{2i-1},\widetilde{U}_{2i},\widetilde{U}_{2i+1},\widetilde{U}_{2i+2}}(y_1,y_2,\dots,y_n, \\
		                 & \hspace*{2.8in} \widetilde{u}_1,\widetilde{u}_2,\dots,\widetilde{u}_{2i-2}|\widetilde{u}_{2i-1},\widetilde{u}_{2i},\widetilde{u}_{2i+1},\widetilde{u}_{2i+2})                                        \\
		=                & \frac{1}{4}   \bP_{Y_1,Y_3,\dots,Y_{n-1},\widetilde{U}_1,\widetilde{U}_3,\dots,\widetilde{U}_{2i-3}|\widetilde{U}_{2i-1},\widetilde{U}_{2i+1}}(y_1,y_3,\dots,y_{n-1},                                \\
		                 & \hspace*{3.2in} \widetilde{u}_1,\widetilde{u}_3,\dots,\widetilde{u}_{2i-3}|\widetilde{u}_{2i-1},\widetilde{u}_{2i+1})                                                                                \\
		                 & \hspace*{1.2in} \cdot \bP_{Y_2,Y_4,\dots,Y_n,\widetilde{U}_2,\widetilde{U}_4,\dots,\widetilde{U}_{2i-2}|\widetilde{U}_{2i},\widetilde{U}_{2i+2}}(y_2,y_4,\dots,y_n,                                  \\
		                 & \hspace*{3.2in} \widetilde{u}_2,\widetilde{u}_4,\dots,\widetilde{u}_{2i-2}|\widetilde{u}_{2i},\widetilde{u}_{2i+2})                                                                                  \\
		=                & \frac{1}{4}   V_i^{(n/2)}(y_1,y_3,\dots,y_{n-1}, \widetilde{u}_1,\widetilde{u}_3,\dots,\widetilde{u}_{2i-3}|\widetilde{u}_{2i-1},\widetilde{u}_{2i+1})                                               \\
		                 & \hspace*{1.2in} \cdot V_i^{(n/2)} (y_2,y_4,\dots,y_n, \widetilde{u}_2,\widetilde{u}_4,\dots,\widetilde{u}_{2i-2}|\widetilde{u}_{2i},\widetilde{u}_{2i+2})                                            \\
		=                & \frac{1}{4}   V_i^{(n/2)}(y_1,y_3,\dots,y_{n-1}, \widetilde{u}_1,\widetilde{u}_3,\dots,\widetilde{u}_{2i-3}|u_{2i-1}+u_{2i},u_{2i+1}+u_{2i+2})                                                       \\
		                 & \hspace*{1.2in} \cdot V_i^{(n/2)} (y_2,y_4,\dots,y_n, \widetilde{u}_2,\widetilde{u}_4,\dots,\widetilde{u}_{2i-2}|u_{2i},u_{2i+2})                                                                    \\
		=                & (V_i^{(n/2)})^\vartriangle (y_1,y_2,\dots,y_n,\widetilde{u}_1,\widetilde{u}_2,\dots,\widetilde{u}_{2i-2},u_{2i-1},u_{2i}|u_{2i+1},u_{2i+2}),
	\end{align*}
	where $\widetilde{u}_1,\widetilde{u}_2,\dots,\widetilde{u}_{2i+2}$ in equality $(a)$ are defined the same way as above.
	This proves $V_{2i+1}^{(n)} = (V_i^{(n/2)})^\vartriangle$ and completes the proof of Lemma~\ref{lemma:recur_ST_DB}.
	\qed

	\bibliographystyle{IEEEtran}
	\bibliography{ABS}
\end{document}